\documentclass[acmsmall,screen,authorversion]{acmart}
\pdfoutput=1

\setcopyright{none}

\usepackage{tablefootnote}

\usepackage{amsthm,amsfonts}
\usepackage{mathtools}
\usepackage{mathrsfs}
\usepackage{extarrows}
\usepackage{stmaryrd}
\usepackage{thmtools}

\usepackage{mathpartir}
\usepackage{pftools}
\usepackage{cleveref}

\usepackage{tabularx}

\crefname{equation}{}{}

\usepackage{listings,lstautogobble}
\newcommand{\silqinline}[1]{\lstinline[language=silq,basicstyle=\ttfamily]{#1}}
\newcommand{\qurtsinline}[1]{\lstinline[language=qurts,basicstyle=\ttfamily]{#1}}

\definecolor{light-gray}{gray}{0.95}
\definecolor{light-yellow}{RGB}{255, 255, 200}
\definecolor{light-green}{RGB}{130, 255, 130}
\definecolor{light-blue}{RGB}{130, 130, 255}
\definecolor{light-orange}{RGB}{255, 230, 200}
\definecolor{dark-green}{RGB}{80, 128, 70}
\definecolor{dark-blue}{RGB}{50,50,170}
\lstdefinelanguage{silq}{
  sensitive=true,
  alsoletter={0123456789=},
  morekeywords={[1]  =, if, then, else, def, return },
  morekeywords={[2] const, B, qfree, lifted},
  morekeywords={[3] X, measure, H, f, main,f1,f2,g1,g2, my_identity},
  morestring = [b]",
  stringstyle = \color{purple},
  comment=[l]{//},
}
\lstdefinelanguage{qurts}{
  sensitive=true,
  alsoletter={0123456789=},
  morekeywords={[1]  =, qif, if, then, else, fn, for, in, let, return, mut},
  morekeywords={[2] qbit, bool, &},
  morekeywords={[3] CNot, meas, swap, U, V, H, collect, map, into_iter, phase, oracle, grover, grover_diffusion, non_zero, or, or3, and, and3, nope, main, drop, copy, dup, forget, reinitialize, f, },
  morekeywords={[4] ', a, b, c, static, 0},
  morestring = [b]",
  stringstyle = \color{purple},
  comment=[l]{//},
}
\usepackage[normalem]{ulem}
\makeatletter
\newcommand{\reduuline}[1]{%
  \UL@protected\def\temp@uuline{\leavevmode \bgroup
  \UL@setULdepth
  \ifx\UL@on\UL@onin \advance\ULdepth2.8\p@\fi
  \markoverwith{\textcolor{red}{\lower\ULdepth\hbox
  {\kern-.03em\vbox{\hrule width.2em\kern1\p@\hrule}\kern-.03em}}}%
  \ULon}%
  \temp@uuline{#1}%
}
\makeatother
\usepackage{tikz}
\usetikzlibrary{quantikz2}

\tikzset{vertex/.style={circle,draw=black,thick,inner sep=2pt}}
\tikzset{pebbled/.style={circle,draw=black,fill=light-green,thick,inner sep=2pt}}

\numberwithin{equation}{subsection}

\renewcommand{\phi}{\varphi}

\usepackage{xspace}

\usepackage{wrapfig}

\newcommand{\sdef}{\mathrel{\,::=\,\,}}
\newcommand{\sor}{\mathrel{\,|\,}}
\newcommand{\keyword}[1]{\mathsf{#1}}

\newcommand{\scref}[1]{\textsc{\ref{#1}}}

\newcommand{\lft}[1]{\mathfrak{#1}}

\newcommand{\Own}{{\#}}
\newcommand{\Active}{{\keyword{active}}}

\newcommand{\Fn}{\mathop{\keyword{fn}}}
\newcommand{\Newlft}{\mathop{\keyword{newlft}}}
\newcommand{\Endlft}{\mathop{\keyword{endlft}}}
\newcommand{\If}{\mathop{\keyword{if}}}
\newcommand{\Else}{\mathop{\keyword{else}}}
\newcommand{\Qif}{\mathop{\keyword{qif}}}
\newcommand{\Let}{\mathop{\keyword{let}}}
\newcommand{\Copy}{\mathop{\keyword{copy}}}
\newcommand{\As}{\mathop{\keyword{as}}}
\newcommand{\Bool}{{\keyword{bool}}}
\newcommand{\True}{{\keyword{true}}}
\newcommand{\False}{{\keyword{false}}}

\newcommand{\Meas}{\mathop{\keyword{meas}}}
\newcommand{\Drop}{\mathop{\keyword{drop}}}
\newcommand{\Qbit}{{\keyword{qbit}}}
\newcommand{\Noop}{{\keyword{noop}}}

\newcommand{\COPY}{\mathop{\keyword{Copy}}}
\newcommand{\DROP}{\mathop{\keyword{Drop}}}
\newcommand{\PQ}{\mathop{\keyword{PQ}}}

\newcommand{\defrost}{{\keyword{defrost}}}

\newcommand{\sem}[1]{\left\llbracket{#1}\right\rrbracket}
\newcommand{\ev}[1]{\llparenthesis{#1}\rrparenthesis}
\newcommand{\loc}{\mathrm{loc}}
\newcommand{\Loc}{\mathrm{Loc}}

\newcommand{\Hoare}[4][]{{\left%
\{\ {#2}\ \right\}\ {#3}\ \left\{\ %
\ifx&#1&\else#1{\ \mid\ }\fi%
{#4}\ %
\right\}}}

\renewcommand{\ket}[1]{{\left\vert{#1}\right\rangle}}
\renewcommand{\bra}[1]{{\left\langle{#1}\right\vert}}
\newcommand{\HH}{\mathcal{H}}

\newcommand{\N}{\mathbb{N}}
\newcommand{\C}{\mathbb{C}}
\newcommand{\id}{\mathrm{id}}
\newcommand{\pmap}{\rightharpoonup}

\acmDOI{10.1145/3704842}

\begin{document}

\title[Qurts: Automatic Quantum Uncomputation by Affine Types with Lifetime]{Qurts: Automatic Quantum Uncomputation\\by Affine Types with Lifetime}

\author{Kengo Hirata}
\orcid{0009-0005-4416-2655}
\affiliation{%
  \institution{University of Edinburgh}
  \city{Edinburgh}
  \country{UK}
}
\affiliation{%
  \institution{Kyoto University}
  \city{Kyoto}
  \country{Japan}
}
\email{k.hirata@sms.ed.ac.uk}
\author{Chris Heunen}
\orcid{0000-0001-7393-2640}
\affiliation{%
  \institution{University of Edinburgh}
  \city{Edinburgh}
  \country{UK}
}
\email{chris.heunen@ed.ac.uk}

\begin{abstract}
  Uncomputation is a feature in quantum programming that allows the programmer to discard a value without losing quantum information, and that allows the compiler to reuse resources.
  Whereas quantum information has to be treated linearly by the type system, automatic uncomputation enables the programmer to treat it affinely to some extent.
  Automatic uncomputation requires a substructural type system between linear and affine, a subtlety that has only been captured by existing languages in an ad hoc way.
  We extend the Rust type system to the quantum setting to give a uniform framework for automatic uncomputation called Qurts (pronounced quartz).
  Specifically, we parameterise types by lifetimes, permitting them to be affine during their lifetime, while being restricted to linear use outside their lifetime.
  We also provide two operational semantics: one based on classical simulation, and one that does not depend on any specific uncomputation strategy.
\end{abstract}



\maketitle

\section{Introduction}\label{sec:intro}
Quantum information famously cannot be copied or discarded by the laws of nature.
Consequently, programming languages for quantum computers typically enforce a \emph{linear type discipline} to prevent misuse of quantum information~\cite{selingervaliron:lambda,pratt:linear,wadler:linear,Altenkirch_G05_QML}.
The contraction rule, which allows copying, must be abandoned in order to obey the no-cloning theorem of quantum mechanics~\cite{wootterszurek:nocloning}.
Similarly, the weakening rule, which allows discarding, must be abandoned in order to prevent unintended implicit measurements~\cite{patibraunstein:nodeleting}.
However, the latter restriction on weakening is cumbersome because it mandates that quantum values is explicitly transformed into classical data through measurement before dropping them.

\emph{Uncomputation} is a technique in quantum programming that allows quantum values to be discarded without sacrificing quantum information. Discarding `dirty' qubits has undesired effects on the rest of the computation via entanglement, so they first need to be `cleaned'. A qubit is cleaned by disentangling it from the other qubits and resetting it to its initial state. This \emph{uncomputation} needs to be inserted into the computation at the appropriate place before discarding the qubit, giving a space-time trade-off~\cite{MeuliSRBM_2019_RevPebble,Unqomp_2021,Reqomp_2024,venevetal:modularsynthesis}: it is more space-efficient to uncompute as early as possible, but more time-efficient to uncompute as late as possible.
Moreover, uncomputing a qubit cannot be done too late -- the process of uncomputation may use other qubits that are entangled with the qubit in question, and must be done before the information in those other qubits is lost.
Thus, the timing of uncomputation is crucial.

\emph{Automatic uncomputation} is a technique to insert uncomputation automatically in a program, so that the programmer does not need to worry about it.
The idea was introduced by ReVerC~\cite{AmyRoettelerSvore_CAV2017_ReVerC} (in a non-linear setting), and gained further traction with the development the quantum programming language Silq~\cite{Silq_2020}, which implements \emph{implicit automatic uncomputation} (alongside linear types).
Silq's ad hoc type system can statically verify that a function behaves like a classical Boolean circuit, which is a pivotal prerequisite if it is to be uncomputed.
In that case the compiler automatically inserts the uncomputation without needing explicit annotation from the programmer.
Consequently, Silq programmers can safely discard qubits implicitly; the type system assures that these qubits can be uncomputed, hence preventing the act of discarding from interfering with the rest of the computation.
Silq has a purely functional core, but the full language has imperative features such as assignments, leading to puzzling behaviour, for example, inlining an assignment can cause type errors, or conversely, inlining can fix type errors and make type checks pass.

Automatic uncomputation can be thought of as a system that can type a qubit more loosely than linear type but more strongly than affine type: the qubit can be dropped when we have kept enough information to clean the qubit, but otherwise must be treated linearly.
No existing language expresses such a subtle substructural type system, in which the type of a variable may be treated affinely only during certain periods in the computation.

\emph{Rust}~\cite{rust-web,rust-paper} has a type system that tracks the \emph{lifetime} of information. Explicit \emph{ownership} of variables was introduced there to achieve memory safety and thread safety in the face of shared memory use.
In Rust, while a shared reference, called \emph{immutable reference}, with read-only permission is lent to another, write permission is suspended.
When such references are no longer used, the owner of the resource can write again.
The lifetime feature is used to ensure that the associated resource eventually `returns' to its owner.

\subsubsection*{Qurts}
This article introduces a Rust-like quantum programming language with a proper substructural type system that supports implicit automatic uncomputation, and that can statically verify the correctness of uncomputation.
The type system has explicit lifetime and ownership of variables, like in Rust, and is meant to express values that can be uncomputed during their lifetime: in the type \qurtsinline{#'a}, the lifetime annotation \qurtsinline{'a} represents the interval during which the value can be treated affinely.
Qurts supports a fragment of affine types in which variables can be dropped at any time: for example, a boolean value with type \qurtsinline{#'static bool} can be treated affinely during the whole program since the lifetime is static, that is, the whole program.
Qurts also has a fragment of linear types in which variables can never be dropped: for example, a qubit with type \qurtsinline{#'0 qbit}, whose lifetime is empty, must be treated linearly.

Qurts also has immutable references \qurtsinline{&'a} as in Rust.
A variable of type \qurtsinline{&'a qbit} is treated as a pointer to a qubit
that can only be used for quantum control, that is, the qubit is treated as immutable data with respect to the computational basis.
As in Rust, such immutable reference represents shared memory.
Data of type \qurtsinline{&'a qubit} can be freely copied or discarded,
but doing so only copies the pointer to the qubit, not the qubit itself.
The most important role of immutable references in Qurts is not to represent shared data,
but to represent quantum information that must be kept alive for some time \qurtsinline{'a} to uncompute other qubits which are entangled with the qubits that the references point to.

\subsubsection*{Semantics}
We provide two semantics for Qurts, that we prove equivalent.
Both are small-step operational semantics,
but each handles uncomputation differently.

\emph{Simulation semantics} interprets dropping variables by simply discarding the qubits according to the mathematical definition of uncomputation,
which converts a state
$ \sum_{i \in {\{0,1\}}^{n}} \ket{\phi_i}\ket{i}$
of $m + n$ qubits into the state
$\sum_{i \in {\{0,1\}}^{n}} \ket{\phi_i}\ket{0\dots 0}$,
forcing the last $n$ qubits into the state $\ket{0}$.
This operation can easily be simulated on a classical computer, but cannot be executed as a quantum circuit because it is not reversible and need not preserve probabilities.
However, we prove that the Qurts type system guarantees that this semantics is reversible and preserves probabilities.

\emph{Uncomputation semantics} interprets Qurts programs as \emph{reversible pebble games}~\cite{Bennett-89-rev-pebble-tradeoff}.
These form models of quantum uncomputation~\cite{MeuliSRBM_2019_RevPebble,BhattacharjeeSD19-ISMVL-rev-pebble}: a quantum circuit without uncomputation can be interpreted as a game, and any winning strategy of the game corresponds to a circuit with uncomputation inserted.
We show that Qurts programs can be interpreted as an extended version of pebble games.
This gives a more flexible semantics, that allows us to include more efficient uncomputation strategies than simulation semantics.
As far as we know, Qurts is the first quantum programming language that integrates reversible pebble games for flexible uncomputation strategies.

\subsubsection*{Contributions}
After discussing background in \Cref{sec:prelim}, we present the following contributions.
\begin{itemize}
  \item \Cref{sec:ourlang} introduces the quantum programming language Qurts with many examples. It successfully represents automatic uncomputation by simply annotating a lifetime to the types in Rust.
  \item \Cref{sec:simulation-semantics} presents operational semantics based on classical simulation, which is not generally physically implementable. We prove that the Qurts type system guarantees that the semantics of dropping variables is reversible and preserves probability.
  \item \Cref{sec:unc-semantics} provides another operational semantics, based on reversible pebble games, that is non-deterministic and gives physical implementation. We prove that the simulation and uncomputation semantics are equivalent.
\end{itemize}
\Cref{sec:silq} discusses related work, briefly summarised in \Cref{tab:comparison}.
Finally, \Cref{sec:future} lists directions for further investigation.

\begin{table}[tb]
  \caption{Some quantum programming languages with uncomputation types}
  \label{tab:comparison}
  \small \tabcolsep=0.3cm
    \begin{tabular}{rccccc}
                                      & ReVerC     & ReQWIRE    & Silq        & Qrisp      & Qurts      \\
      \toprule
      Automatic uncomputation         & \checkmark &            & \checkmark  & \checkmark & \checkmark \\
      Linear types                    &            & \checkmark & \checkmark  & \checkmark & \checkmark \\
      Affine types                    &            &            &
      *\tablefootnote{Silq has a \texttt{qfree} annotation for functions, but not for types.}
                                      &            & \checkmark                                         \\
      Immutable borrow                &            &            & \checkmark
      *\tablefootnote{Silq has \texttt{const} types, but they do not follow move semantics.}
                                      &            & \checkmark                                         \\
      Lifetime                        &            &            &             &            & \checkmark \\
      Flexible uncomputation strategy &            &            &             &            & \checkmark \\
      \bottomrule
    \end{tabular}
\end{table}

\section{Background}\label{sec:prelim}
This section summarises the basic notions that this article builds on, namely quantum computation, the Rust type system, and pebble games.
\subsection{Quantum Computation}\label{subsec:background/quantum}
\pdfoutput=1
Let us start with a brief overview of the essential concepts in quantum computing. For more details we refer to textbooks such as~\cite{nielsenchuang:quantum,yanofskymannucci:quantumcomputing}.

\subsubsection*{Qubits}

The basic carrier of quantum information is the \emph{qubit}. The state of a qubit is a unit vector $\varphi = \left(\begin{smallmatrix} x\\y \end{smallmatrix}\right)$ in the Hilbert space $\mathbb{C}^2$, that is, a pair of complex numbers $x$ and $y$ such that $|x|^2+|y|^2=1$. This vector is often written in \emph{ket notation} $\ket{\varphi}$. Two special vectors are the \emph{computational basis} states $\ket{0}=\left(\begin{smallmatrix} 1 \\ 0 \end{smallmatrix}\right)$ and $\ket{1}=\left(\begin{smallmatrix} 0 \\ 1 \end{smallmatrix}\right)$; a general state $\ket{\varphi}$ is in a \emph{superposition} of these two. Two such states that we will often use are $\ket{+}=\tfrac{1}{\sqrt{2}}\ket{0}+\tfrac{1}{\sqrt{2}}\ket{1}$ and $\ket{-}=\tfrac{1}{\sqrt{2}}\ket{0}-\tfrac{1}{\sqrt{2}}\ket{1}$.

\subsubsection*{Entanglement}

In a system with multiple qubits, the state is a unit vector in the \emph{tensor product}. For example, if the first qubit is in state $\ket{0}$, and the second qubit is in state $\ket{1}$, then the state of the compound system is the vector $\ket{0} \otimes \ket{1} \in \mathbb{C}^2 \otimes \mathbb{C}^2$, also written as $\ket{01}$. Not all states of a system with multiple qubits are of this form. For example, the state $\tfrac{1}{\sqrt{2}}\left(\ket{00}+\ket{11}\right)$ is \emph{entangled}: it cannot be written in the form $\ket{\varphi} \otimes \ket{\psi}$.

\subsubsection*{Circuits}

Qubits are acted upon by unitary matrices, also called \emph{quantum gates}. Standard one- and two-qubit gates include
\[
	H = \frac{1}{\sqrt{2}}\begin{pmatrix} 1 & 1 \\ 1 & -1 \end{pmatrix}
	\qquad
	X = \begin{pmatrix} 0 & 1 \\ 1 & 0 \end{pmatrix}
	\qquad
	Z = \begin{pmatrix} 1 & 0 \\ 0 & -1 \end{pmatrix}
	\qquad
	T = \begin{pmatrix} 1 & 0 \\ 0 & \frac{1+i}{\sqrt{2}} \end{pmatrix}
	\qquad
	CX = \left(\begin{smallmatrix} 1 & 0 & 0 & 0 \\ 0 & 1 & 0 & 0 \\ 0 & 0 & 0 & 1 \\ 0& 0 & 1 & 0 \end{smallmatrix}\right)\text.
\]
The gate $H$ is called the \emph{Hadamard gate}.
A quantum computation is described by a \emph{quantum circuit}, which is a product of tensor products of quantum gates. For example, here is the two-qubit quantum circuit notation for the 4-by-4 unitary matrix $CX \circ (Z \otimes T) \circ (H \otimes X)$:
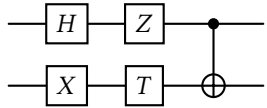
\[\begin{quantikz}[row sep=.3cm]
		&\gate{H}&\gate{Z}&\ctrl{1}& \\
		&\gate{X}&\gate{T}&\targ{}&
	\end{quantikz}\]

\subsubsection*{Measurement}

The outcome produced by a quantum computation is read out by a \emph{measurement}. This is a probabilistic procedure, that, when given the state $x\ket{0}+y\ket{1}$ of a qubit, returns outcome $0$ with probability $|x|^2$, and outcome $1$ with probability $|y|^2$. Measurement collapses superposition, and afterwards the state is $\ket{0}$ or $\ket{1}$, respectively, depending on the outcome.
In quantum circuits, measurements are drawn as $\smash{\begin{tikzpicture}[baseline=-1mm]\node[scale=.7]{\begin{quantikz} &\meter{} \end{quantikz}};\end{tikzpicture}}$.

Quantum circuits can have \emph{auxiliary qubits}, that are not part of the input or output, but only store some temporary information.
Without loss of generality, they are first initialised in state $\ket{0}$.
One way to discard an auxiliary qubit is to measure it at the end of its usage.
When we can guarantee that an auxiliary qubit is in state $\ket{\varphi}$ at the point of measurement, we may also draw $\begin{tikzpicture}[baseline=-1mm, outer sep=0pt, inner sep=0pt] \node (a) {$\bra{\varphi}$}; \draw[-](a) to (-.6,0); \end{tikzpicture}$ instead of $\smash{\begin{tikzpicture}[baseline=-1mm]\node[scale=.7]{\begin{quantikz} &\meter{} \end{quantikz}};\end{tikzpicture}}$. This \emph{bra} notation turns the vector $\ket{\varphi} = \left(\begin{smallmatrix} x \\ y \end{smallmatrix} \right) \in \mathbb{C}^2$ into its adjoint $\bra{\varphi} = \left(\begin{smallmatrix} x^* & y^* \end{smallmatrix}\right) \colon \mathbb{C}^2 \to \mathbb{C}$.

\subsubsection*{Uncomputation}

Measuring one qubit in a multiple-qubit system can affect the whole state. Consider the following left quantum circuit that operates on a single qubit, using one auxiliary qubit.
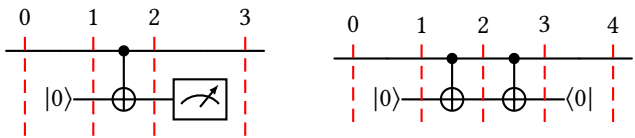
\[\begin{tikzcd}[row sep=0.3cm]
		\slice{0}&\slice{1}&\ctrl{1}\slice{2}&\slice{3}& \\
		&\ket{0}\setwiretype{n}&\targ{}\setwiretype{q}&\meter{}&\setwiretype{n}
	\end{tikzcd}
	\qquad
	\begin{tikzcd}[row sep=.3cm]
		\slice{0}&\slice{1}&\ctrl{1}\slice{2}&\ctrl{1}\slice{3}&\slice{4}& \\
		&\ket{0}\setwiretype{n}&\targ{}\setwiretype{q}&\targ{}& \bra{0}
		&\setwiretype{n}
	\end{tikzcd}\]
The top qubit is only used as a control, and hence one might expect the output to equal the input. However, the mere fact that the auxiliary qubit is initialised, interacts with the data qubit, and is discarded alters this expected behaviour. Suppose that the input at step 0 is $\ket{+}=\tfrac{1}{\sqrt{2}}(\ket{0}+\ket{1})$. This evolves to $\tfrac{1}{\sqrt{2}}(\ket{00}+\ket{10})$ at step 1, and to $\tfrac{1}{\sqrt{2}}(\ket{00}+\ket{11})$ at step 2. But after the measurement, at step 3, the output is $\ket{0}$ with probability $\tfrac{1}{2}$, and $\ket{1}$ with probability $\tfrac{1}{2}$.
This behaviour occured because the auxiliary qubit was discarded while it was in a \emph{dirty} state, that is, it was not guaranteed to be in the same state it was initialised in. To prevent this behaviour, the auxiliary qubit needs to be cleaned before discarding it. We can do this by inserting an \emph{uncomputation} before discarding as in the right circuit above.
The inserted circuit between steps 2 and 3 uncomputes everything that happened to the auxiliary qubit, guaranteeing that it is measured in the same state it was initialised in, and therefore has no effect on the rest of the computation.
Indeed, if the input state at step 0 is $\ket{+}$, then at step 3 the state is now $\tfrac{1}{\sqrt{2}}(\ket{00}+\ket{10})$, and so the output at step 4 is again $\ket{+}$, as expected.
The general definition is as follows.

\begin{definition}
	Given a quantum state $\ket{\phi} = \sum_{i \in \{0,1\}^n} \ket{\phi_i}\ket{i}$ of $m + n$ qubits, an \emph{uncomputation} on the last $n$ qubits is an operation that maps $\ket{\phi}$ to $\sum_{i\in\{0,1\}^n} \ket{\phi_i}\ket{0\dots 0}$.
\end{definition}
This definition differs slightly from correctness of uncomputation as defined in~\cite{Unqomp_2021,Reqomp_2024}.
There, correctness is defined as a predicate on a pair of operations $(G, \mathcal{G})$
stating that $\mathcal{G}$ is an operation that does the same as $G$ but leaves auxiliary qubits uncomputed.
Our definition above focuses on the state of the qubits after the operation rather than the operation itself.
This may seem less general because uncomputation is defined only as a post-processing step, but this is not a problem because this paper allows uncomputation at any point of the computation.

The following inequality shows that the norm of uncomputed state is always larger than or equal to that of the original state, with equality holds if and only if $\{\ket{\phi_i}\}_i$ are orthogonal.
\[
	{\Big|\ket{\phi}\Big|}^2 = {\Big|\sum_{i \in \{0,1\}^n} \ket{\phi_i}\ket{i}\Big|}^2
	= \sum_{i \in \{0,1\}^n} {\Big|\ket{\phi_i}\Big|}^2
	\leq {\Big|\sum_{i \in \{0,1\}^n} \ket{\phi_i}\Big|}^2
	= {\Big|\sum_{i \in \{0,1\}^n} \ket{\phi_i}\ket{0\dots0} \Big|}^2
\]
This means that the total probability of measurement success before uncomputation is no less than after.
As the total probability cannot be increased by any physical operation, uncomputation is physically realisable only if the set of states $\{\ket{\phi_i}\}_i$ are orthogonal, which is not always the case.

Uncomputation can also be used for space optimisation, as it allows auxiliary qubits to be `deallocated' in order to save space.
However, uncomputation may cost time because an uncomputed value might need to be recomputed later.
Vice versa, leaving a value uncomputed for longer occupies memory space, but may save computation time when the value is reused later.
We will see this trade-off between time and space later in~\Cref{subsec:background/pebble}.

\subsubsection*{Quantum Control}
The quantum state $\ket0$ is often identified with the boolean value \texttt{false}, and $\ket1$ with \texttt{true}.
Quantum circuits can implement boolean operations this way.
For example, the $X$ gate implements the \texttt{not} operation.
When a unitary function $U$ maps a state $\ket{x} \in \C^2\otimes\cdots\otimes\C^2$ to $\ket{f(x)}$ for a bijective function $f\colon\{\False,\True\}^n\to\{\False,\True\}^n$, we call $U$ the \emph{lift} of $f$.
Together with a function that introduces a new qubit $\ket0:\C\to\C^2;1\mapsto \ket0$, this definition of lift can be extended to any injective function $f\colon\{\False,\True\}^n\to\{\False,\True\}^m$.

There are quantum analogues of conditional operations in classical computing, called \emph{quantum control}~\cite{ying,sabryvalironvizzotto:control}.
Some programming languages allow quantum control flow similar to classical control flow~\cite{Altenkirch_G05_QML,Silq_2020}.
For example, we denote the quantum analogue of \texttt{if} statements by the command \texttt{qif x \{ U(y) \} else \{ V(y) \}}.
It can be implemented with \emph{controlled gates}.
Given any $n$-qubit gate $U$, the controlled-$U$ operation is an $n+1$-qubit unitary that maps $\ket0\ket y\in \C^2\otimes \C^{2^n}$ to $\ket0\ket y$ and $\ket1\ket y$ to $\ket1(U\ket y)$.
For example, the Controlled-$X$ gate is the matrix $CX$ from above.
This controlled-$U$ operation represents the \text{then} branch of the \texttt{qif} statement.
To implement the \text{else} branch, we need negatively-controlled-$V$ operations, which conjugates Controlled-$V$ with $X$ gates: \texttt{X(x); qif x \{ V(y) \}; X(x);}.
Since the \emph{control} qubit $x$ can be in superposition, both branches are executed in superposition, which is a key difference from classical control flow, where only one branch is executed.
Gates can similarly be controlled on multiple qubits; a \emph{single target gate} is a $n$-qubit quantum gate that can be represented as the following map
for a boolean function $i\colon \{x_1,\dots,x_{n}\} \to \{0,1\}$ where $\neg^1 x = \neg x$ and $\neg^0 x = x$.
\[
	\ket{x_1 \dots x_{n}}\ket{y} \mapsto
	\ket{x_1 \dots x_{n}} \ket{y \oplus (\neg^{i_{x_1}}x_1\wedge \cdots \wedge \neg^{i_{x_n}}x_{n})}
\]
For example, the gate $\ket{x_1 x_2}\ket{y} \mapsto \ket{x_1 x_2}\ket{y \oplus (x_1 \wedge x_2)}$ is called the \emph{Toffoli gate}.
\subsection{Rust: Lifetime, Borrowing, and Ownership}\label{subsec:background/rust}
The Rust programming language~\cite{rust-web,rust-paper} enforces memory safety via a type system, and has a high level of performance optimisation. We briefly overview aspects we adapt.

\paragraph{Ownership: Move versus Copy}

Instead of garbage collection, Rust uses a system of \emph{ownership}~\cite{rust-ownership}.
In this system, each value has a single \emph{owner} at a time,
and as soon as the owner goes out of scope, the value is dropped.
This ownership is \emph{moved} when the value is assigned to another variable, or passed to a function.
After the ownership is moved to another variable, the original variable loses access to the value.
For example, in the code snippet below, the variable \texttt{s1} owns a string
\texttt{"hello"} after running line 1.
When \texttt{s1} is assigned to \texttt{s2} in line 2,
the ownership of the string \emph{moves} from \texttt{s1} to \texttt{s2}.
Since Rust assignments are not copies but moves,
the value of \texttt{s1} is no longer accessible after line 2.
Thus, the \texttt{print!} call with \texttt{s1} in line 3 will not compile.
\begin{center}
  \begin{minipage}{0.6\textwidth}
    \begin{lstlisting}[language=qurts]
      let s1 = String::from("hello");
      let s2 = s1;
      print!("{$\texttt{s1}$},world!");
    \end{lstlisting}
  \end{minipage}
\end{center}
Ownership is underlain by \emph{affine types}, substructural types that prohibit copying of values.
This type discipline enables Rust to detect when a value becomes orphaned and free its associated memory.
Moreover, by not regarding assignments as copies, Rust can avoid unnecessary deep copying of values,
which can be expensive in terms of time and memory.

\subsubsection*{References}

Rust's ownership system prohibits a value from having multiple owners at the same time.
However, sometimes it is useful to share a value using a pointer.
Rust allows this by using \emph{references}.~\cite{rust-reference}.
A reference is essentially a pointer to a value, and does not take ownership of the value.
Instead, it \emph{borrows} the value from the owner.
There are two types of references in Rust: \emph{immutable references} and \emph{mutable references}.
Immutable references allow the value to be read but not modified, while mutable references allow the value to be modified.
When a variable is immutably borrowed, there can be multiple references to the value, but the value cannot be modified through the reference.
On the other hand, when a variable is mutably borrowed, there can be only one reference to the value, and the value can only be read or updated through the reference.
After all references go out of scope, the owner regains full access to the value.
In the code snippet below left, since the value which \texttt{x} owned in line 2 moved to \texttt{y} in line 3, the value of \texttt{x} is no longer accessible in line 4.
However, in the code snippet on the right-hand side, the value of \texttt{x} is immutably borrowed in line 3, so the value that \texttt{x} owns is still readable in line 4.

\begin{center}
  \begin{minipage}{0.38\textwidth}
    \begin{lstlisting}[language=qurts]
      fn main() {
        let x = [1,2,3];
        let y = f(x); // value of x is moved
        print!("{@\reduuline{\texttt{x}}@},{$\texttt{y}$}"); // so x lost access
      }
      fn f<T>(x : T) -> T {
        x
      }
    \end{lstlisting}
  \end{minipage}
  \hspace*{0.07\textwidth}
  \begin{minipage}{0.36\textwidth}
    \begin{lstlisting}[language=qurts]
      fn main() {
        let x = [1,2,3];
        let y = f(&x); // x is borrowed
        print!("{$\texttt{x}$},{$\texttt{y}$}"); // x can be read
      }
      fn f<T>(x : T) -> T {
        x
      }
    \end{lstlisting}
  \end{minipage}
\end{center}

\subsubsection*{Lifetime}

To know when a reference is no longer needed,
and the owner can regain full access to the value,
Rust uses the concept of \emph{lifetime}~\cite{rust-lifetime}.
Every reference in Rust has a lifetime: the scope during which that reference is valid.
This prevents dangling pointers or race conditions by ensuring that when someone has write access to a value, no one else has any access to the value.
For example, the following right Rust code will not compile:
\begin{wrapfigure}{r}{0.22\textwidth}
  \begin{center}
    \quad\begin{minipage}{0.12\textwidth}
      \begin{lstlisting}[language=qurts]
        let r;
        {
            let x = 5;
            r = &x;
        }
        print!(r);
      \end{lstlisting}
    \end{minipage}
  \end{center}
\end{wrapfigure}
The variable \qurtsinline{r} has a lifetime of lines 1--6, and \qurtsinline{x} has a lifetime of lines 3--5. Storing a reference to \qurtsinline{x} in \qurtsinline{r} in line 4 is not allowed, because the lifetime of \qurtsinline{x} ends before that of \qurtsinline{r}. Most of the time lifetimes are implicit and inferred, just like most of the time types are inferred.
However, when you want to define a function whose outputs include references, you need to explicitly annotate the lifetimes.
This is because, if a reference is returned as a function output, it becomes difficult to track the aliveness of the reference statically.
For example, \qurtsinline{fn f<'a>(x: \&'a bool) -> \&'a bool} declares that the function \qurtsinline{f} takes a reference to a boolean value with (external) lifetime \qurtsinline{'a}, and returns a reference to a boolean value with the same lifetime.

\subsubsection*{Similarity to Quantum Uncomputation}

Uncomputation and garbage collection are similar in that both manage memory, and both trade-off time against space.
Automatic uncomputation and Rust share the feature that the type system gives the compiler hints, and the compiler manages memory, so the programmer does not have to think about it.

Quantum computation has adopted linear type systems because generic quantum states cannot be copied or discarded.
There is a quantum analogue of copying a quantum state
defined by $\C^2\ni\ket{i} \mapsto \ket{ii}\in \C^2\otimes \C^2$,
which respects the computational basis $\langle\ket0,\ket1\rangle$.
However, implicitly performing this copy-like operation is problematic in quantum computation because it entangles the original state with the copy, which interferes with the computation.
Also, this operation uses a two-qubit gate, which is expensive in terms of error rates.

\subsection{Reversible Pebble Games}\label{subsec:background/pebble}
Reversible pebble games were first applied to quantum uncomputation in~\cite{MeuliSRBM_2019_RevPebble}.
To understand them, consider the example program
`\lstinline[basicstyle=\ttfamily\small]{and(nand(not(and(a,b)),not(c)),d)}'. 
It is equivalent to the following pseudocode computing $w$ by uncomputing temporary $x$, $y$, $z$, $v$:
\begin{center}
  \lstinline[basicstyle=\ttfamily\small]{x := And(a,b); y := Not(x); z := Not(c); v := Nand(y,z); w := And(v,d); uncomp(x,y,z,v)}
\end{center}
From the program, we construct the \emph{circuit graph} at the top left
in~\Cref{fig:cgraph/example-strategy}: a directed graph whose vertices correspond to gates, and whose edges represent dependencies between gates.
For example, the edge from $x$ to $y$ means that $y$ uses the result of $x$ as its input.

A reversible pebble game is played on a directed acyclic graph, by placing tokens called pebbles on vertices, according to the following rules:
\begin{itemize}
  \item You can place a pebble on a vertex when all its predecessors are already pebbled.
  \item You may remove a pebble from a vertex when all its predecessors are pebbled.
\end{itemize}
Intuitively, the vertices represent a value during program execution, and a pebble indicates that this value is stored in the memory. 
The aim of the game is to get pebbles on all vertices except those that are uncomputed.
Typically, this means you have to get a single pebble on the last vertex, with no other vertices being pebbled. 
In our example, the goal is to get a pebble only on $w$.

This case allows two pebbling strategies.
The first, naive, strategy puts pebbles on $x$, $y$, $z$, $v$, and $w$ in this order, and then removes all pebbles but the one on $w$ again in the reverse order.
This strategy uses $5$ pebbles and takes $9$ execution steps.
Another strategy is illustrated in~\Cref{fig:cgraph/example-strategy}.
The first $4$ steps are the same as the naive strategy, but then it removes the pebble on $x$.
This strategy uses only $4$ pebbles at the same time, but takes $11$ time steps.

Each pebbling strategy corresponds to a quantum circuit. A pebble represents a qubit, and a vertex represents a quantum state. When a pebble is on a vertex, it indicates that the qubit is in the state of that vertex. Placing a pebble on a vertex corresponds to storing the qubit's state in memory, while removing a pebble corresponds to uncomputing the qubit. Thus, the number of steps in the pebbling strategy corresponds to the number of gates in the quantum circuit.
For example, the step putting a pebble on $y$ in the configuration
\smash{\begin{tikzpicture}[baseline={(y.base)}]
  \node[pebbled] (x) at (0,0) {$x$};
  \node[vertex] (y) at (0.8,0) {$y$};
  \draw[->] (x) -- (y);
\end{tikzpicture}}
corresponds to computing $\ket{x}\ket0\mapsto\ket{x}\ket{\neg x}$,
since $y$ is the negation of $x$ in the example program.
As seen in the example, the pebbling strategy affects the number of qubits and gates of the eventual quantum circuit.
This gives a trade-off between time and space, whose optimisation is PSPACE-complete~\cite{ChanLNV_15_RevPebble_PSPACE}.

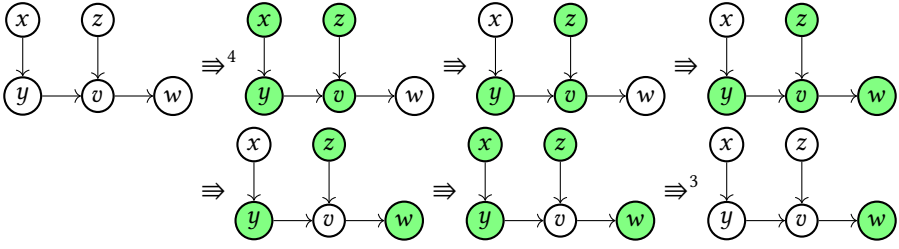
\begin{figure}[t]
  \begin{align*}
    \begin{tikzpicture}[baseline={(0,0.3)}]
      \node[vertex](x) at (1,1) {$x$};
      \node[vertex](y) at (1,0) {$y$};
      \node[vertex](z) at (2,1) {$z$};
      \node[vertex](v) at (2,0) {$v$};
      \node[vertex](w) at (3,0) {$w$};
      \draw[->] (x) -- (y);
      \draw[->] (y) -- (v);
      \draw[->] (z) -- (v);
      \draw[->] (v) -- (w);
    \end{tikzpicture}
    \Rrightarrow^4
    \begin{tikzpicture}[baseline={(0,0.3)}]
      \node[pebbled](x) at (1,1) {$x$};
      \node[pebbled](y) at (1,0) {$y$};
      \node[pebbled](z) at (2,1) {$z$};
      \node[pebbled](v) at (2,0) {$v$};
      \node[vertex](w) at (3,0) {$w$};
      \draw[->] (x) -- (y);
      \draw[->] (y) -- (v);
      \draw[->] (z) -- (v);
      \draw[->] (v) -- (w);
    \end{tikzpicture}
    \Rrightarrow
    \begin{tikzpicture}[baseline={(0,0.3)}]
      \node[vertex](x) at (1,1) {$x$};
      \node[pebbled](y) at (1,0) {$y$};
      \node[pebbled](z) at (2,1) {$z$};
      \node[pebbled](v) at (2,0) {$v$};
      \node[vertex](w) at (3,0) {$w$};
      \draw[->] (x) -- (y);
      \draw[->] (y) -- (v);
      \draw[->] (z) -- (v);
      \draw[->] (v) -- (w);
    \end{tikzpicture}
    \Rrightarrow
    \begin{tikzpicture}[baseline={(0,0.3)}]
      \node[vertex](x) at (1,1) {$x$};
      \node[pebbled](y) at (1,0) {$y$};
      \node[pebbled](z) at (2,1) {$z$};
      \node[pebbled](v) at (2,0) {$v$};
      \node[pebbled](w) at (3,0) {$w$};
      \draw[->] (x) -- (y);
      \draw[->] (y) -- (v);
      \draw[->] (z) -- (v);
      \draw[->] (v) -- (w);
    \end{tikzpicture}
    \\
    \Rrightarrow
    \begin{tikzpicture}[baseline={(0,0.3)}]
      \node[vertex](x) at (1,1) {$x$};
      \node[pebbled](y) at (1,0) {$y$};
      \node[pebbled](z) at (2,1) {$z$};
      \node[vertex](v) at (2,0) {$v$};
      \node[pebbled](w) at (3,0) {$w$};
      \draw[->] (x) -- (y);
      \draw[->] (y) -- (v);
      \draw[->] (z) -- (v);
      \draw[->] (v) -- (w);
    \end{tikzpicture}
    \Rrightarrow
    \begin{tikzpicture}[baseline={(0,0.3)}]
      \node[pebbled](x) at (1,1) {$x$};
      \node[pebbled](y) at (1,0) {$y$};
      \node[pebbled](z) at (2,1) {$z$};
      \node[vertex](v) at (2,0) {$v$};
      \node[pebbled](w) at (3,0) {$w$};
      \draw[->] (x) -- (y);
      \draw[->] (y) -- (v);
      \draw[->] (z) -- (v);
      \draw[->] (v) -- (w);
    \end{tikzpicture}
    \Rrightarrow^3
    \begin{tikzpicture}[baseline={(0,0.3)}]
      \node[vertex](x) at (1,1) {$x$};
      \node[vertex](y) at (1,0) {$y$};
      \node[vertex](z) at (2,1) {$z$};
      \node[vertex](v) at (2,0) {$v$};
      \node[pebbled](w) at (3,0) {$w$};
      \draw[->] (x) -- (y);
      \draw[->] (y) -- (v);
      \draw[->] (z) -- (v);
      \draw[->] (v) -- (w);
    \end{tikzpicture}
  \end{align*}
  \Description{A strategy of pebbling using minimum space}
  \caption{An example strategy with a minimum number of pebbles}
  \label{fig:cgraph/example-strategy}
\end{figure}

\section{Qurts}\label{sec:ourlang}
\pdfoutput=1

This section introduces Qurts\footnote{We follow the precedent of naming programming languages after precious materials, such as Ruby. Just like Perl abbreviates `pearl', Qurts abbreviates 'quartz'. Additionally, Qurts is an anagram of `Q' and `Rust'.}\footnote{We will use some standard features from Rust in this article that are not formally defined in \Cref{subsec:lang-def} below, but that are part of an implementation of Qurts in development.}. \Cref{subsec:lang-example} first overviews its features and capabilities by way of example. \Cref{subsec:lang-def} then more formally introduces the syntax and type system.

\subsection{Qurts by Example}\label{subsec:lang-example}

Every variable in Qurts has a type, the primitive types being \qurtsinline{bool} and \qurtsinline{qbit}. Declaring a variable with explicit type requires stating ownership, as in Rust: either the variable owns its value, as in \qurtsinline{x:#qbit}, or the variable immutably borrows its value, as in \qurtsinline{x:&qbit}, or the variable mutably borrows its value, as in \qurtsinline{x:&mut qbit}.
Functions like \qurtsinline{H} are applied to variables like \qurtsinline{x:qbit} using the command \qurtsinline{x.H();} or the expression \qurtsinline{H(x)}.

\paragraph{Grover's algorithm}
We begin with a real-world example of Qurts: Grover's algorithm~\cite{grover:search}.
The following very succinct code snippets implement the inner loop of the algorithm, illustrating Qurts' core functionality of automatic uncomputation.
\begin{center}
  \begin{minipage}{.35\textwidth}
    \begin{lstlisting}[language=qurts,numbers=none]
      qif &oracle(&xl) { phase($\pi$) }
    \end{lstlisting}
  \end{minipage}
  \qquad
  \begin{minipage}{.35\textwidth}
    \begin{lstlisting}[language=qurts,numbers=none]
      qif &non_zero(&x1) { phase($\pi$) }
    \end{lstlisting}
  \end{minipage}
\end{center}
The fragments above are syntactic sugar for their equivalents below with explicit uncomputation.
\begin{center}
  \begin{minipage}{.35\textwidth}
    \begin{lstlisting}[language=qurts,numbers=none]
      let r = &xl;
      let tmp: #qbit = oracle(r);
      qif &tmp { phase($\pi$) } else { noop };
      drop tmp;
    \end{lstlisting}
  \end{minipage}
  \qquad
  \begin{minipage}{.35\textwidth}
    \begin{lstlisting}[language=qurts,numbers=none]
      let r = &xl;
      let tmp: #qbit = non_zero(r);
      qif &tmp { phase($\pi$) } else { noop };
      drop tmp;
    \end{lstlisting}
  \end{minipage}
\end{center}
The function call {\color{purple}\texttt{phase}}\texttt{($\pi$)} represents the map $-1\colon\C\to\C$
that multiplies the quantum state by -1.
Calling it within the \qurtsinline{qif} causes the Z gate to be applied to the qubit owned by \qurtsinline{tmp}.
The \qurtsinline{drop} statement at the end uncomputes the temporary qubit.
This uncomputation is possible because \qurtsinline{tmp} is obtained
by the function call \qurtsinline{oracle(r)} or \qurtsinline{non_zero(r)},
and since the value which \qurtsinline{r} refers to is not used after the function call,
we can apply the reverse operation of the function call to \qurtsinline{r} and \qurtsinline{tmp}
to uncompute the temporary qubit. The full algorithm is listed in \Cref{fig:code/grover} (for 3-qubit input to keep it explicit and not import Rust syntax for arrays).
See \Cref{ax-sec:grover} for more details.

From the type system perspective, the functions \qurtsinline{oracle} and \qurtsinline{non_zero} have the same signature \qurtsinline{fn(&qbit, &qbit, &qbit) -> #qbit}.
This type means that the function takes immutable references to qubits and returns a qubit.

\begin{figure}
  \begin{minipage}{0.4\textwidth}
    \begin{lstlisting}[language=qurts,numbers=none]
      fn grover(
        x &mut qbit,
        y:&mut qbit, 
        z:&mut qbit
      ) {
        x.H(); y.H(); z.H();  
        for _ in $0$..$2$ {
          qif &oracle(&x, &y, &z) { phase($\pi$) }
          grover_diffusion(x, y, z)
        }
      }
    \end{lstlisting}
  \end{minipage}
  \qquad
  \begin{minipage}{0.4\textwidth}
    \begin{lstlisting}[language=qurts,numbers=none]
      fn grover_diffusion(
        x:&mut qbit,
        y:&mut qbit, 
        z:&mut qbit
      ) {
        x.H(); y.H(); z.H();  
        qif &non_zero(&x, &y, &z) { phase($\pi$) }
        x.H(); y.H(); z.H();  
      }
    \end{lstlisting}
  \end{minipage}
  \caption{Grover's algorithm for 3 qubits in Qurts.}
  \label{fig:code/grover}
\end{figure}

\paragraph{Implementing Boolean circuits}
Next, we implement in Qurts boolean circuits, whose return qubit can be uncomputed, such as the function \qurtsinline{non_zero} in the previous example.
The Qurts code on the left below implements a binary \textsc{and} on \qurtsinline{qbit}s.
\begin{center}
  \begin{minipage}{0.5\textwidth}
    \begin{lstlisting}[language=qurts,numbers=none]
      fn and<'a>(x:&'a qbit, y:&'a qbit) -> #'a qbit {
        qif x {
          qif y { $\ket1$ } else { $\ket0$ }
        } else {
          $\ket0$
        }
      }
    \end{lstlisting}
  \end{minipage}
  \qquad
  \begin{minipage}{0.3\textwidth}
    \begin{lstlisting}[language=qurts,numbers=none]
      fn and3<'a>(
        x:&'a qbit,
        y:&'a qbit,
        z:&'a qbit
      ) -> #'a qbit {
        and(&and(x, y), z)
      }
    \end{lstlisting}
  \end{minipage}
\end{center}
Note that after calling this function \qurtsinline{and}, the caller can use three qubits: the two input qubits and the result qubit.
The input variables \qurtsinline{x} and \qurtsinline{y} are immutable references, so they are not consumed by the function \qurtsinline{and}. 
On the other hand, the result is of type \qurtsinline{#qbit} so will be moved to the caller.

The code on the right above uses the 2-qbit \textsc{and} to build the 3-qbit version \qurtsinline{and3}.
Notice how the first argument of the outer call to the 2-qbit \textsc{and} has to be converted from \qurtsinline{and(x,y):#qbit} to \qurtsinline{&and(x,y):&qbit}.
The function \qurtsinline{and} gets compiled to the circuit on the left below.
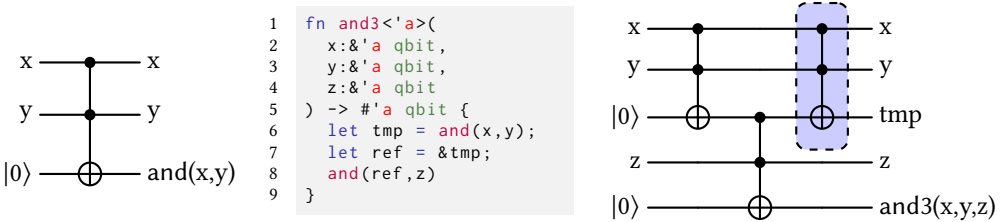
\begin{center}
  \begin{tikzcd}
    \llap{x}\;&\ctrl{1}&\;\rlap{x} \\
    \llap{y}\;&\ctrl{1}&\;\rlap{y} \\
    \llap{$\ket{0}$}\;\setwiretype{n}&\targ{}\setwiretype{q}&\;\rlap{\textsf{and}(x,y)}
  \end{tikzcd}
  \hspace*{2cm}
  \begin{minipage}{0.2\textwidth}
    \begin{lstlisting}[language=qurts]
      fn and3<'a>(
        x:&'a qbit,
        y:&'a qbit,
        z:&'a qbit
      ) -> #'a qbit {
        let tmp = and(x,y);
        let ref = &tmp;
        and(ref,z)
      }
    \end{lstlisting}
  \end{minipage}
  \qquad\
  \begin{tikzcd}[row sep=.35cm]
    \llap{x}\;&\ctrl{1}&&\ctrl{1}
    \gategroup[3,steps=1,style={dashed,rounded corners,fill=blue!20,inner xsep=2pt}, background]{}
    &\;\rlap{x} \\
    \llap{y}\;&\ctrl{1}&&\ctrl{1}&\;\rlap{y} \\
    \llap{$\ket{0}$}\;&\targ{}&\ctrl{1}&\targ{}&\;\rlap{tmp} \\
    \llap{z}\;&&\ctrl{1}&&\;\rlap{z} \\
    \llap{$\ket{0}$}\;&&\targ{}&&\;\rlap{\textsf{and3}(x,y,z)}
  \end{tikzcd}
\end{center}
Let us unfold the original function \qurtsinline{and3} to make the auxiliary qubit \qurtsinline{tmp} explicit.
This gets compiled to the first two Toffoli gates in the circuit on the right above
-- before any uncomputation: the auxiliary qubit \texttt{tmp} still needs to be uncomputed
before it can be discarded.
One way to uncompute the auxiliary qubit is to apply the third Toffoli gate, marked in the circuit.
If the input was $\ket{ij0k0}$, the state after the first two Toffoli gates is
$\ket{ij(i\wedge j)k(i\wedge j\wedge k)}$, and after the uncomputation it becomes $\ket{ij0k(i\wedge j\wedge k)}$.
Similarly, we can build a 3-(qu)bit \textsc{OR} function to implement \qurtsinline{non_zero}.

\paragraph{Lifetime annotations}
We did not go into detail above, but \qurtsinline{and} and \qurtsinline{and3} required lifetime annotations \qurtsinline{'a}.
As in Rust, all references in Qurts have lifetimes.
In addition, Qurts also has lifetime annotation of linearity \qurtsinline{#'a T}.
This type contains values that can only be dropped/uncomputed by the end of the lifetime \qurtsinline{'a}.
For example, the signature \qurtsinline{fn and<'a>(x:&'a qbit, y:&'a qbit) -> #'a qbit} means that the function \qurtsinline{and} takes two immutable references to qubits and returns a qubit which can only be uncomputed as long as the references are available.

The key typing rule in Qurts is the lifetime annotation for the return value of \qurtsinline{qif} statements.
\begin{center}
  \begin{minipage}{0.6\textwidth}
    \begin{lstlisting}[language=qurts,numbers=none]
      // r: &'a qbit
      let x: #'a T = qif r { M } else { N }; // M, N: T
    \end{lstlisting}
  \end{minipage}
\end{center}
This illustrates that, if the control qubit \qurtsinline{r} has type \qurtsinline{&'a qbit},
the return value of the \qurtsinline{qif} statement has type \qurtsinline{#'a T} with the same lifetime \qurtsinline{'a}, where \qurtsinline{T} is the type of the branches.
In other words, it restricts the return value to be uncomputed only until the end of \qurtsinline{'a}.
The reason for this restriction is that the control is required to uncompute the return value, but the control is only available until the lifetime \qurtsinline{'a} ends.

\paragraph{Lifted functions}
In Qurts, you can \emph{lift} any classical injection $c\colon \{\True,\False\}^n \rightarrow \{\True,\False\}^m$ to a primitive function \texttt{[}$c$\qurtsinline{]: #'a [qbit;n] -> #'a [qbit;m]} which takes a list of qubits and returns a list of qubits.
For example, \texttt{[not]} is a 1-qubit lifted function which represents the $X$-gate, \texttt{[cnot]} is a 2-qubit lifted function which represents the controlled-$X$ gate, and \texttt{[swap]} is a 2-qubit lifted function which represents the \textsc{swap} gate.
These can be used to make affine qubits from other affine qubits without borrowing them.
The lifetime annotation indicates that the output is uncomputable for the same period of time that the input was uncomputable.

\paragraph{Uncomputability via types}
The following code explains how lifetimes assigned to linearity can detect cases
where uncomputation is not trivially possible through type checking.
\begin{center}
  \hspace*{5mm}
  \begin{minipage}{0.4\textwidth}
    \begin{lstlisting}[language=qurts]
      {
        // mut p: #qbit
        let q = $\ket0$;
        let q = qif &p { [not](q) } else { q };


        p.H();
        return p;
      } // q is dropped before here
    \end{lstlisting}
  \end{minipage}
  \hspace*{1cm}
  \begin{minipage}{0.38\textwidth}
    \begin{lstlisting}[language=qurts]
      {
        // mut p: #qbit
        let q = $\ket0$;
        let q: #'a qbit =
          qif &'a p { [not](q) } else { q };
        drop q; endlft 'a;
        p.H();
        return p;
      }
    \end{lstlisting}
  \end{minipage}
\end{center}
Both snippets are valid Qurts, and are equivalent; the code on the right has merely inserted inferred lifetime.
Line~6 on the right is the explicit drop of the qubit \qurtsinline{q}
and explicit declaration of the end of lifetime \qurtsinline{'a}.
This code compiles to the quantum circuit on the left below
-- before the uncomputation of the auxiliary qubit \qurtsinline{q}.
The circuit on the right below is one of the possible circuits with uncomputation inserted for the auxiliary qubit \qurtsinline{q}.
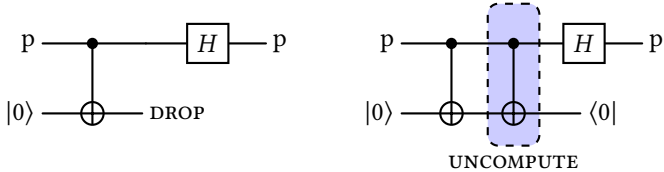
\begin{center}
  \begin{tikzcd}
    \llap{p}\;&\ctrl{1}&
    &\gate{H}&\;\rlap{p} \\
    \llap{$\ket{0}$}\;\setwiretype{n}&\targ{}\setwiretype{q}& \;\rlap{\textsc{drop}}
  \end{tikzcd}
  \qquad \qquad
  \begin{tikzcd}
    \llap{p}\;&\ctrl{1}&
    \ctrl{1}
    \gategroup[2,steps=1,style={dashed,rounded corners,fill=blue!20, inner xsep=2pt},
      background,label style={label position=below,anchor=north,yshift=-0.2cm}]{\textsc{uncompute}}
    &\gate{H}&\;\rlap{p} \\
    \llap{$\ket{0}$}\;\setwiretype{n}&\targ{}\setwiretype{q}&\targ{}&\;
    \rlap{$\bra0$}
  \end{tikzcd}
\end{center}

On the other hand, the following code causes a type error.
\begin{center}
  \hspace*{5mm}
  \begin{minipage}{0.4\textwidth}
    \begin{lstlisting}[language=qurts]
      { // mut p: #qbit
        let q = $\ket0$;
        let q = qif &p { [not](q) } else { q };


        let p = qif &q { [not](p) } else { p };
        
        
        return p;
      } // q is dropped before here
    \end{lstlisting}
  \end{minipage}
  \hspace*{1cm}
  \begin{minipage}{0.38\textwidth}
    \begin{lstlisting}[language=qurts]
      { // mut p: #qbit
        let q = $\ket0$;
        let q: #'a qbit =
          qif &'a p { [not](q) } else { q };
        endlft 'a;
        let p =
          qif &'b q { [not](p) } else { p };
        endlft 'b;
        return p;
      } // Type error: q cannot be dropped
    \end{lstlisting}
  \end{minipage}
\end{center}
This program replaces the Hadamard gate from before with another \qurtsinline{qif} statement.
It uses two different lifetimes \qurtsinline{'a} and \qurtsinline{'b}, where \qurtsinline{'a} ends at line~5.
This lifetime end is necessary:  all references to the qubit \qurtsinline{p} must have expired before the \qurtsinline{[not]} in line~7 can mutably use it.
In other words, the qubit \qurtsinline{p} borrowed in line~4 must be taken back before line~7.
Here, we can see that the type of variable \qurtsinline{q} is bounded by \qurtsinline{#'a} at line~3.
This is because of the typing rule for \qurtsinline{qif} statement:
the return value of the \qurtsinline{qif} statement has the lifetime used in the control qubit.
Therefore, the type of the variable \qurtsinline{q} becomes linear after the end of
\qurtsinline{'a} at line~5.
We have to drop the qubit \qurtsinline{q} in this program. But \qurtsinline{q} is linear after line~5, and since it is last used in line~7, it cannot be dropped. This program thus causes a type error.

In fact, this type error detects when uncomputation is impossible.
This computation can be represented as the following quantum circuit.
If qubit \qurtsinline{p} starts in state $\ket{+}$, then the computation proceeds as on the right below, where the first qubit is \qurtsinline{p} and the second qubit is \qurtsinline{q}.
The probability of the output state is $2 = (\sqrt2)^2$, which exceeds $1$,
showing that qubit \qurtsinline{q} is not uncomputable.
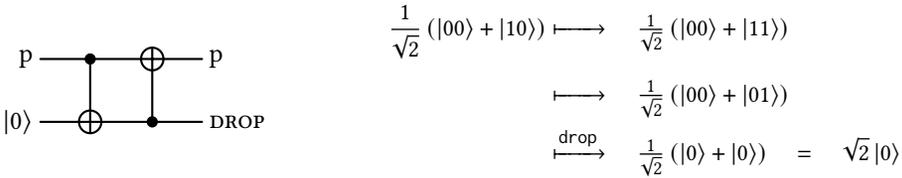
\begin{center}
  \begin{tikzcd}
    \llap{p}\;&\ctrl{1}&\targ{}&\;\rlap{p} \\
    \llap{$\ket{0}$}\;\setwiretype{n}&\targ{}\setwiretype{q}&\ctrl{-1}&\;\rlap{\textsc{drop}}
  \end{tikzcd}
  \qquad\qquad\qquad
  \begin{minipage}{0.5\textwidth}\small
    \begin{align*}
      \frac{1}{\sqrt2} \left(\ket{00} + \ket{10}\right)
       & \xmapsto{\phantom{\texttt{drop}}}\quad \tfrac{1}{\sqrt2} \left(\ket{00} + \ket{11}\right) \\
       & \xmapsto{\phantom{\texttt{drop}}}\quad \tfrac{1}{\sqrt2} \left(\ket{00} + \ket{01}\right) \\
       & \xmapsto{\texttt{drop}}\quad \tfrac{1}{\sqrt2} \left(\ket{0} + \ket{0}\right)
      \quad=\quad \sqrt2\, \ket0
    \end{align*}
  \end{minipage}
\end{center}

\paragraph{Two implementations of CX}
The CX gate, in quantum circuit form as on the right below, can be written in two ways in Qurts.
One uses \qurtsinline{qif} and \qurtsinline{[not]} as on the left below.
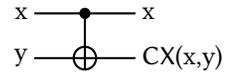
\begin{center}
  \begin{minipage}{0.55\textwidth}
    \begin{lstlisting}[language=qurts]
      fn my_cnot<'a>(x:&'a qbit, y:#'a qbit) -> #'a qbit {
        let y' = qif x { [not](y) } else { y }; return y'
      }
    \end{lstlisting}
  \end{minipage}
  \qquad
  \begin{minipage}{0.2\textwidth}
    \centering
    \begin{tikzcd}[row sep=.35cm]
      \llap{x}\;&\ctrl{1}&\;\rlap{x} \\
      \llap{y}\;&\targ{}&\;\rlap{\textsf{CX}(x,y)}
    \end{tikzcd}
  \end{minipage}
\end{center}
The other implementation is given by a lifted function -- the lift of the classical boolean controlled-\textsc{not} gate.
This can be written as the primitive Qurts-term below.
\begin{center}
  \begin{minipage}{0.6\textwidth}
    \begin{lstlisting}[language=qurts,numbers=none]
      [cnot] : (#'a qbit, #'a qbit) -> (#'a qbit, #'a qbit)
    \end{lstlisting}
  \end{minipage}
\end{center}
The type difference may seem minor, but it is important to verify uncomputability.
We have already seen that the second qubit in the circuit on the left below is not uncomputable, and that this circuit can in fact not result from compilation of Qurts code using \qurtsinline{qif} because of a type error.
Using a lifted function \qurtsinline{[cnot]} also causes a type error.
Firstly, the output of the lifted function has the same type as its input.
Therefore, if we can drop the second qubit, then we must be able to drop the first qubit at $p$.
But since the quantum state of the first qubit was $\ket+$ at $p$, this could not be affine and hence is not uncomputable.
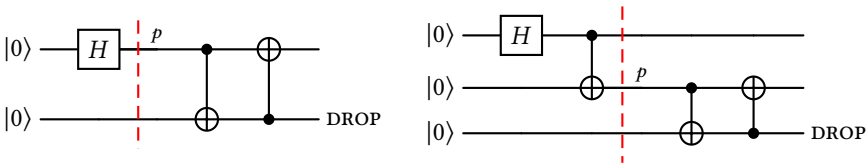
\begin{center}
  \begin{tikzcd}
    \llap{$\ket0$}\;&\gate{H}\slice{}&\wire[l][1]["p"{above,pos=0}]{a}&\ctrl{1}&\targ{}&\;\rlap{} \\
    \llap{$\ket0$}\;\setwiretype{n}&\setwiretype{q}&&\targ{}&\ctrl{-1}&\;\rlap{\textsc{drop}}
  \end{tikzcd}
  \qquad\qquad
  \begin{tikzcd}[row sep=.27cm]
    \llap{$\ket0$}\;&\gate{H}&\ctrl{1}&&&&\;\\
    \llap{$\ket0$}\;&&\targ{}\slice{}&\wire[l][1]["p"{above,pos=0}]{a}& \ctrl{1}&\targ{}&\;\rlap{} \\
    \llap{$\ket0$}\;\setwiretype{n}&\setwiretype{q}&&&\targ{}&\ctrl{-1}&\;\rlap{\textsc{drop}}
  \end{tikzcd}
\end{center}

On the other hand, the circuit on the right above is valid, and the third qubit is uncomputable.
Indeed, the final quantum state of the circuit on the right above is $\frac{1}{\sqrt2}(\ket{000} + \ket{101})$,
and after uncomputing the third qubit, the final quantum state is $\frac{1}{\sqrt2}(\ket{00}+\ket{10})$.
However, this can not be generated by the Qurts program using CX gates implemented with \qurtsinline{qif} for the same reason as above: we cannot mutably use the second qubit in the last CX gate.
On the other hand, the Qurts program using the lifted function \qurtsinline{[cnot]} can generate this circuit, because at $p$ the second qubit can be affine, as in line 2 of the following code.
\begin{center}
  \begin{minipage}{.5\textwidth}
    \begin{lstlisting}[language=qurts]
      let x:&qbit = &$\ket0$.H();
      let p:#qbit = qif x { $\ket1$ } else { $\ket0$ };
      let q:#qbit = $\ket0$;
      let (p,q) = [cnot](p, q);
      let (q,p) = [cnot](q, p);
      drop q;
    \end{lstlisting}
  \end{minipage}
\end{center}

\paragraph{Dropping variables}
In Qurts, you can define a function which drops a qubit even if the function itself does not know how to uncompute it.
The code on the left below implements a function \qurtsinline{forget} which drops the qubit \qurtsinline{x} after lifetime \qurtsinline{'a}.
The constraint \qurtsinline{<'a!='0>} in the signature means that the lifetime \qurtsinline{'a} is still valid when the function is called.
This verifies that the qubit \qurtsinline{x} is dropped while it is affine.
The code on the right below implements a function \qurtsinline{reinitialise} which implicitly
drops the qubit \qurtsinline{x} inside the \qurtsinline{qif} statement.
\begin{center}
  \hspace*{5mm}
  \begin{minipage}{0.37\textwidth}
    \begin{lstlisting}[language=qurts]
      fn forget<'a!='0>(x: #'a qbit) {
        // x is affine during the lifetime 'a
      }
    \end{lstlisting}
    \begin{lstlisting}[language=qurts]
      // the code above is equivalent to
      fn forget<'a!='0>(x: #'a qbit) {
        drop x;
      }
    \end{lstlisting}
  \end{minipage}
  \qquad\qquad
  \begin{minipage}{0.3\textwidth}
    \begin{lstlisting}[language=qurts]
      fn reinitialise<'a!='0>(
        x: #'a qbit,
        y: &'a qbit
      ) -> #'a qbit {
        qif y { // drop x;
          $\ket0$
        } else {
          x
        }
      }
    \end{lstlisting}
  \end{minipage}
\end{center}
Compare this with the identity function \qurtsinline{fn id(x:#'a qbit) -> #'a qbit \{ x \}}, which needs no restrictions on the lifetime \qurtsinline{'a}.

\paragraph{Toy example}\label{paragraph:toy-example}
To illustrate semantics we will use a final toy example.
\begin{center}
  \begin{minipage}{0.27\textwidth}
    \begin{lstlisting}[language=qurts]
      let mut x: #'0 qbit = $\ket0$;
      x.H();
      let r = $\&$x;
      let y = qif r { $\ket1$ }
              else  { $\ket0$ };
      meas(x)
    \end{lstlisting}
  \end{minipage}
  \qquad\qquad
  \begin{minipage}{0.32\textwidth}
    \begin{lstlisting}[language=qurts]
      let x0 = [$0$]();
      x0 as $\Own^\bot$qbit;
      let x1 = H(x0);
      newlft $\alpha$;
      let r = $\&^\alpha$x1;
      let y = qif r { let z = $[1]$(); z }
              else  { let z = $[0]$(); z };
      drop r;
      drop y;
      endlft $\alpha$;
      let ret = meas(x1);
      ret
    \end{lstlisting}
  \end{minipage}
\end{center}
The code on the left is a simple quantum program in Qurts.
The only feature we have not discussed yet is the \qurtsinline{meas} function,
which measures a qubit and returns a bit.
On the right is the same program in the core language which we will define formally in the next section.
Let us explain it briefly.
All the initialization of a qubit are now identified with the call of the lifted function $[0]$ or $[1]$
where $0$ is the constant function of $\False$, and $1$ is that of $\True$.
For simplicity, we omit mutable variables.
Line~2 coerces the variable \qurtsinline{x0} to $\Own^\bot \Qbit$, meaning that it will be treated linearly.
We make every introduction and elimination of lifetime explicit with $\Newlft$ and $\Endlft$.
This is the same for $\Drop$.

\subsection{Syntax and Type System}\label{subsec:lang-def}

We formalize the core language of Qurts by presenting its syntax and type system.
We only overview them in this section. The detailed formalisation is in \Cref{ax-sec:type-system}.

\subsubsection{Syntax}
\paragraph{Program}
The syntax of Qurts-core, which is the core language of Qurts, is defined in \Cref{fig:bnf/var-syntax}.
Most of the syntax and typing rules are taken from Calculus of Ownership and Reference~\cite{rusthorn}, a formalised core language of Rust.
A program is defined as a sequence of function definitions.
A function definition consists of a function name, a signature, and a body block.
A signature includes a lifetime preorder, a list of arguments, and a return type.
Here, a lifetime preorder is a set of lifetime variables and a preorder on them that presents the constraints
on the lifetime variables.
The lifetime preorder will also used as a part of the type derivation context later.
A block is a pair of a statement and a return value.

\begin{figure}
  \raggedright \textbf{Syntax:}
  \footnotesize\begin{mathpar}
    \begin{array}{llrl}
      \textit{Program}
       & \Pi
       & \sdef
       & F_0 \dots F_{n-1}
      \\
      \textit{Function}
       & F
       & \sdef
       & \Fn f\ \Sigma\ B
      \\
      \textit{Signature}
       & \Sigma
       & \sdef
       & \mathbf{A}\ (x_0\colon T_0,\dots,x_{n-1}\colon T_{n-1}) \rightarrow T
      \\
      \textit{Block}
       & B
       & \sdef
       & \{\ S;\ x\ \}
      \\
      \textit{Statement}
       & S
       & \sdef
       & \Noop \sor S_1;\ S_2 \sor \Newlft \alpha \sor \Endlft \alpha \sor \alpha \leq \beta \sor x \As T
      \\
       &
       & \sor
       & \Let y = \&^\alpha x \sor \Let y = e \sor \Let (y_0,y_1) = x \sor \Drop x
      \\
      \textit{Expression}
       & e
       & \sdef
       & x \sor \False \sor \True \sor () \sor (x_0, x_1)  \sor \Copy x
      \\
       &
       & \sor
       & \Meas(x) \sor U(x) \sor [c](x) \sor f\langle \alpha_0,\dots,\alpha_{m-1} \rangle(x_0,\dots,x_{n-1})
      \\
       &
       & \sor
       &
      \If \ x\ B_t\ \Else\ B_f\ \sor
      \Qif\ x\ B_\ket{1}\ \Else\ B_\ket{0}
      \\
      \textit{Lifetime Preorder}
       & \mathbf{A}
       & \sdef
       & \langle A, R\rangle \text{ where $A$: a set of lifetime variables, $R$: a preorder on $A\cup{\{\bot,\top\}}$} \\
    \end{array}
  \end{mathpar}
  \begin{center}
    ($x$,$y$ are local variables, $\alpha,\beta$ are lifetime variables, $f$ is a function name.)
  \end{center}
  \Description{variable, syntax, and type of the language}
  \caption{Qurts syntax}
  \label{fig:bnf/var-syntax}
\end{figure}

\paragraph{Types}
The types of Qurts-core are defined in \Cref{fig:bnf/context}.
As we have seen in the examples, Qurts-core has linearity $\Own^\alpha T$ and immutable reference $\&^\alpha T$.
We also have product types $T_1 \times T_2$ and unit type $()$.
For the lifetime, in addition to the lifetime variables $\alpha,\beta,\ldots$,
we have the shortest lifetime $\bot$, called `empty',
and the longest lifetime $\top$, called `static'.
In ascii, $\alpha$ is written as \qurtsinline{'a}, $\bot$ as \qurtsinline{'0}, and $\top$ as \qurtsinline{'static}.
The empty lifetime is a new special lifetime which was not in Rust,
and it is used to represent the linear type by $\Own^\bot$.

\begin{figure}
  \begin{minipage}{.49\textwidth}
    \raggedright \textbf{Types:}
    \footnotesize
    \begin{mathpar}
      \begin{array}{llrl}
        \textit{Lifetime}
         & \lft{a}
         & \sdef
         & \alpha \sor \bot \sor \top                             \\
        \textit{Type}
         & T
         & \sdef
         & \Bool \sor \Qbit \sor ()
        \\
         &
         & \sor
         & T_1 \times T_2 \sor \&^\lft{a}\ T \sor \Own^\lft{a}\ T
      \end{array}
    \end{mathpar}
  \end{minipage}
  \begin{minipage}{.49\textwidth}
    \raggedright \textbf{Contexts:}
    \footnotesize
    \begin{mathpar}
      \begin{array}{llrl}
        \textit{Type Context} & \mathbf{\Gamma} & \sdef & x_0\colon^\mathbf{a_0} T_0,\dots,x_{n-1}\colon^\mathbf{a_{n-1}} T_{n-1} \\
        \textit{Aliveness}    & \mathbf{a}      & \sdef & \Active \sor \dagger \alpha                                             \\
      \end{array}
    \end{mathpar}
    \vspace{2em}
  \end{minipage}
  \Description{}
  \caption{Qurts types and judgement contexts}
  \label{fig:bnf/context}
\end{figure}

\paragraph{Rust-like features}

As we have seen in the examples, Qurts-core has explicit declarations of lifetimes.
The statement $\Newlft{\alpha}$ starts a new lifetime variable $\alpha$, and the statement $\Endlft{\alpha}$ ends it. Every statement $\Newlft{\alpha}$ must have a unique corresponding statement $\Endlft{\alpha}$ in the same scope. In full Qurts such statements should be automatically inferred if the programmer does not specify them, but in the core language they are explicit. We also impose that every lifetime variable only lives once. That is, a lifetime variable cannot be restarted after it has ended, or more precisely, the following code is prohibited:
\begin{center}
  \begin{minipage}{0.5\textwidth}
    \begin{lstlisting}[language=qurts]
        endlft 'a;
        newlft 'a;
    \end{lstlisting}
  \end{minipage}
\end{center}
The statement $\alpha \leq \beta$ is used to add constraints on lifetimes ($\alpha$ ends before $\beta$).
$x \As T$ is used to coerce the type of $x$ to $T$.
For example, we can coerce the linearity type $\Own^\beta$ to $\Own^\alpha$ when $\alpha\leq\beta$.
We also have a statement $\Let y = \&^\alpha x$ to immutably borrow the variable $x$ with lifetime $\alpha$,
$\Drop x$ to drop the variable $x$, $\Copy x$ to copy the variable $x$.
These $\alpha$, $\Drop$, and $\Copy$ are also assumed to be inferred in full Qurts.
As in Rust, to call a function, we need to assert that the lifetime of the arguments is valid.
To do this, the correspondence between the lifetime variables in the function signature
and the lifetime variables in the caller are assumed to be given explicitly in the core language.
For simplicity, we do not allow a loop or a recursive function in the core language,
so we assume all the well-typed functions terminates.

\paragraph{Quantum features}

Qurts also has statements and expressions that are specific to quantum computing. The programmer can apply a \emph{unitary gate} $U$ to a list of qubits $x$ with the syntax $U(x)$.
We also have primitive isometric gates where the isometry is the \emph{lift} of a classical Boolean injective function $c$ on bits, given by special consideration with the syntax $[c](x)$.
This includes the introduction of a new qubit $[0]$, $X$-gate, $CX$-gate, and so on.
The syntax $\Meas(x)$ performs a \emph{quantum measurement} on the qubit $x$ in the computational basis; the value of this expression is the outcome of the measurement as a boolean value.
Finally, Qurts has an analogue of the classical conditional statement with syntax $\Qif x\ B_1 \Else B_0$.
This expression evaluates to the execution of the blocks in both branches in superposition according to qubit $x$.

\subsubsection{Type System}

\paragraph{Typing judgements}
A judgement for a statement is defined by a relation on pairs of a type context and a lifetime preorder
as $S \colon (\mathbf{\Gamma}, \mathbf{A}) \to (\mathbf{\Gamma}', \mathbf{A}')$,
where $(\mathbf{\Gamma},\mathbf{A})$ is the context before the statement
and $(\mathbf{\Gamma}',\mathbf{A}')$ is the context after the statement.
The typing context $\mathbf{\Gamma}$ is basically a list of variables and their types, but also annotates variables with \emph{aliveness}, which can be $\Active$ or $\dagger^\alpha$.
By default a variable is $\Active$, written simply as $x\colon T$,
and it becomes \emph{frozen} ($\dagger^\alpha$) during the lifetime $\alpha$ when the variable is immutably borrowed.
The lifetime preorder $\mathbf{A}$ represents the set of lifetimes that has not ended yet.
Selected (simplified) typing rules for statements are shown in \Cref{fig:rule/selected-ty-stmt}.

\begin{figure}[tb]
  \raggedright \textbf{Typing Statements:}
  \hfill
  \fbox{
    $S \colon (\mathbf{\Gamma}, \mathbf{A}) \to (\mathbf{\Gamma}', \mathbf{A}')$
  }
  \footnotesize
  \begin{mathpar}
    \inferH{typ new lft}{}{
      \Newlft\alpha
      \colon (\mathbf{\Gamma}, \mathbf{A}) \to
      (\mathbf{\Gamma}, \mathbf{A} \cup \{\alpha\})
    }

    \inferH{typ end lft}{
      \alpha \text{ is minimal in }\mathbf{A}-\{\bot\}\\
      \&^\alpha \text{ does not apper in } \Gamma
    }{
      \Endlft\alpha \colon (\mathbf{\Gamma}, \mathbf{A}) \to
      (\defrost_\alpha(\mathbf{\Gamma}), \mathbf{A}-\alpha)
    }

    \defrost_\alpha(x\colon^\mathbf{a}T) \coloneqq
    \begin{cases}
      x\colon T            & (\mathbf{a}= \dagger\alpha) \\
      x\colon^\mathbf{a} T & (\text{otherwise})
    \end{cases}

    \inferH{typ fn call}{
    \Sigma_{g} = \langle
    \beta_0,\dots,\beta_{m-1} \mid \beta_{a_0} \leq \beta_{b_0},\dots,\beta_{a_{l-1}} \leq \beta_{b_{l-1}}
    \rangle
    (x'_0\colon S_0,\dots,x'_{n-1}\colon S_{n-1})
    \rightarrow S_n \\
    \forall j \in \{ 0,\dots,l-1 \},\ \alpha_{a_j}\leq \alpha_{b_j} \in \mathbf{A} \\
    \forall i \in \{ 0,\dots,n   \},\ T_i = S_i[\alpha_0/\beta_0,\dots,\alpha_{m-1}/\beta_{m-1}]
    }{
    \Let y = g\langle \alpha_0,\dots,\alpha_{m-1} \rangle (x_0,\dots,x_{n-1}) \colon
    (\mathbf{\Gamma} + \{ x_0 \colon T_0,\dots,x_{n-1} \colon T_{n-1} \}, \mathbf{A}) \to
    (\mathbf{\Gamma} + \{ y \colon T_n \}, \mathbf{A})
    }

    \inferH{typ borrow}{
      \forall \gamma \in \{ \alpha\ |\ \&^\alpha \text{ appear in }  T\},\ \alpha \leq \gamma \in \mathbf{A}\\
    }{
      \Let y = \&^\alpha x \colon
      (\mathbf{\Gamma} + \{ x \colon T \}, \mathbf{A})
      \to
      (\mathbf{\Gamma} + \{ y \colon \&^\alpha T,\ x \colon^{\dagger\alpha} T \}, \mathbf{A})
    }

    \inferH{typ drop}{
      \mathbf{A} \vdash T\colon \DROP
    }{
      \Drop x \colon (\mathbf{\Gamma} + \{ x \colon T \}, \mathbf{A}) \to (\mathbf{\Gamma}, \mathbf{A})
    }

    \inferH{typ meas}{}{
      \Let y = \Meas(x) \colon
      (\mathbf{\Gamma} + \{ x\colon\Own^{\color{purple}\lft{a}}\Qbit \}, \mathbf{A})
      \to
      (\mathbf{\Gamma} + \{ y\colon\Own^{\color{purple}\top}\Bool \}, \mathbf{A})
    }

    \inferH{typ unitary}{
      U \in \mathbf{Unitary}(n)
    }{
      \Let y = U(x) \colon
      (\mathbf{\Gamma} + \{ x\colon\Own^{\color{purple}\lft{a}}\Qbit^n \}, \mathbf{A})
      \to
      (\mathbf{\Gamma} + \{ y\colon\Own^{\color{purple}\bot}\Qbit^n \}, \mathbf{A})
    }

    \inferH{typ lifted}{
      c \colon {\{0,1\}}^n \to {\{0,1\}}^m; \text{injection}
    }{
      \Let y = [c](x) \colon
      (\mathbf{\Gamma} + \{ x\colon\Own^{\color{purple}\lft{a}}\Qbit^n \}, \mathbf{A})
      \to
      (\mathbf{\Gamma} + \{ y\colon\Own^{\color{purple}\lft{a}}\Qbit^m \}, \mathbf{A})
    }

    \inferH{typ qif}{
      x\colon \&^\alpha \Qbit \in \mathbf{\Delta} \\
      \alpha \in \mathbf{A} \\
      S_i \colon (\mathbf{\Gamma}, \mathbf{A}) \to (y_i\colon T, \mathbf{A}) \\
      S_i \colon \PQ \\
      T\colon\PQ
    }{
      \Let y = \Qif\ x\ \{\ S_1;\,y_1\ \} \ \Else\ \{\ S_0;\,y_0\ \}
      \colon
      (\mathbf{\Gamma} + \mathbf{\Delta}, \mathbf{A})
      \to
      ((y\colon\Own^\alpha T) + \mathbf{\Delta}, \mathbf{A})
    }
  \end{mathpar}
  \Description{}
  \normalsize\caption{Selected typing rule for statements}
  \label{fig:rule/selected-ty-stmt}
\end{figure}

\paragraph{Lifetimes}
Introduction and elimination of lifetimes is governed by the rules \scref{typ new lft} and \scref{typ end lft}.
The rule \scref{typ new lft} simply adds a new lifetime $\alpha$ to the lifetime preorder $\mathbf{A}$.
The rule \scref{typ end lft} states that lifetime $\alpha$ can be ended if there is no other lifetime that has to end before $\alpha$, and there is no reference $\&^\alpha$ in the typing context.
After ending lifetime $\alpha$, all frozen variables with lifetime $\alpha$ are \emph{defrosted} by
recovering the original context $x\colon T$ from $x\colon^{\dagger\alpha} T$.

The presence of lifetimes may make the Qurts-core typing rules look a bit complicated.
For example, the premise of the rule \scref{typ fn call} for function calls verifies the lifetime constraints of the arguments by substituting all the lifetimes in the function signature with those in the caller.
Once lifetimes are ignored, however, these rules are no different from the conventional ones.

The rule \scref{typ borrow} governs borrowing a variable.
This rule freezes the variable $x$ with lifetime $\alpha$ and introduces a new variable $y$ with the type $\&^\alpha T$.
The premise ensures that no reference can be borrowed after the lifetime of the reference itself.

The rule \scref{typ drop} governs dropping a variable from the context, and is critical for affine types with lifetime restrictions.
The judgement $\mathbf{A} \vdash T\colon \DROP$ in \Cref{fig:rule/Drop} means that values of type $T$ can be dropped at time $\mathbf{A}$.
The rule \scref{drop borrow} states that types $\&^\alpha T$ can be dropped at any time.
On the other hand, \scref{drop own} states that the type $\Own^\lft{a} T$ can only be dropped until the lifetime $\lft{a}$ ends.

\begin{figure}[tb]
  \raggedright \textbf{Drop:}
  \hfill
  \fbox{
    $\mathbf{A} \vdash T\colon \DROP$
  }
  \footnotesize\begin{mathpar}
    \inferH{drop bool}{}{
      \mathbf{A}\vdash \Bool \colon \DROP
    }

    \inferH{drop borrow}{}{
      \mathbf{A}\vdash \&^\lft{a}T \colon \DROP
    }

    \inferH{drop own}{
      \lft{a}\in \mathbf{A}
    }{
      \mathbf{A}\vdash \Own^\lft{a}T \colon \DROP
    }

    \inferH{drop unit}{}{
      \mathbf{A}\vdash()\colon \DROP
    }

    \inferH{drop tuple}{
      T_0\colon \DROP \\
      T_1\colon \DROP
    }{
      T_0 \times T_1 \colon \DROP
    }
  \end{mathpar}
  \Description{}
  \normalsize\caption{Drop trait}
  \label{fig:rule/Drop}
\end{figure}

\paragraph{Quantum operations}
The next three typing rules
\scref{typ meas}, \scref{typ unitary} and \scref{typ lifted}
govern basic quantum operations: measurements, unitary gates and lifted isometries.
All three expressions take qubit(s) bound with $\Own^\lft{a}$ by lifetime $\lft{a}$ as an argument, but their return types are different.
The return type of a measurement is $\Own^\top\Bool$: a classical boolean that can always be dropped.
General unitary gates have a linear return type $\Own^\bot\Qbit^n$, because the output is not uncomputable.
Lifted isometries have return type $\Own^\lft{a}\Qbit^m$ and remain affine if the input is.

The most complicated and critical typing rule for Qurts-core is the rule \scref{typ qif} governing $\Qif$ statements.
This rule only applies when no measurement occurs under quantum control, and no `superposition of classical data' is created in the return type.
These properties are checked by the judgement $S_i\colon\PQ$ and $T\colon\PQ$ (Purely Quantum) in \Cref{fig:rule/PQ}.
The former means that the statement does not include any measurement, which is close to Silq's \texttt{mfree}~\cite{Silq_2020}.
The latter means that the type does not include any booleans or references.
It is also important to note that the type of the return value $y$ is not the type $T$ of the two branches,
but must be $\Own^\alpha T$ with the lifetime $\alpha$ of the reference $x$.

\begin{figure}[tb]
  \raggedright \textbf{Purely Quantum:}
  \hfill
  \fbox{
    $\underline{\ \ }\,\colon \PQ$
  }
  \footnotesize\begin{mathpar}
    \inferH{pq ty base}{}{\Qbit\colon\PQ}\hspace{-1em}

    \inferH{pq ty own}{T\colon\PQ}{\Own^\lft{a}T\colon\PQ}\hspace{-1em}

    \inferH{pq ty tuple}{
      T_0\colon\PQ \\
      T_1\colon\PQ
    }{
      T_0\times T_1\colon\PQ
    }\hspace{-1em}

    \inferH{pq expr}{
      e \neq \Meas x \wedge
      \left(\left(e = f\langle\overline\alpha\rangle(\overline{x}) \wedge \text{fn-name}(F) = f\right) \Rightarrow F \colon \PQ\right)
    }{
      e\colon\PQ
    }

    \inferH{pq stmt}{
      S \text{ does not include non-PQ expression}
    }{
      S \colon \PQ
    }

    \inferH{pq block}{
      S\colon\PQ
    }{
      \{\ S;\ x\ \}\colon\PQ
    }

    \inferH{pq fn}{
      B\colon\PQ\\
    }{
      \Fn\ f\ \mathbf{A}\ \mathbf{\Gamma} \rightarrow T\ B \colon \PQ
    }
  \end{mathpar}
  \Description{}
  \normalsize\caption{Purely Quantum: types without classical data and programs without measurements}
  \label{fig:rule/PQ}
\end{figure}

\paragraph{Block, Function}
The typing rule for blocks and functions are shown in \Cref{fig:rule/ty-block-fn}.
The rule \scref{typ block} for typing a block $\{\ S;\, x\ \}$ states that there must be only one variable $x$ in the context after evaluating $S$ for linearity.
The rule \scref{typ fn} for typing a function $\Fn\ f\ \Sigma\ B$ states that the function $f$ is well-typed if the body $B$ is well-typed under the context $(\mathbf{\Gamma},\mathbf{A})$ generated from the signature $\Sigma$.

\begin{figure}[t]
  \raggedright \textbf{Typing Block, Function:}
  \footnotesize
  \begin{mathpar}
    \inferH{typ block}{
      S\colon (\mathbf{\Gamma}, \mathbf{A}) \to (x\colon T, \mathbf{A})\\
    }{
      (\mathbf{\Gamma},\mathbf{A}) \vdash \{\ S;\ x\ \} \colon T
    }

    \inferH{typ fn}{
      F = \Fn\ f\ \Sigma\ B\\
      \Sigma = \langle \alpha_0,\dots,\alpha_{m-1} \mid \alpha_{a_0} \leq \alpha_{b_0},\dots,\alpha_{a_{l-1}} \leq \alpha_{b_{l-1}} \rangle
      \ \mathbf{\Gamma} \rightarrow T\\
      \mathbf{A}\colon\text{ the smallest preorder on } \{ \alpha_i \} \text{ including } \{ \alpha_{a_j} \leq \alpha_{b_j} \}\\
      (\mathbf{\Gamma},\mathbf{A}) \vdash B\colon T \\
    }{
      \vdash F\colon\text{Function}
    }
  \end{mathpar}
  \Description{}
  \normalsize\caption{Typing rule for blocks, function, and whole program}
  \label{fig:rule/ty-block-fn}
\end{figure}
\section{Simulation Semantics}\label{sec:simulation-semantics}
This section discusses the first of the two small-step operational semantics, which we call \emph{simulation semantics}. It represents a computation as a transition relation on the environment, including the state of quantum and classical memories.
From now on, we assume all the programs we present are well-typed and equipped with their type derivation tree including the lifetimes.

The most notable rule of the simulation semantics is the one governing $\Drop$.
When a variable is dropped, the memory locations owned by the variable are \emph{deallocated}.
The deallocation of the last qubit of the quantum memory is defined as the operation that converts the state
$\ket{\phi_0}\ket0 + \ket{\phi_1}\ket1$ to $\ket{\phi_0} + \ket{\phi_1}$.
For general states $\ket{\phi_i}$, this operation is neither unitary nor norm-non-increasing, and therefore not physically realisable.
For example, if $\ket{\phi_0} = \ket{\phi_1} = \frac{1}{\sqrt2}\ket{0}$, then the norm of the initial state is
$|\frac{1}{\sqrt2}\ket{00}+\frac{1}{\sqrt2}\ket{01}| = 1$,
while the norm of the final state is
$|\frac{1}{\sqrt2}\ket{0} + \frac{1}{\sqrt2}\ket{0}| = |\sqrt2\ket0| = \sqrt2 > 1$.
For orthogonal states $\{\ket{\phi_i}\}$ physicality is restored, as the dropping operation preserves norms.

\Cref{subsec:sim-property} below proves that the type system guarantees that the evaluation step of a well-typed program in a well-formed environment can only perform $\Drop$ for a state satisfying the orthogonality condition.
Therefore, the simulation semantics is norm-preserving for all cases except measurement.

\subsection{Simulation Semantics Definition}\label{subsec:sim-def}
All variables are considered to be pointers:
states of quantum and classical memory are stored in locations, and variables are mapped to memory locations.

Let $\Loc_q$ and $\Loc_c$ be the sets of locations
for quantum and classical memories.
An \emph{environment} is a triple $(\mathrm{loc}, \ket{\phi}, s)$ of a map
$\mathrm{loc}\colon \mathrm{Var} \rightarrow \mathrm{List}(\Loc_c + \Loc_q)$ from variables to locations,
$\ket{\phi}\in\HH_{L_q} \Coloneqq \bigotimes_{l\in L_q}\C^2$ and $s\in \{0,1\}^{L_c}$
are the states of quantum and classical memories, and $L_q\subset \Loc_q$ and $L_c\subset\Loc_c$ are the image of $\mathrm{loc}$.

An environment $(\mathrm{loc}, \ket{\phi}, s)$ is \emph{compatible} with a context $\mathbf{\Gamma}$ when $\mathrm{loc}(x)$ has the form $\sem{T}$ for all the variables $x\colon^\mathbf{a} T$ in $\mathbf{\Gamma}$, where $\sem{T}$ is defined as follows.
(`$+$' denotes the concatination of lists.)
\begin{align*}
  \sem{\Qbit}          & = [l_q],                 &
  \sem{\Bool}          & = [l_c],                   \\
  \sem{\Own^\lft{a} T} & = \sem{T},               &
  \sem{\&^\alpha T}    & = \sem{T},                 \\
  \sem{()}             & = [],                    &
  \sem{T_1 \times T_2} & = \sem{T_1} + \sem{T_2}.
\end{align*}
Each statement $(\mathbf{\Gamma}, \mathbf{A})\stackrel{S}{\rightarrow} (\mathbf{\Gamma}', \mathbf{A}')$
gives a transition relation $e\xrightarrow S e'$ on compatible environments.

To define the semantics of the statements $\Drop$ and $\Copy$,
we first discuss \emph{ownership} of the locations.
Let $e$ be an environment compatible with a type context $\mathbf{\Gamma}$, and $x\colon T$ be a variable in that context.
When a location $l$ is included in $\mathrm{loc}(x)$,
and the corresponding appearance of basic type
$\Qbit$ or $\Bool$ in $T$ is not bound by
any $\&^\alpha$, then we say that $l$ is \emph{owned} by $x$.

Selected rules are presented in \Cref{fig:rule/sim-selected-rules}, all (other) rules are defined in \Cref{ax-sec:sim-sem}.

\begin{figure}
  \footnotesize
  \begin{mathpar}
    \inferH{simulation drop}{
    \Drop x \colon (\mathbf{\Gamma},-) \rightarrow (-,-)                     \\
    L_q = [\ l_q\in\Loc_q \ |\ l_q\ \text{owned by}\ x\ ]                    \\
    L_c = [\ l_c\in\Loc_c \ |\ l_c\ \text{owned by}\ x\ ]                    \\\\
    \ket\phi = \sum_{i\in\{0,1\}^{|L_q|}} \ket{\phi_i}\ket{i} \in \HH_{\mathrm{others}}\otimes\HH_{L_q} \\
    s = (s_0,s_1) \in {\{0,1\}}^\text{others} \times {\{0,1\}}^{L_c}
    }{
    (\loc, \ket\phi, s)
    \xrightarrow{\Drop x}
    (\loc[x\mapsto[]], \sum_{i\in\{0,1\}^{|L_q|}} \ket{\phi_i}, s_0)
    }

    \inferH{simulation copy}{
      L_c = \{\ l_c\in \Loc_c\ |\ l_c: \text{owned by}\ x\ \}\\
      L'_c : \text{fresh classical locations}\\
      L' = \loc(x)[L'_c/L_c]\\
      t = s(L_c)
    }{
      (\loc,\ket\phi,s) \xrightarrow{\Let y\ =\ \Copy x}
      (\loc[y\mapsto L'], \ket\phi, s[t \leftarrow L'_c])
    }

    \inferH{simulation qif}{
    S = (\Let y = \Qif\ x\ \{\ S_1; x_1\ \}\ \Else\ \{\ S_0; x_0\ \}) \\
    \loc(x) = [l_x]                                                   \\
    L_i = \loc_i(x_i)                                                 \\
    |L| = |L_i|\\
    L : \text{fresh}                                                  \\\\
    (\loc,\ket{0}\ket{\phi_0},s)\xrightarrow{S_0}
    (
    \loc_0,
    \ket{0}\ket{\psi_0}\in \HH_{l_x}\otimes(\HH_{L_0}\otimes\HH_\text{others}),
    s'
    )                                                                 \\
    (\loc,\ket{1}\ket{\phi_1},s)\xrightarrow{S_1}
    (
    \loc_1,
    \ket{1}\ket{\psi_1}\in \HH_{l_x}\otimes(\HH_{L_1}\otimes\HH_\text{others}),
    s'
    )                                                                 \\
    }{
    (\loc,\sum_{i\in\{0,1\}} \ket{i}\ket{\phi_i},s)\xrightarrow{S}
    (\loc_0[x_0\mapsto[],y\mapsto L], \sum_{i\in\{0,1\}} \ket{i}\ket{\psi_i}, s')
    }
  \end{mathpar}
  \Description{}
  \normalsize\caption{Selected simulation semantics rules.}
  \label{fig:rule/sim-selected-rules}
\end{figure}

\paragraph{Drop}
The transition step \scref{simulation drop} for \qurtsinline{drop x} forgets the variable $x$ from the environment and drops the value stored in the memory locations owned by $x$. Note that when the variable $x$ includes a reference of a location $l$ but does not own it, the location $l$ is not deallocated.

\paragraph{Copy}
The rule \scref{simulation copy} defines the transition step for \qurtsinline{copy(x)}. When the variable $x$ is copied to $y$, it copies the classical memory owned by $x$ to $y$, as well as all references $x$ has.

\paragraph{Quantum If}
The most subtle rule in the simulation semantics is perhaps that governing \qurtsinline{qif}.
The semantics of the whole quantum if statement is the sum of the resulting state of the two blocks $\{S_i; x_i\}$ of the quantum if statement in \scref{simulation qif}.
The \textbf{Purely Quantum} rule guarantees that the state of the classical memory $s$ remains unchanged during the evaluation of the blocks.
As for the locations, since the two variables $x_i$ might not have the same locations,
we introduce a fresh set of locations $L$ and replace the locations of $x_i$ with $L$.
Thus, the program
\qurtsinline{qif p \{ let t1 = (y,x); t1 \} else \{ let t0 = (x,y); t0 \}}
is interpreted as a controlled-swap operation.

For the post-environment, the next lemma states that the locations $\loc_0$ and $\loc_1$ only differ in the locations owned by $x_i$,
proving that this rule is not biased towards either branch.
\begin{lemma}\label{lem:sim-PQ-loc}
  The map $\loc_i[x_i \mapsto [], y\mapsto L]$
  in the \scref{simulation qif}
  does not depend on $i$.
\end{lemma}

\subsection{Example}\label{subsec:sim-example}
Consider again the \hyperref[paragraph:toy-example]{toy example} from \Cref{subsec:lang-example}.
We will use it to illustrate how the simulation semantics works.
Say that the environment starts at $e = (\loc,v,s)$ with initial state $(\loc=\lambda \_. [], 1, ())$.
The program then evaluates as in \Cref{fig:eq:sim-semantics-toy-example}.
\begin{figure}
  \begin{center}
    \small \tabcolsep=0.25cm
    \begin{tabular}{@{}l>{$}r<{$}>{$}l<{$}>{$}l<{$}>{$}l<{$}@{}}
       line & 
       & \loc\colon \mathrm{Var}\to\mathrm{List}(\Loc_q+\Loc_c)
       & \text{quantum state}
       & \text{classical state}
      \\
       &
       & [],
       & 1,
       & ()
      \\
       1 &
      \xrightarrow{\text{\qurtsinline{let x0 = [0]();}}}
       & [x_0\mapsto [l_x]],
       & \ket{0},
       & ()
      \\
       2,3&
      \xrightarrow{\text{\qurtsinline{let x1 = H(x0);}}}
       & [x_1\mapsto [l_x]],
       & \ket{+} = \frac{1}{\sqrt{2}}\ket0 + \frac{1}{\sqrt{2}}\ket1,
       & ()
      \\
       4,5&
      \xrightarrow{\text{\qurtsinline{let r = &x1;}}}
       & [x_1\mapsto [l_x], r\mapsto [l_x]],
       & \ket{+},
       & ()
      \\
       6,7&
      \xrightarrow{\text{\qurtsinline{let y = qif r \{ ...}}}
       & [x_1\mapsto [l_x], r\mapsto [l_x], y\mapsto [l_y]],
       & {\textstyle \frac{1}{\sqrt{2}}}\ket{00} + {\textstyle \frac{1}{\sqrt{2}}}\ket{11},
       & ()
      \\
       8&
      \xrightarrow{\text{\qurtsinline{drop r;}}}
       & [x_1\mapsto [l_x], y\mapsto [l_y]],
       & {\textstyle \frac{1}{\sqrt{2}}}\ket{00} + {\textstyle \frac{1}{\sqrt{2}}}\ket{11},
       & ()
      \\
      9,10&
      \xrightarrow{\text{\qurtsinline{drop y;}}}
       & [x_1\mapsto [l_x]],
       & \ket{+},
       & ()
      \\
      11&
      \xrightarrow{\text{\qurtsinline{let ret = meas(x1);}}}
       & [b\mapsto [l_b]],
       & {\textstyle \frac{1}{\sqrt{2}}},
       & (0)
      \\
      11&
      \xrightarrow{\text{\qurtsinline{let ret = meas(x1);}}}
       & [b\mapsto [l_b]],
       & {\textstyle \frac{1}{\sqrt{2}}},
       & (1)
    \end{tabular}
  \end{center}
  \Description{}
  \caption[]{Toy example: evaluation steps of the simulation semantics.}
  \label{fig:eq:sim-semantics-toy-example}
\end{figure}
In line~1, the quantum memory $x_0$ is initialized to $\ket0$.
In line~2, the type of $x_0$ is coerced from $\Own^\top\Qbit$ to $\Own^\bot\Qbit$,
making $x_0$ linearly owned to be able to apply the Hadamard gate to it in line~3.
Lines~4 and 5 declare a new lifetime $\alpha$,
and borrow the variable $x_1$ to make a reference $r$ of type $\Own^\alpha \Qbit$.
Lines~6 and~7 are evaluated as the superposition of the following transition steps for $i = 0,1$.
\begin{center}
  \small \tabcolsep=0.3cm
  \begin{tabular}{@{}>{$}r<{$}>{$}l<{$}>{$}l<{$}>{$}l<{$}@{}}
     & [r\mapsto [l_x]],
     & {\textstyle \frac{1}{\sqrt{2}}}\ket{i},
     & ()
    \\
    \xrightarrow{\text{\qurtsinline{let z = [i]()}}}
     & [r\mapsto [l_x],\ z\mapsto [l_z]],
     & {\textstyle \frac{1}{\sqrt{2}}}\ket{ii},
     & ()
  \end{tabular}
\end{center}
In line~8, we drop the reference $r$, which leaves the memory unaltered.
On the other hand, dropping the variable $y$ deallocates the memory location $l_y$ that $y$ owned.
Line~11 measures the quantum memory $x_1$, resulting in two possible branches of transitions.
The two resulting environments have different classical states, each with probability $\frac{1}{2}$.

\subsection{Properties}\label{subsec:sim-property}
Next, we investigate properties of the simulation semantics.
For the proof, see \Cref{ax-sec:sim-sem}.
The first one is progress.

\begin{lemma}\label{lem:simulation-progress}
  Let $S$ be a statement, $e$ be an environment, and 
  $S\colon(\mathbf{\Gamma}, \mathbf{A}) \rightarrow (\mathbf{\Gamma'},\mathbf{A}')$ be a well-typed statement.
  If $e$ is compatible with $\mathbf{\Gamma}$,
  then there exists an compatible environment $e'$ such that $e \xrightarrow{S} e'$.
\end{lemma}

However, as we have seen in the example at the beginning of \Cref{sec:simulation-semantics}, compatibility of the environment does not guarantee physical realisability of the state.
We need to take into account types, especially whether they are affine or linear, and consider only certain \emph{well-formed} environments.

\begin{definition}
  \label{def:sim-affine-linear-frozen-loc}
  Let $\mathbf{A}$ be a lifetime preorder and $x\colon^\mathbf{a} T$ in type context $\mathbf{\Gamma}$. A location is:
  \begin{itemize}
    \item \emph{affinely owned} if $x\colon^\mathbf{a} T$ owns it and the corresponding basic type in $T$ is only bounded by $\Own^\lft{a}$ such that $0<\lft{a} \in \mathbf{A}$; otherwise it is \emph{linearly owned};
    \item \emph{frozen} when it is owned by a frozen variable $x\colon^{\dagger\alpha} T$ or not owned by any variable at all.
  \end{itemize}
\end{definition}

\begin{definition}
  \label{def:sim-well-formed-env}
  A triple of a type context $\mathbf{\Gamma}$, a lifetime preorder $\mathbf{A}$,
  and an environment $e$, is \emph{well-formed} if the following conditions are satisfied:
  \begin{itemize}
    \item the environment $e$ is compatible with $\mathbf{\Gamma}$;
    \item any location $l\in L_q+L_c$ is owned by at most one variable in $\mathbf{\Gamma}$;
    \item the state of the quantum memory in $\HH_{f} \otimes \HH_{l} \otimes \HH_{a}$ is of the form
          \[
            \sum_{i\in{\{0,1\}}^n}  \ket{i}\ket{\phi_{i}}\ket{f(i)}
          \]
          for a function $f\colon{\{0,1\}}^n\to {\{0,1\}}^m$, where $\HH_{f}$, $\HH_{l}$, and $\HH_{a}$ are the Hilbert spaces of memory locations which are frozen, linearly owned, and affinely owned.
  \end{itemize}
\end{definition}

Let us explain the last condition for well-formedness.
The space $\HH_{f}$ contains quantum memory locations that are immutably referenced or outside the scope,
and the space $\HH_{a}$ is the space of locations that are owned by some variable in the context,
which could be dropped in the future.
The space $\HH_{l}$ contains the remaining locations, with no information on their state.
The existence of a function $f\colon{\{0,1\}}^n\to {\{0,1\}}^m$ says that frozen locations determine which part of the quantum memory can be deallocated.
This condition also guarantees that deallocation is norm-preserving:
\[
  \Big|\sum_{i\in{\{0,1\}}^n}  \ket{i}\ket{\phi_{i}}\ket{f(i)} \Big|^2
  = \sum_{i\in{\{0,1\}}^n} \Big|\ket{\phi_i} \Big|^2
  = \Big|\sum_{i\in{\{0,1\}}^n}  \ket{i}\ket{\phi_{i}} \Big|^2.
\]

For each environment $e$, we define the \emph{probability} of the environment $e$ by $\|e\|^2 \coloneqq |\ket\phi|^2$ where $\ket\phi$ is the quantum state of $e$.

\begin{theorem}\label{thm:simulation}
  Let $\Fn f\ \mathbf{A}\ \mathbf{\Gamma}\to T\ \{\ S;\,x\ \}$ be a well-typed function, and
  $E$ be the set of environments $\{e' \mid e\smash{\stackrel{S}{\rightarrow}} e' \}$.
  If $(\mathbf{\Gamma}, \mathbf{A}, e)$ is well-formed,
  then so is $({x\colon T}, \mathbf{A}, e')$ for each $e'\in E$,
  and the sum of the probabilities of environments in $E$ is equal to the probability of $e$:
  \[
    \|e\|^2 = \sum_{e'\in E} \|e'\|^2.
  \]
\end{theorem}

\section{Uncomputation Semantics}\label{sec:unc-semantics}
This section concerns a second operational semantics for Qurts, called
uncomputation semantics, which uses reversible pebble games.
The uncomputation semantics compiles a Qurts program into a circuit graph, and the execution of the program is defined by the play of the pebble game.

The simulation semantics performed uncomputation immediately after a drop statement, but this strategy is not efficient in terms of time or space as a quantum circuit.
The uncomputation semantics, on the other hand, does not follow any fixed evaluation strategy.
Instead, a program is interpreted by a circuit graph and a pebble game play on it, where the latter is non-deterministic.
Thus, different pebbling strategies give rise to different execution sequences, with better strategies yielding more efficient quantum circuits.

Even though different strategies may induce different circuits,
\Cref{thm:uncomp-independent} below proves that the resulting quantum state is the same for any strategy.
We also show that simulation semantics and uncomputation semantics are equivalent in \Cref{thm:semantics-equiv}.

\subsection{Reversible Pebble Game for Qurts}\label{subsec:unc/pebble}
The pebble game we use for Qurts has some differences from the original reversible pebble game~\cite{Bennett-89-rev-pebble-tradeoff}. In the original, the rule of putting a pebble on $w$ from the pebbling state
\begin{tikzpicture}[baseline={(y.base)}]
  \node[pebbled] (x) at (0,0) {$v$};
  \node[vertex] (y) at (0.8,0) {$w$};
  \draw[->] (x) -- (y);
\end{tikzpicture}
corresponds to an application of a Controlled-$U_f$ gate to an auxiliary qubit initialized to $\ket0$ as follows,
where $U_f\colon \C^{2^n}\to\C^{2^n}$ is the unitary lifting a bijection $f \colon \{0,1\}^n\to\{0,1\}^n$.
\[
  \ket x \mapsto \ket x\ket0 \mapsto \ket x\ket{f(x)} \in \HH_{v}\otimes\HH_\text{auxiliary}
\]
On the other hand, to stay faithful to Rust's move semantics, the default execution should be interpreted as an application of $U_f$ to the same qubit, as $\ket x \mapsto \ket{f(x)} \in \HH_{v}$.
This has a corresponding pebbling state
\begin{tikzpicture}[baseline={(y.base)}]
  \node[vertex] (x) at (0,0) {$v$};
  \node[pebbled] (y) at (0.8,0) {$w$};
  \draw[->] (x) -- (y);
\end{tikzpicture}.
Also, we should still have a rule to introduce an auxiliary qubit.
Therefore, we replace the original rule of putting a pebble on $w$ by three new rules: \emph{init}, \emph{copy} and \emph{gate}.

The \emph{init} rule is for the special \emph{init vertices} in the circuit graph, corresponding to the state $\ket0$. The init rule allows to add or remove a pebble on any init vertex.
The \emph{copy} rule allows moving a pebble from the initial vertex to a vertex $v$ which is already pebbled.
Note that multiple pebbles on the same vertex are allowed.
This move corresponds to the application of the controlled-not gate, converting the state
$\ket{0i}\in \HH_\text{green}\otimes\HH_\text{yellow}$ to $\ket{ii}$.
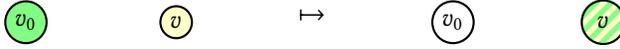
\begin{center}
  \begin{tabular}{ccc}
    \begin{tikzpicture}[baseline={(0,0)}]
      \node[pebbled]  (y) at (0,0) {$v_0$};
      \node[pebbled,fill=light-yellow] (x) at (2,0) {$v$};
    \end{tikzpicture}
     & $\mapsto$ &
    \begin{tikzpicture}[baseline={(0,0)}]
      \node[vertex]  (y) at (-2,0) {$v_0$};
      \begin{scope}
        \clip (0,0) circle (0.3);
        \foreach \i in {1,...,5}
          {
            \fill[light-green] (-0.5,.5) -- (-.5,-1.2+\i*.2) -- (1.2-\i*.2,.5) -- cycle;
            \fill[light-yellow]   (-0.5,.5) -- (-.5,-1.2+\i*.2+0.1) -- (1.2-\i*.2-.1,.5) -- cycle;
          }
        \draw[ultra thick] (0,0) circle (0.3) node {$v$};
      \end{scope}
    \end{tikzpicture}
  \end{tabular}
\end{center}
The rule \emph{gate} is for single target gates, which allows to move a pebble from $v$ to $w$ as in the following diagram.
The edge from $v$ to $w$ represents the target of the gate,
and the other edges from $v$ are the controls.
This rule only applies when all controls already have a pebble.
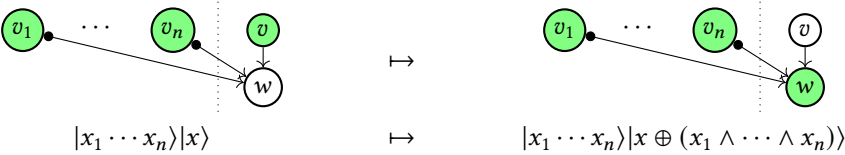
\begin{center}
  \begin{tabular}{ccc}
    \begin{tikzpicture}[baseline={(0,-0.5)},yscale=.75]
      \node[pebbled] (x) at (0,0) {$v$};
      \node[pebbled] (a) at (-1.2,0) {$v_n$};
      \node              at (-2.2,0) {$\cdots$};
      \node[pebbled] (b) at (-3.2,0) {$v_1$};
      \node[vertex] (y) at (0,-1) {$w$};
      \draw[dotted] (-0.6,0.5) -- (-0.6,-1.5);
      \draw[->] (x) -- (y);
      \draw[{Circle}->] (a) -- (y);
      \draw[{Circle}->] (b) -- (y);
    \end{tikzpicture}
     & $\mapsto$ &
    \begin{tikzpicture}[baseline={(0,-0.5)},yscale=.75]
      \node[vertex]  (x) at (0,0) {$v$};
      \node[pebbled] (a) at (-1.2,0) {$v_n$};
      \node              at (-2.2,0) {$\cdots$};
      \node[pebbled] (b) at (-3.2,0) {$v_1$};
      \node[pebbled] (y) at (0,-1) {$w$};
      \draw[dotted] (-0.6,0.5) -- (-0.6,-1.5);
      \draw[->] (x) -- (y);
      \draw[{Circle}->] (a) -- (y);
      \draw[{Circle}->] (b) -- (y);
    \end{tikzpicture}
    \\
    $\ket{x_1\cdots x_n}\ket{x}$
     & $\mapsto$ &
    $\ket{x_1\cdots x_n}\ket{x\oplus (x_1\wedge\cdots\wedge x_n)}$
  \end{tabular}
\end{center}
Together, these three rules can simulate the original rule of putting a pebble on $w$ by applying the init, copy, and gate rules in sequence.

A further new \emph{split} rule governs the quantum if statement.
Consider the following Qurts code:
\begin{center}
  \qurtsinline{qif x \{ [cnot](y,t)$ $ \} else \{ (y,[not](t))$ $ \}}
\end{center}
In a quantum if statement, it should not matter which branch is executed first.
Equivalently, the applications of two gates $CCX(x,y,t)$ and $C^\bot X(x,t)$ commute, where the latter is the negatively-controlled-X gate.
To express this commutativity as a circuit graph and a pebbling strategy, we allow \emph{splitting a pebble in half} and introduce a \emph{merge vertex with a guard}.

First, for each vertex $x$, we add a rule to split a pebble into half pebbles conditioned on $x=\ket0$ and $x=\ket1$ as in the left diagram below.
The idea is that a pebble conditioned on $x=\ket 1$ is used in the then branch,
and the other is used in the else branch of $\Qif x$.
In order to move these fragments, the vertex $x$ must be pebbled since the semantics of the movement uses $x$ as a control.
These two fragments can be moved independently, which enables executing the branches in parallel.

The graph on the right below represents the code above.
The vertex $v$ corresponds to the state of $t$, and vertices $x$, $y$ correspond
to the state of the variables $x$, $y$.
The vertices $v_0$ and $v_1$ correspond to the application of the gates in the branches of the quantum if statement, for which $w$ is the \emph{merge} vertex.
There are two incoming edges, from $v_0$ and $v_1$, through which only pebbles fragmented by the guards $x=\ket0$ respectively $x=\ket1$ can pass.
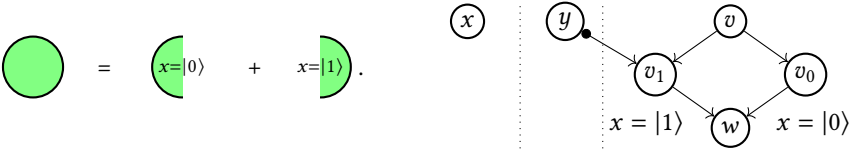
\begin{center}
  \begin{tikzpicture}[baseline={(0,-.1)}]
    \draw[thick,fill=light-green] (0,0) circle (0.4);
  \end{tikzpicture}
  \quad = \quad
  \begin{tikzpicture}[baseline={(0,-.1)}]
    \begin{scope}
      \clip (-.5,-.5) rectangle (0,.5);
      \fill[light-green] (0,0) circle (0.4);
      \draw[thick] (0,0) circle (0.4);
    \end{scope}
    \node at (0,0) {$\scriptstyle x=\ket0$};
  \end{tikzpicture}
  \quad+\quad
  \begin{tikzpicture}[baseline={(0,-.1)}]
    \begin{scope}
      \clip (.5,-.5) rectangle (0,.5);
      \fill[light-green] (0,0) circle (0.4);
      \draw[thick] (0,0) circle (0.4);
    \end{scope}
    \node at (0,0) {$\scriptstyle x=\ket1$};
  \end{tikzpicture}.
  \qquad
  \quad
  \begin{tikzpicture}[baseline={(0,.7)}]
    \node[vertex] (x) at (0,1.4) {$v$};
    \node[vertex] (y) at (-1,.7) {$v_1$};
    \node[vertex] (z) at (1,.7) {$v_0$};
    \node[vertex] (m) at (0,0) {$w$};
    \draw[->] (x) -- (y);
    \draw[->] (x) -- (z);
    \draw[->] (y) -- node[auto,swap]{$x=\ket1$} (m);
    \draw[->] (z) -- node[auto]{$x=\ket0$} (m);
    \node[vertex] (a) at (-2.2,1.4) {$y$};
    \draw[{Circle}->] (a) -- (y);
    \node[vertex] at (-3.5,1.4) {$x$};
    \draw[dotted] (-1.7,1.6) -- (-1.7,-.3);
    \draw[dotted] (-2.8,1.6) -- (-2.8,-.3);
  \end{tikzpicture}
\end{center}

This splitting rule also allows to handle a situation when a $\Drop$ appears inside a quantum if statement.
By this splitting rule, for example, one of the possible pebbling strategies for the following code says that this is equivalent to a Toffoli gate.
See \Cref{ax-sec:unc-semantics} for the detailed strategy.
\begin{center}
  \begin{minipage}{0.4\textwidth}
    \begin{lstlisting}[language=qurts,numbers=none]
      let (x,y) = [cnot](x,$\ket0$);
      let y' = qif z { y } else { drop y; $\ket{0}$ }
    \end{lstlisting}
  \end{minipage}
  \begin{minipage}{.4\textwidth}
    \hspace{1cm}\footnotesize
    $
      \ket{ij0} \mapsto \ket{iji} \mapsto \begin{cases}
        \ket{i1i} & \text{if } j=1, \\
        \ket{i00} & \text{if } j=0.
      \end{cases}
    $
  \end{minipage}
\end{center}

\subsection{Uncomputation Semantics}\label{subsec:unc-semantics}
We are ready to describe the Qurts uncomputation semantics, leaving the details to \Cref{ax-sec:unc-semantics}.

\subsubsection{Circuit Graphs}
Without loss of generality, we assume that all the lifted isometries $[c]$ are decomposed into $\ket{0}$ or single-target gates, denoted as $\Let y = \ket{0}$, $\Let y' = \neg y$, or $\Let y' = y \oplus (\neg^{i_{1}}x_1\wedge \cdots \wedge \neg^{i_{n}}x_{n})$. Note that negation applies the uncontrolled single-target gate $X$.
The following definition of circuit graphs and pebbles extends Unqomp~\cite{Unqomp_2021}.

\begin{definition}[Circuit Graph]
  A \emph{circuit graph} $G = (V,E)$ is a directed graph with a partition $V = V_\text{init}\sqcup V_\text{linear}\sqcup V_\text{gate} \sqcup V_\text{merge}$ of the vertices and a function $\loc\colon V \rightarrow \Loc_q$, such that:
  \begin{itemize}
    \item only a vertex in $v \in V_\text{gate}$ can have incoming edges with different locations;
    \item a vertex $v$ is in $V_\text{init}$ if and only if it does not have any incoming edge, and it is unique in the subgraph $G_l$ induced by $\{v\ |\ \loc(v) = l\}$ for each location $l$;
    \item a vertex $v \in V_\text{gate}$ has a unique incoming edge $w\to v$ such that $\loc(v) = \loc(w)$;
    \item a vertex $v \in V_\text{merge}$ has exactly two ordered incoming edges $e_0$ and $e_1$;
          the vertex $v$ has an associated pair $g = (w, \lft{a})$ of a vertex $w$ and a lifetime $\lft{a}$, called the \emph{guard} of $v$;
    \item a vertex $v$ in $V_\text{linear}$ is unique in $G_l$ for each $l$; and each edge with target in $V_\text{linear}$ has a formal conjunction of guards $g_1\wedge\cdots\wedge g_n$.
  \end{itemize}
\end{definition}

A vertex in $V_\text{init}$ has state $\ket0$, and a vertex in $V_\text{gate}$ corresponds to applying a single-target gate.
The vertex that corresponds to the X-gate applied to $v$ is denoted by $\neg v$.
The new parts of the definition are $V_\text{merge}$ and $V_\text{linear}$.

Vertices $V_\text{merge}$ model the quantum if statement according to the previous subsection.
The lifetime $\lft{a}$ of a guard represents the restriction that `passage is allowed only during this period', reflecting the intuition that variables can only be (un)computable during their lifetime.

Vertices in $V_\text{linear}$ model locations becoming linear when pebbled instead of affine.
Intuitively, this restricts attention to the forward closed subgraph of gates applied to linearly owned locations.
Vertices in $V_\text{linear}$ also have guards, because they may later need to play a role similar to $V_\text{merge}$ vertices when the location becomes linearly owned under quantum control.

\begin{definition}
  Let $G$ be a circuit graph.
  A \emph{guard} $g = (v,\lft{a})$ is a vertex $v$ with a lifetime $\lft{a}$,
  and a \emph{pebble} $p\{g_1\wedge\cdots\wedge g_n\}$ is a formal conjunction of guards with a label
  $p\in \mathrm{Label}$.
  The pebble $p\{\top\}$ is called \emph{whole} if $n=0$ and \emph{fragmented} otherwise.
  Pebbles may be added by
  $p\{g\wedge g'_1\wedge\cdots\wedge g'_n\} + p\{\neg g\wedge g'_1\wedge\cdots\wedge g'_n\} = p\{g'_1\wedge\cdots\wedge g'_n\}$,
  where the negation of a guard is $\neg(v,\lft{a}) = (\neg v,\lft{a})$.

  A \emph{pebbling of the circuit graph}
  ${\{\mathrm{peb}(v)\}}_{v\in{V}}$
  is an assignment of a set of pebbles to each $v\in V$.
\end{definition}

The label $p \in \mathrm{Label}$ identifies a pebble $p\{g_1\wedge\cdots\wedge g_n\}$ and represents a qubit.
The vertex $v$ in the guard indicates that the pebble is split by the condition $v=\ket1$.
For instance, say vertex $v$ has a pebble $q\{\top\}$.
When we have one fragmented pebble $p\{(v,\lft{a})\}$ on a vertex in state $\ket1$
and another fragmented pebble $p\{(\neg v,\lft{a})\}$ on a vertex in state $\ket0$,
then the whole quantum state is $\gamma_0\ket{00} + \gamma_1\ket{11} \in \HH_{q,p}$ ($\gamma_i\in\C$).
A pebble can have multiple guards representing qubit states in  nested quantum if statements.
Just like merge vertices, pebbles also have lifetimes:
a pebble $p\{(v_1,\lft{a}_1)\wedge\cdots\wedge (v_n,\lft{a}_n)\}$
can be moved only during the lifetimes $\lft{a}_1,\dots,\lft{a}_n$.

\subsubsection{Definition of Uncomputation Semantics}\label{subsubsec:unc-semantics-def}
Uncomputation semantics is a small-step operational semantics.
A program execution is represented by a relation $\rightarrow$ on tuples
$(\overline\Theta,\mathbf{A},G,(\ket\phi,s),\mathrm{peb})$
of \emph{threads} $\overline\Theta$, a lifetime preorder $\mathbf{A}$, a circuit graph $G$, a state of the memory $(\ket\phi,s)$, and a pebbling $\mathrm{peb}$.
There are two main rules, \scref{scheduler uncomputer} and \scref{scheduler exec}, that we now describe. Refer to \Cref{fig:rule/scheduler}, and to \Cref{ax-sec:unc-semantics} for formal definitions and proofs.

\paragraph{Uncomputer}
The judgement
$\mathbf{A}, G\models (\ket\phi,\mathrm{peb}) \Rrightarrow (\ket{\phi'}, \mathrm{peb}')$
for the play of the pebble game is called \emph{uncomputer execution}, and is used in the rule \scref{scheduler uncomputer}.
This rule can change the state of the pebbling, but not the circuit graph.

\paragraph{Thread}
Execution of $\Qif$ statements ideally evaluate both branches in parallel.
Therefore, we interpret a $\Qif$ statement as spawning two \emph{threads}, one for each branch.
Each thread has its own statement to execute, and the execution of a thread is presented by a relation
$(\Theta,\mathbf{A},(\ket\phi,s), G) \rightarrow (\overline{\Theta'},\mathbf{A}',(\ket{\phi'},s'), G')$
where $\Theta$ is a thread and $\overline{\Theta'}$ is a set of threads obtained by the execution.
Note that thread execution does not refer to the pebbling state,
and its main role is to extend the circuit graph.
However, there are some cases in which uncomputer execution and thread execution need to be synchronised, namely $\Endlft\alpha$, $x\As T$, $U(x)$, and $\Meas(x)$.
This synchronisation is accomplished by the judgement $G,\mathrm{peb}\models \Theta\colon\text{executable}$,
which checks whether the thread $\Theta$ is executable with pebbling state $\mathrm{peb}$.
Finally, the rule \scref{scheduler await} awakes the parent thread after termination of the child threads.

\begin{figure}
  {
    \footnotesize
    \begin{mathpar}
      \inferH{scheduler uncomputer}{
        \mathbf{A},G\models(\ket\phi, \mathrm{peb})
        \Rrightarrow
        (\ket{\phi'}, \mathrm{peb}')
      }{
        (\overline\Theta,\mathbf{A}, G, (\ket\phi, s), \mathrm{peb})
        \rightarrow
        (\overline\Theta, \mathbf{A}, G, (\ket{\phi'}, s), \mathrm{peb}')
      }

      \inferH{scheduler exec}{
        G, \mathrm{peb}\models \Theta\colon\text{executable} \\
        (\Theta,\mathbf{A},(\ket\phi,s), G)
        \rightarrow
        (\overline{\Theta'},\mathbf{A}',(\ket{\phi'},s'), G')
      }{
        ((\overline\Theta, \Theta),\mathbf{A}, G, (\ket\phi, s), \mathrm{peb})
        \rightarrow
        ((\overline\Theta + \overline{\Theta'}),\mathbf{A}', G', (\ket{\phi'}, s'), \mathrm{peb})
      }

      \inferH{scheduler await}{
        \mathrm{state}(t)=\mathbf{Await}(x,t_{1},\dots,t_{m})\\
        \mathrm{state}(t_i)=\mathbf{Terminated}(\mathbf{\Gamma}_i)\\
        (\Theta_t, G)
        \xrightarrow{(x,\mathbf{A},
          \Theta_{t_{1}},
          \dots,
          \Theta_{t_{n}}
          )}
        (\Theta'_t, G')
      }{
        ((\overline\Theta,\Theta_t,\overline{\Theta_{t_i}}),\mathbf{A}, G, (\ket\phi, s), \mathrm{peb})
        \rightarrow
        ((\overline\Theta,\Theta'_t),\mathbf{A}, G', (\ket{\phi'}, s'), \mathrm{peb})
      }
    \end{mathpar}
  }
  \Description{}
  \caption{Uncomputation semantics}
  \label{fig:rule/scheduler}
\end{figure}

\subsubsection{Properties}\label{subsubsec:unc-semantics-prop}
Uncomputation semantics is non-deterministic in two ways: the order of execution between threads, and between a thread and the uncomputer.
The next theorem states that the resulting quantum state is independent of the order of execution.

\begin{theorem}\label{thm:uncomp-independent}
  The uncomputation semantics is thread-safe
  in the sense that all the execution orders of threads and uncomputer give the same state of the classical and quantum memory.
\end{theorem}

Simulation semantics tries to deallocate qubits eagerly, as soon as the $\Drop$ statement is called. The next theorem shows that we can simulate the simulation semantics by a pebbling strategy clearing corresponding pebbles eagerly.

\begin{theorem}\label{thm:semantics-equiv}
  There is a pebbling strategy that simulates the simulation semantics. Therefore, uncomputation semantics and simulation semantics are equivalent.
\end{theorem}

The previous two theorems show that any pebbling strategy gives the desired result of the computation described in the simulation semantics.

\section{Related Work}\label{sec:silq}
This section compares and contrasts this article with other works, starting with the most relevant.

\emph{Silq}~\cite{Silq_2020} was the initial work attempting to create a type system for automatic uncomputation in quantum programming languages.
Silq has some special type annotations \silqinline{const} and \silqinline{qfree}.
The annotation \silqinline{const} is similar to immutable reference in Qurts, which can only be annotated to arguments.
In their paper~\cite{Silq_2020}, they mentioned that the \silqinline{const} is related to immutable reference in Rust,
but they concluded that the type system of Rust is not expressive enough to capture the subtlety of automatic uncomputation.
They mentioned that the annotation \silqinline{qfree} to functions,
which ensures that the function is made of classical boolean circuits,
is a feature that is not present in Rust, and together with the annotation \silqinline{const},
they provide a type system for automatic uncomputation.

The annotation \silqinline{qfree} in Silq and the type \qurtsinline{#'a qbit} in Qurts are similar
in that they are for safe uncomputation, but they are different in the following aspects.
First, Silq does not have a concept like lifetime, but this is not the only difference.
In Silq, the annotation \silqinline{qfree} is for entire functions,
and it ensures that the function body is made of classical boolean circuits.
Therefore, in the following Silq function,
there is no way to indicate that the second output qubit can be safely uncomputed but the first cannot.
In Qurts, the type \qurtsinline{#'static qbit} ensures that the second qubit can be safely uncomputed at any time
as in the right.
\begin{center}
  \begin{minipage}{0.31\textwidth}
    \begin{lstlisting}[language=silq]
      // B = type for qubit in Silq
      def f(x:B) : B x B { // not qfree
        a := H(x);
        return (a,0:B);
      }
    \end{lstlisting}
    \begin{center}
      \footnotesize Silq
    \end{center}
  \end{minipage}
  \qquad \qquad
  \begin{minipage}{0.43\textwidth}
    \begin{lstlisting}[language=qurts]
      fn f(mut x: qbit) -> (qbit, #'static qbit)
      {
        x.H();
        return (x, $\ket0$)
      }
    \end{lstlisting}
    \begin{center}
      \footnotesize Qurts
    \end{center}
  \end{minipage}
\end{center}
In addition, it is not possible to define the function \qurtsinline{forget} or \qurtsinline{reinitialize}
we presented in~\Cref{sec:ourlang} in Silq.
This problem is already known in~\cite{silq-github-issue}.
This is also because there is no way to annotate the type of qubits for inputs in Silq
to ensure that it can be safely uncomputed.

Silq has imperative features such as assignments, which make the language puzzling.
In \Cref{fig:code:silq-assignment}, we presented four functions in Silq.
The function \silqinline{f2} is obtained by inlining the variable
\silqinline{t} in line~4 to line~5 of \silqinline{f1}, so the two functions should be equivalent.
However, the former is rejected by the type checker.
Similarly, the function \silqinline{g2} is obtained by inlining the variable
\silqinline{t} in line~4 to line~5 of \silqinline{g1}, but only the function \silqinline{g1}
is accepted by the type checker.
This feature is not discussed in their paper~\cite{Silq_2020}
since they only discussed the core language Silq-Core that leaves out the imperative features.
As discussed, the automatic uncomputation feature has a strong imperative aspect,
which motivated the introduction of lifetime in type systems in Qurts.
We do not know the specific reason why the type checker of Silq behaves this way,
but it may be that Silq's type system is not expressive enough for the imperative features.
Qurts does not have this issue by definition
since all intermediate variables are assumed to be bound by $\Let$ in Qurts-core.
See also \Cref{ax-sec:silq}.
\\[-.5\baselineskip]

\begin{figure}
  \begin{minipage}{0.40\textwidth}
    \begin{lstlisting}[language=silq]
def f1(f: const B ! -> qfree B) {
    p := H(0:B);
    r := 0:B;
    t := f(p); // ERROR: variable 't' is not consumed
    if t {
        r := X(r)
    }
    return (p,r)
}
    \end{lstlisting}
  \end{minipage}\hspace{.8cm}
  \begin{minipage}{0.40\textwidth}
    \begin{lstlisting}[language=silq]
def f2(f: const B ! -> qfree B) {
    p := H(0:B);
    r := 0:B;

    if f(p) {
        r := X(r)
    }
    return (p,r)
}
    \end{lstlisting}
  \end{minipage}
  \begin{minipage}{0.40\textwidth}
    \begin{lstlisting}[language=silq]
def g1(f: B ! -> qfree B) {
    p := 0:B;
    r := 0:B;
    t := f(p);
    if t {
        r := X(r)
    }
    return r
}
    \end{lstlisting}
  \end{minipage}\hspace{.8cm}
  \begin{minipage}{0.4\textwidth}
    \begin{lstlisting}[language=silq]
def g2(f: B ! -> qfree B) {
    p := 0:B;
    r := 0:B;

    if f(p) { // ERROR: non-'lifted' quan- tum expression must be consumed
        r := X(r)
    }
    return r
}
    \end{lstlisting}
  \end{minipage}
  \Description{}
  \caption{Silq: assignments and uncomputation}
  \label{fig:code:silq-assignment}
\end{figure}

\emph{Unqomp}~\cite{Unqomp_2021} and \emph{Reqomp}~\cite{Reqomp_2024} are algorithms
for automatic quantum uncomputation.
Unqomp is an algorithm that prevents recomputation by delaying uncomputations to optimise time complexity,
and reqomp is an algorithm that optimises number of qubits by allowing recomputation.
Both algorithms are based on \emph{circuit graphs}.
They first construct a circuit graph from a quantum circuit which is known to be \texttt{qfree},
a circuit that can be represented as a classical Boolean function,
and then construct another circuit graph with the uncomputation.
In Reqomp, it is mentioned that pebble game is not suitable for the uncomputation problem,
because the graph constructed is not expressive enough to represent non-qfree gates or repeatedly used values.
The pebble game for Qurts overcomes these limitations:
Our pebble game is not based on the dependency graph as in the original pebble game,
but rather on the circuit graph defined in the paper Unqomp, which is more expressive.
We can also represent repeatedly used values by using the same vertices in the graph,
or have cyclic dependencies in the graph discussed being problematic in Unqomp.
However, our type system asserts that the uncomputation is possible in some way even in such cases.
\\[-.5\baselineskip]

\emph{Qrisp}.
The Unqomp algorithm is incorporated in Qrisp~\cite{Qrisp_2023}, a framework of Python modules and language extensions that replaces gates and qubits with functions and variables governed by a linear type discipline, and that compiles to quantum circuits.

The lazy approach of Unqomp saves computation time at the expense of a larger qubit overhead, which is often not desirable. 
Qrisp therefore extends \qurtsinline{qfree}-ness to the concept of \emph{permeability}, which allows more flexible resource management. The Qrisp compiler can automatically verify the correctness of the uncomputation, but this is no longer based on a type system. Instead, the compiler checks via operational semantics whether the unitary matrices in question are permeable. This compile-time simulation can be very costly. The type discipline of Qurts allows the flexibility of different pebbling strategies while maintaining compilation performance.
\\[-.5\baselineskip]

\emph{Modular synthesis}.
In parallel to and independent of this paper, the very recent work~\cite{venevetal:modularsynthesis} has a similar goal. First, it presents an algorithm for uncomputation that benchmarks favourably against Unqomp. An interesting question is to which pebbling strategy this algorithm corresponds. Its second contribution is an intermediate representation that tracks data used for the uncomputation. It uses an effect system, similar to Silq's \texttt{qfree} annotations, that gives sufficient conditions for the uncomputation algorithm to be applicable.
In contrast, Qurts tracks similar data using a type system that is a first class citizen, about which formal properties are proven, that is directly available to the programmer, and that is clean in the sense that it only extends the Rust type system with \qurtsinline{#'a}. 
\\[-.5\baselineskip]

\emph{ReQWIRE}~\cite{ReQWIRE_2018} is a similar compiler, whose main feature is that it has been verified with the interactive theorem prover Coq.
It allows discarding qubits whose quantum state is verified to be $\ket{0}$ or $\ket{1}$ on a linear type system, so its goal is slightly different from automatic uncomputation.
\\[-.5\baselineskip]

\emph{Twist} \cite{twist_2022} 
is a quantum programming language whose goal is to guarantee purity (the opposite of entanglement) by type discipline. This has impact on uncomputability, but that is not Twist's main aim. The system enforces the type discipline using a combination of static analysis and runtime verification.
\\[-.5\baselineskip]

\emph{Proto Quipper}~\cite{FKRS_RC2020_Dependent_ProtoQuipper} is a programming language and compiler that addresses uncomputation with a monadic type system that syntactically hides the discarded qubits from the programmer, but does not provide any guarantee or verification of correctness of the program.
\\[-.5\baselineskip]

Similarly, \emph{ScaffCC}~\cite{ScaffCC_framework,ScaffCC_scalable} and \emph{SQUARE}~\cite{SQUARE_2020} are both compiler frameworks for quantum computation, that perform some analysis to optimise for quantum uncomputation, but unlike Qurts they are not guaranteed at the type level.
\section{Future Work}\label{sec:future}
We have shown that affine types with lifetime restrictions can usefully model quantum computation with automatic uncomputation.
One might similarly look for other kinds of computation modeled by lifetime-restricted affine types than quantum computation.
For example, such types could represent a situation where the garbage has to be taken out by the time the garbage collector arrives.
Alternatively, swapping the roles of affine and linear types could model situations where garbage cannot be disposed of until the garbage bins are in place, or where people must remain imprisoned at a location until a certain time has elapsed.

Implementing a Qurts compiler is left to future work.
From a syntactic perspective, the most interesting implementation challenges not yet present in Qurts-core include recursion, mutable variables and algebraic data types.
Their interaction with the type $\Own^\alpha T$ is still yet to be seen.

Since the optimisation problem of the original reversible pebble game is known to be PSPACE-complete~\cite{ChanLNV_15_RevPebble_PSPACE}, we expect that our pebble game for Qurts is also PSPACE-complete.
If that is indeed the case, there could be a reduction from the pebble game to the problem of Quantified Boolean Formulas (QBF), which is a well-known PSPACE-complete problem.
Thus, instead of using specific heuristic algorithms made for uncomputation, the optimisation power of a generic QBF solver could be used as a backend for the Qurts compiler to optimise the uncomputation.

One advanced way to optimise quantum uncomputation that has gained attention is so-called spooky pebble games~\cite{Gidney_2019_Spooky_blog}.
This is a pebble game that models a way to uncompute quantum states using measurement, and is thus not reversible.
Spooky pebble games have been shown to model more efficient computation than the original reversible pebble game in some situations~\cite{Quist_2024_Spooky_tradeoffs,QuistLaarman_RC2023_Spooky_optimizing_quantum,Kornerup_2024_Spooky_tight_bound}.
Since our uncomputation semantics uses reversible pebble games as a backend, it can be easily extended spooky pebble games. It would be interesting to see much this can optimise Qurts performance.

The uncomputation semantics exhibited strong parallelism in executing two branches of \texttt{qif} statements.
Both branches can not only share a reference $\&T$, they can also have write access to the same location. In classical programming, this normally produces race conditions or deadlock.
However, Qurts' uncomputation semantics are thread safe, as we have shown, because both branches are controlled by a qubit.
An interesting area of further study is to analyse this behaviour in terms of ownership.

\begin{acks}
  The authors would like to thank the anonymous referees for their valuable comments and helpful suggestions.
  We also sincerely thank Ryo Wakizaka, Yusuke Matsushita, Sam Lindley, and other colleagues in Tokyo, Chiba, Kyoto, and Edinburgh for their valuable discussions.
\end{acks}

\clearpage
\appendix
\section{Grover's Algorithm in Qurts}\label{ax-sec:grover}
In \Cref{subsec:lang-example}, we presented a simple example of Grover's algorithm in Qurts. Some syntactic sugar was used to make the code more readable, and it is mostly similar to Rust. In this section, we explain the syntactic sugar in detail for the sake of completeness.

\begin{center}
  \begin{minipage}{0.4\textwidth}
    \begin{lstlisting}[language=qurts]
      fn grover(
        x &mut qbit,
        y:&mut qbit, 
        z:&mut qbit
      ) {
        x.H(); y.H(); z.H();  
        for _ in $0$..$2$ {
          qif &oracle(&x, &y, &z) { phase($\pi$) }
          grover_diffusion(x, y, z)
        }
      }
    \end{lstlisting}
  \end{minipage}
\end{center}

In line~6, we used the notation \qurtsinline{x.H()}.
Here, we assume the function $H$ (Hadamard gate) is defined as a function that takes \qurtsinline{&mut qbit} and returns \qurtsinline{()}.
This is Rust-like method call syntax, which is equivalent to \qurtsinline{H(&mut x)}.

In line~7, the notation \qurtsinline{for _ in 0..2} represents a loop that iterates twice.

In line~8, we used the notation \qurtsinline{qif &oracle(&x, &y, &z) \{ phase(}$\pi$\qurtsinline{) \}}.
This is syntantic sugar for the following code (which still uses some syntantic sugar):
\begin{center}
  \begin{minipage}{0.7\textwidth}
    \begin{lstlisting}[language=qurts,numbers=none]
      let oracle_result: #'a qbit = oracle(&'a x, &'a y, &'a z);
      let oracle_result_ref: &'a #'a qbit = &oracle_result;
      let _ = qif oracle_result_ref { -1() } else { () };
      drop(oracle_result_ref);
      drop(oracle_result);
    \end{lstlisting}
  \end{minipage}
\end{center}
The function \qurtsinline{oracle} is assumed to be defined as a function that takes three \qurtsinline{&'a qbit}s and returns a \qurtsinline{#'a qbit} with lifetime \qurtsinline{'a}.
The return value of \qurtsinline{oracle} is used in the condition of the \qurtsinline{qif} statement, and then it is dropped immediately afterward.
The function \qurtsinline{phase} is assumed to be a unitary operation that takes 0 qubits, meaning it is simply a complex number multiplication.
Since the return value of \qurtsinline{qif} is \qurtsinline{()}, we can safely ignore it or assign it to \qurtsinline{_}.

In line~9, the notation \qurtsinline{grover_diffusion(x, y, z)} is used.
The types of \qurtsinline{x}, \qurtsinline{y}, and \qurtsinline{z} are \qurtsinline{&mut qbit}.
This function returns \qurtsinline{()}, which does not need to be assigned to any variable.

Because we are using for-loops, mutable references are required in the loop body, which is a feature that Qurts-core does not support.
However, once the loop is unrolled, the code can be translated to Qurts-core without any issues.
This is possible because we can redefine the variables after each use by changing the function signatures from \qurtsinline{fn (&mut T) -> ()} to \qurtsinline{fn (T) -> T} as shown below.
\begin{center}
  \begin{minipage}{0.5\textwidth}
    \qurtsinline{let x = H(x);}\\
    \qurtsinline{let x,y,z = grover_diffusion(x,y,z);}
  \end{minipage}
\end{center}

\section{Type System}\label{ax-sec:type-system}
In this section, we present the detailed type system for Qurts.
For simplicity, we assume that each variable is declared only once in the program, and variable shadowing is not allowed.
Let $\lft{a}$ be a lifetime, which is either $\bot$, $\top$ or a lifetime variable $\alpha$.
We write $\lft{a}\in \mathbf{A}$ to mean that $\lft{a}$ is a lifetime in the lifetime preorder $\mathbf{A}$ satisfying $\bot < \lft{a}$.
We also assume that $\&^\bot$, which it is meaningless, does not appear in the program.

\paragraph{Subtyping}
The rules for subtyping are defined in \Cref{fig:rule/subtyping}.
The rule \scref{subty shorten} states that the lifetime of a pointer can be shortened via coercion.
The rules \scref{subty reborrow}, \scref{subty double affine}, \scref{subty borrow affine}, and \scref{subty affine borrow} are the rules for subtyping of a pointer of another pointer.
When you have the same pointer type, you can merge them into a single pointer type with a shorter lifetime
(\scref{subty reborrow}, \scref{subty double affine}).
Also, when you have a pointer which is bounded by both $\&$ and $\Own$, then it can be squashed to an immutable reference $\&$
(\scref{subty borrow affine}, \scref{subty affine borrow}).
The rest of the rules are for the subtyping of tuples.
The notation $\Qbit^n$ is an abbreviation of $\Qbit \times \dots \times \Qbit$.
The last rule \scref{subty unit} states that the pointer to the unit type can be thought as the unit type itself.

\begin{figure}[tbh]
  \raggedright \textbf{Subtyping:}
  \qquad$P$ is a pointer type $\Own$ or $\&$.
  \hfill
  \fbox{
    $\mathbf{A} \vdash T_1 \leq T_2$
  }
  \footnotesize\begin{mathpar}
    \inferH{subty shorten}{
      \lft{b}\leq \lft{a}\in \mathbf{A}
    }{
      \mathbf{A} \vdash P^\lft{a}T \leq P^\lft{b}T
    }

    \inferH{subty pointer}{
      \mathbf{A} \vdash T_1 \leq T_2
    }{
      \mathbf{A} \vdash P^\lft{a}T_1 \leq P^\lft{a}T_2
    }

    \inferH{subty reborrow}{
      \alpha, \beta \geq \gamma
    }{
      \&^\alpha \&^\beta\ T \leq \&^\gamma\ T
    }

    \inferH{subty double affine}{
      \lft{a}, \lft{b} \geq \lft{c}
    }{
      \Own^\lft{a} \Own^\lft{b} T \leq \Own^\lft{c} T
    }

    \inferH{subty borrow affine}{}{
      \&^\alpha \Own^\lft{a} T \leq \&^\alpha T
    }

    \inferH{subty affine borrow}{}{
      \Own^\lft{a} \&^\alpha T \leq \&^\alpha T
    }

    \inferH{subty tuple}{
      \mathbf{A} \vdash T_0 \leq T'_0 \\
      \mathbf{A} \vdash T_1 \leq T'_1
    }{
      \mathbf{A} \vdash T_0 \times T_1 \leq T'_0 \times T'_1
    }

    \inferH{subty qbits}{}{
      \Qbit^n \times \Qbit^m \cong \Qbit^{n+m}
    }

    \inferH{subty ptr tuple}{}{
      P^\lft{a}T_0 \times P^\lft{a} T_1 \cong P^\lft{a} (T_0 \times T_1)
    }

    \inferH{subty unit}{}{
      P^\lft{a}() \cong ()
    }
  \end{mathpar}
  \Description{}
  \normalsize\caption{Subtyping rule}
  \label{fig:rule/subtyping}
\end{figure}

\paragraph{Copy} The rules for the \qurtsinline{Copy} trait, the class of types which can be implicitly copied, are defined in \Cref{fig:rule/Copy}.
The rule \scref{cpy bool} and \scref{cpy borrow} state that boolean and immutable references are the basic copyable types.
The rule \scref{cpy own} states that the type $\Own^\lft{a} T$ is copyable only when $T$ is copyable.
For instance, $\Own^\alpha \Qbit$ is not copyable, but $\&^\alpha \Qbit$ is copyable.
The rest of the rules concern tuples and the unit type.
Note that the semantics of copy of $\&^\alpha\Qbit$ in Qurts does not copy any quantum state but only copies the reference to the quantum state.

\begin{figure}[tbh]
  \raggedright \textbf{Copy:}
  \hfill
  \fbox{
    $T\colon \COPY$
  }
  \footnotesize\begin{mathpar}
    \inferH{cpy bool}{}{
      \Bool \colon \COPY
    }

    \inferH{cpy borrow}{}{
      \&^\alpha T \colon \COPY
    }

    \inferH{cpy own}{
      T \colon \COPY
    }{
      \Own^\lft{a} T \colon \COPY
    }

    \inferH{cpy unit}{}{
      ()\colon \COPY
    }

    \inferH{cpy tuple}{
      T_0\colon \COPY \\
      T_1\colon \COPY
    }{
      T_0 \times T_1 \colon \COPY
    }
  \end{mathpar}
  \Description{}
  \normalsize\caption{Copy trait}
  \label{fig:rule/Copy}
\end{figure}

\paragraph{Drop}
The rule for the \qurtsinline{Drop} trait, the class of affine types which can be dropped, is defined in \Cref{fig:rule/Drop}.

\paragraph{Purely Quantum}
The rule $\PQ$ in \Cref{fig:rule/PQ} is defined to prohibit the appearance of measurements under the quantum control, or the creation of superposition of classical data.
However, in fact, it might not be essential to prohibit superposition of immutable reference.
For example, consider the following code, which is prohibited in Qurts-core.
\begin{center}
  \begin{minipage}{0.7\textwidth}
    \begin{lstlisting}[language=qurts]
      // q: &'a qbit
      // r1, r2: &'b qbit
      // z: #'0 qbit
      let r: #'a &'b qbit = qif q { r1 } else { r2 };
      let z' = qif r { U(z) } else { z };
    \end{lstlisting}
  \end{minipage}
\end{center}
It would be natural to consider this code as equivalent to the following code.
\begin{center}
  \begin{minipage}{0.7\textwidth}
    \begin{lstlisting}[language=qurts,numbers=none]
      let z' = qif (q && r1) || ($\neg$q && r2) { U(z) } else { z };
    \end{lstlisting}
  \end{minipage}
\end{center}
This, in turn, is equivalent to the following.
\begin{center}
  \begin{minipage}{0.7\textwidth}
    \begin{lstlisting}[language=qurts,numbers=none]
      let z' = qif q {
        qif r1 { U(z) } else { z }
      } else {
        qif r2 { U(z) } else { z }
      };
    \end{lstlisting}
  \end{minipage}
\end{center}
In terms of the lifetime of quantum information, it seems that
if we allow to cast $\Own^\alpha \&^\beta$ into $\&^\gamma$ with some $\gamma$ satisfying $\gamma \leq \alpha,\beta$,
then it would represent the correct lifetime of the information held by $r$, making the code above valid.
However, the downside is that it further complicates the semantics of the language
if we allow immutable references to appear in return value for \qurtsinline{qif}.
Therefore, we decided to prohibit this, leaving this feature for future work.

\paragraph{Typing Rules}
The rules for the typing judgement are defined in
\Cref{fig:rule/ty-expr,fig:rule/ty-stmt,fig:rule/ty-fn}.

\paragraph{Typing Expressions}
\Cref{fig:rule/ty-expr} is the typing rule for expressions.
A type judgement has the form
$e . T \colon (\mathbf{\Gamma}, \mathbf{A}) \xrightarrow{\Pi,f} \mathbf{\Gamma}'$,
where $e$ is an expression, $T$ is the type of $e$,
$\mathbf{\Gamma}$ and $\mathbf{A}$ are the preceding context and lifetime preorder,
and $\mathbf{\Gamma}'$ is the succeeding context.
The superscript $\Pi,f$ represents the functions already defined and the function name where the expression $e$ appears.

\begin{figure}[tb]
  \raggedright \textbf{Typing Expr:}
  \hfill
  \fbox{
    $e . T \colon (\mathbf{\Gamma}, \mathbf{A}) \xrightarrow{\Pi,f} \mathbf{\Gamma}'$
  }
  \footnotesize
  \begin{mathpar}
    \inferH{expr var}{}{
      x . T \colon (\mathbf{\Gamma} + \{ x \colon T \}, \mathbf{A}) \xrightarrow{\Pi,f} \mathbf{\Gamma}
    }

    \inferH{expr const bool}{
      b \in \{ \False, \True \}
    }{
      b .\Own^\top \Bool \colon (\mathbf{\Gamma}, \mathbf{A}) \xrightarrow{\Pi,f} \mathbf{\Gamma}
    }

    \inferH{expr unit}{}{
      () . () \colon (\mathbf{\Gamma}, \mathbf{A}) \xrightarrow{\Pi,f} \mathbf{\Gamma}
    }

    \inferH{expr tuple}{}{
      (x_0, x_1) . T_0 \times T_1 \colon
      (\mathbf{\Gamma} + \{ x_0 \colon T_0, x_1 \colon T_1 \}, \mathbf{A}) \xrightarrow{\Pi,f} \mathbf{\Gamma}
    }

    \inferH{expr copy}{
      T\colon \COPY \\ x \colon T \in \mathbf{\Gamma}
    }{
      \Copy x . T \colon (\mathbf{\Gamma}, \mathbf{A}) \xrightarrow{\Pi,f} \mathbf{\Gamma}
    }

    \inferH{expr measure}{}{
      \Meas(x) . \Own^\top\Bool \colon (\mathbf{\Gamma} + \{ x \colon \Own^\bot\Qbit \}, \mathbf{A}) \xrightarrow{\Pi,f} \mathbf{\Gamma}
    }

    \inferH{expr unitary}{
      U \in \mathbf{Unitary}(n)
    }{
      U(x) . \Own^\bot \Qbit^n \colon (\mathbf{\Gamma} + \{ x \colon \Own^\bot \Qbit^n \}, \mathbf{A})
      \xrightarrow{\Pi,f} \mathbf{\Gamma}
    }

    \inferH{expr lifted}{
      c \colon {\{0,1\}}^n \to {\{0,1\}}^m; \text{injection}
    }{
      [c](x) . \Own^\lft{a}\Qbit^m \colon
      (\mathbf{\Gamma} + \{ x \colon \Own^\lft{a}\Qbit^n \}, \mathbf{A})
      \xrightarrow{\Pi,f} \mathbf{\Gamma}
    }

    \inferH{expr function}{
    \Sigma_{\Pi,g} = \langle
    \alpha'_0,\dots,\alpha'_{m-1} \mid \alpha'_{a_0} \leq \alpha'_{b_0},\dots,\alpha'_{a_{l-1}} \leq \alpha'_{b_{l-1}}
    \rangle
    (x'_0\colon T'_0,\dots,x'_{n-1}\colon T'_{n-1})
    \rightarrow T'_n \\
    \forall j \in \{ 0,\dots,l-1 \},\ \alpha_{a_j}\leq \alpha_{b_j} \in \mathbf{A} \\
    \forall i \in \{ 0,\dots,n   \},\ T_i = T'_i[\alpha_0/\alpha'_0,\dots,\alpha_{m-1}/\alpha'_{m-1}]
    }{
    g\langle \alpha_0,\dots,\alpha_{m-1} \rangle (x_0,\dots,x_{n-1}) . T_n \colon
    (\mathbf{\Gamma} + \{ x_0 \colon T_0,\dots,x_{n-1} \colon T_{n-1} \}, \mathbf{A}) \xrightarrow{\Pi,f} \mathbf{\Gamma}
    }

    \inferH{expr classical if}{
      x\colon \Bool \in \mathbf{\Delta} \\
      (\mathbf{\Gamma}, \mathbf{A}) \vdash^{\Pi, f} B_i \colon T
    }{
      \If\ x\ B_t\ \Else\ B_f. T \colon
      (\mathbf{\Gamma}+\mathbf{\Delta}, \mathbf{A}) \xrightarrow{\Pi,f} (\mathbf{\Delta}, \mathbf{A})
    }

    \inferH{expr quantum if}{
      x\colon \&^\alpha \Qbit \in \mathbf{\Delta} \\
      \alpha \in \mathbf{A} \\
      (\mathbf{\Gamma}, \mathbf{A}) \vdash^{\Pi,f} B_i \colon T \\\\
      B_i \colon \PQ \\
      T\colon\PQ
    }{
      \Qif\ x\ B_\ket{1}\ \Else\ B_\ket{0}.\Own^\alpha T \colon
      (\mathbf{\Gamma}+\mathbf{\Delta}, \mathbf{A}) \xrightarrow{\Pi,f} (\mathbf{\Delta}, \mathbf{A})
    }
  \end{mathpar}
  \Description{}
  \normalsize\caption{Typing rule for expressions}
  \label{fig:rule/ty-expr}
\end{figure}

In Qurts, lifted functions $[c]$ are treated as special isometry operators.
For any unitary operator $U$, the rule \scref{expr unitary} states that
$U$ takes qubits $x$ and returns qubits $U(x)$,
where we do not know how to uncompute the returns.
On the other hand, the rule \scref{expr lifted} states that
$[c]$ also takes qubits as input and returns qubits $[c](x)$,
but we know how to uncompute the returns if we also know how the inputs could be uncomputed.

The rule \scref{expr function} replicates the rule for the COR rule in~\cite{rusthorn}, which formalises the core of Rust.
The assumptions for the rule looks a bit complicated, but it is simply doing the following. 
First, it checks the lifetime constraints in the function signature are satisfied in the caller's context.
If so, it then substitutes the lifetime variables in the signature of the function $g$ with those in the context,
and finally it type-checks the arguments and the return value of the function $g$.

The two rules for conditional statements \scref{expr classical if} and \scref{expr quantum if}
are for the classical if and quantum if statements, respectively.
For the conditioning variable $x$ in these statements,
the classical if takes a boolean, and the quantum if takes an immutable reference to a qubit.
For the quantum if statement, the type for the conditioning variable $x$ and return value have to be chosen carefully.
In fact, the type of the variable $x$ in the quantum if statement has to be $\&^\alpha \Qbit$,
and it influences the return type -- it has to be $\Own^\alpha T$ rather than $T$.
This is because, to perform uncomputation to a qubit $y$ which undergoes a unitary operator in the quantum if statement, we need to apply the inverse unitary operator controlled by the conditioning qubit $x$, which has to live longer than the target qubit $y$.
Also, since we do not want a measurement to appear in the quantum if statement, or create superposition of a classical data, we restrict the bodies $B_i$ and the type $T$ to be purely quantum.

\paragraph{Typing Statements}
\Cref{fig:rule/ty-stmt} shows the typing rules for statements.
The judgement has the form $S \colon (\mathbf{\Gamma}, \mathbf{A}) \xrightarrow{\Pi,f} (\mathbf{\Gamma}', \mathbf{A}')$,
where $S$ is a statement, $\mathbf{\Gamma}$ and $\mathbf{A}$ are the preceding context and lifetime preorder,
and $\mathbf{\Gamma}'$ and $\mathbf{A}'$ are the succeeding context and lifetime preorder.
Again, the superscript $\Pi,f$ is the function environment and the function name where the statement $S$ appears.

\begin{figure}[tbh]
  \raggedright \textbf{Typing Stmt}
  \hfill
  \fbox{
    $S \colon (\mathbf{\Gamma}, \mathbf{A}) \xrightarrow{\Pi,f} (\mathbf{\Gamma}', \mathbf{A}')$
  }
  \footnotesize
  \begin{mathpar}
    \inferH{stmt noop}{}{\Noop\colon (\mathbf{\Gamma},\mathbf{A})\xrightarrow{\Pi, f}(\mathbf{\Gamma},\mathbf{A})}

    \inferH{stmt composition}{
      S_0 \colon (\mathbf{\Gamma}, \mathbf{A}) \xrightarrow{\Pi,f} (\mathbf{\Gamma}', \mathbf{A}') \\
      S_1 \colon (\mathbf{\Gamma}', \mathbf{A}') \xrightarrow{\Pi,f} (\mathbf{\Gamma}'', \mathbf{A}'')
    }{
      S_0;S_1 \colon (\mathbf{\Gamma}, \mathbf{A}) \xrightarrow{\Pi,f} (\mathbf{\Gamma}'', \mathbf{A}'')
    }

    \inferH{stmt new lifetime}{}{
    \Newlft\alpha
    \colon (\mathbf{\Gamma}, \mathbf{A}) \xrightarrow{\Pi,f}
    (\mathbf{\Gamma}, (A, R + \{ \alpha \leq \beta \mid \beta \in A_{ex,\Pi,f} \}))
    }

    \inferH{stmt end lifetime}{
      \alpha \not\in A_{ex,\Pi,f} \\
      \alpha \text{ is minimal in }\mathbf{A}-\{\bot\}\\
      \&^\alpha \text{ does not apper in } \Gamma
    }{
      \Endlft\alpha \colon (\mathbf{\Gamma}, \mathbf{A}) \xrightarrow{\Pi,f}
      (\defrost_\alpha(\mathbf{\Gamma}), \mathbf{A}-\alpha)
    }

    \defrost_\alpha(x\colon^\mathbf{a}T) \coloneqq
    \begin{cases}
      x\colon T            & (\mathbf{a}= \dagger\alpha) \\
      x\colon^\mathbf{a} T & (\text{otherwise})
    \end{cases}

    \inferH{stmt lft ineq}{
      \alpha, \beta \not\in A_{ex,\Pi,f} \\
      \mathbf{A}' \text{ is the smallest preorder including } \mathbf{A} \cup \{ \alpha \leq \beta \}
    }{
      \alpha \leq \beta \colon (\mathbf{\Gamma}, \mathbf{A}) \xrightarrow{\Pi,f}
      (\mathbf{\Gamma}, \mathbf{A}')
    }

    \inferH{stmt borrow}{
      \alpha \not\in A_{ex,\Pi,f}\\
      \forall \gamma \in \{ \alpha\ |\ \&^\alpha \text{ appear in }  T\},\ \alpha \leq \gamma \in \mathbf{A}\\
    }{
      \Let y = \&^\alpha x \colon
      (\mathbf{\Gamma} + \{ x \colon T \}, \mathbf{A})
      \xrightarrow{\Pi,f}
      (\mathbf{\Gamma} + \{ y \colon \&^\alpha T,\ x \colon^{\dagger\alpha} T \}, \mathbf{A})
    }

    \inferH{stmt expr}{
      e . T \colon (\mathbf{\Gamma}, \mathbf{A}) \xrightarrow{\Pi} \mathbf{\Gamma}'\\
    }{
      \Let y = e \colon (\mathbf{\Gamma}, \mathbf{A}) \xrightarrow{\Pi,f} (\mathbf{\Gamma}' + \{ y \colon T \}, \mathbf{A})
    }

    \inferH{stmt proj}{}{
      \Let (y_0,y_1) = x \colon (\mathbf{\Gamma} + \{ x \colon T_0 \times T_1 \}, \mathbf{A})
      \xrightarrow{\Pi,f} (\mathbf{\Gamma} + \{ y_0 \colon T_0,\ y_1 \colon T_1 \}, \mathbf{A})
    }

    \inferH{stmt drop}{
      \mathbf{A} \vdash T\colon \DROP
    }{
      \Drop x \colon (\mathbf{\Gamma} + \{ x \colon T \}, \mathbf{A}) \xrightarrow{\Pi,f} (\mathbf{\Gamma}, \mathbf{A})
    }

    \inferH{stmt coercion}{
      \mathbf{A}\vdash U\leq T
    }{
      x \As T \colon (\mathbf{\Gamma} + \{ x \colon U \}, \mathbf{A}) \xrightarrow{\Pi,f} (\mathbf{\Gamma} + \{ x \colon T \}, \mathbf{A})
    }
  \end{mathpar}
  \Description{}
  \normalsize\caption{Typing rule for statements}
  \label{fig:rule/ty-stmt}
\end{figure}

The rules \scref{stmt new lifetime}, \scref{stmt end lifetime}, and \scref{stmt lft ineq} are for the lifetime management.
The rule \scref{stmt new lifetime} introduces a new lifetime $\alpha$ into the lifetime preorder $\mathbf{A}$.
The newly introduced lifetime $\alpha$ cannot be carried out of the function, and it has to be ended before the function returns.
This is expressed by the constraint $\alpha\leq\beta$ for each $\beta\not\in A_{ex,\Pi,f}$
where the set $A_{ex,\Pi,f}$ is the set of lifetimes that are external to the function $f$.
The rule \scref{stmt end lifetime} declares the end of lifetime $\alpha$.
The assumption `$\alpha$ is minimal' in $\mathbf{A}-\{\bot\}$ is to prevent the lifetime $\alpha$ from ending before some $\beta$ which is smaller than $\alpha$.
Also, we assume all the immutable references whose lifetime has ended are dropped before the lifetime ends.
After the end of lifetime $\alpha$, the lifetime $\alpha$ is removed from the lifetime preorder $\mathbf{A}$,
and we defrost all the variables which had been frozen by the lifetime $\alpha$ to regain their write permission.
The rule \scref{stmt lft ineq} adds a new constraint $\alpha\leq\beta$ to the lifetime preorder $\mathbf{A}$.

The rule \scref{stmt borrow} is for borrowing a value.
When we have a variable $x$ of type $T$ in the context, we can borrow the value of $x$ by taking a reference to $x$.
At the same time, we freeze the value of $x$ in the typing context to prevent the value from being modified while the reference is alive.
Note that we cannot create references that live for too long:
the owner has to be alive while the reference is alive ($\alpha \not\in A_{ex,\Pi,f}$),
and a reference to a reference cannot live longer than the original ($\alpha \leq \gamma \in \mathbf{A}$).

The rest of the rules are straightforward.
The rule \scref{stmt expr} is for typing an expression, the rule \scref{stmt proj} is for decomposing a tuple,
the rule \scref{stmt drop} is for dropping a value, and the rule \scref{stmt coercion} is for type coercion.

\paragraph{Typing Blocks and Functions}
\Cref{fig:rule/ty-fn} defines typing rules for blocks and functions.

\begin{figure}[tbh]
  \raggedright \textbf{Typing Block, Function, Program:}
  \footnotesize
  \begin{mathpar}
    \inferH{typing block}{
      S\colon (\mathbf{\Gamma}, \mathbf{A}) \xrightarrow{\Pi,f} (x\colon T, \mathbf{A})\\
    }{
      (\mathbf{\Gamma},\mathbf{A}) \vdash^{\Pi,f} \{\ S;\ x\ \} \colon T
    }

    \inferH{typing fn}{
      F = \Fn\ f\ \Sigma\ B\\
      \Sigma = \langle \alpha_0,\dots,\alpha_{m-1} \mid \alpha_{a_0} \leq \alpha_{b_0},\dots,\alpha_{a_{l-1}} \leq \alpha_{b_{l-1}} \rangle
      \ \mathbf{\Gamma} \rightarrow T\\
      \mathbf{A}\colon\text{ the smallest preorder on } \{ \alpha_i \} \text{ including } \{ \alpha_{a_j} \leq \alpha_{b_j} \}\\
      (\mathbf{\Gamma},\mathbf{A}) \vdash^{\Pi,f} B\colon T \\
    }{
      \Pi \vdash F\colon\text{Function}
    }

    \inferH{typing program}{
      \Pi = F_0\dots F_{n-1}\\
      \forall i\in\{0,\dots,n-1\},\ F_0\dots F_{i-1} \vdash F_i\colon\text{Function}
    }{
      \Pi \colon\text{Program}
    }
  \end{mathpar}
  \Description{}
  \normalsize\caption{Typing rule for blocks, function, and whole program}
  \label{fig:rule/ty-fn}
\end{figure}

The judgement $(\mathbf{\Gamma},\mathbf{A}) \vdash^{\Pi,f} B \colon T$ for a block $B$ states that
the block $B$ has a return value of type $T$ under the context $\mathbf{\Gamma}$ and the lifetime preorder $\mathbf{A}$.
Note that all lifetimes introduced in a block have to end within the block, and that an externally introduced lifetime cannot end within a block.
Therefore, the lifetime preorder $\mathbf{A}$ has to be the same before and after the block.
We also assume all the variables in the context $\Gamma$ have to be used or dropped to ensure linear use of the variables -- only the return value can be left in the context.

The rule \scref{typing fn} is for typing a function definition.
It first creates a typing context $\mathbf{\Gamma}$ and a lifetime preorder $\mathbf{A}$ from the signature $\Sigma$ of the function $F$.
Then, it type checks the body $B$ of the function $F$ under the context $(\mathbf{\Gamma}, \mathbf{A})$.

The whole program is defined by a sequence of functions $\Pi$.
Note that in the function definition $F_i$, we can only use functions defined earier $F_0,\dots,F_{i-1}$ to prevent the recursive function call.

\section{Simulation Semantics}\label{ax-sec:sim-sem}
\subsection{Rules for Simulation}\label{ax-subsec:sim-def}
\Cref{fig:rule/sim-expr} shows the operational semantics for the simulation of expressions,
and \Cref{fig:rule/sim-stmt} shows that of statements.
Most of the rules are straightforward, and the non-trivial case are explained in \Cref{subsec:sim-def}.
Note that measurement can have multiple outcomes, which makes the semantics non-deterministic.

\begin{figure}[tbh]
  \footnotesize\begin{mathpar}
    \infer{
      e = x\\
      L = \loc(x)
    }{
      \loc[x\mapsto[],y\mapsto L],\quad
      \ket{\phi},\quad
      s
    }

    \infer{
      e = b\\
      b\in\{\False,\True\}\\
      l_c : \text{fresh}
    }{
      \loc[y\mapsto l_c],\quad
      \ket{\phi},\quad
      s[l_c \leftarrow b]
    }

    \infer{
      e = (x_0,x_1)\\
      L_0 = \loc(x_0)\\
      L_1 = \loc(x_1)
    }{
      \loc[x_0,x_1\mapsto[], y\mapsto L_0 + L_1],\quad
      \ket{\phi},\quad
      s
    }

    \infer{
      e = \Copy(x)\\
      L_c = \{\ l_c\in \Loc_c\ |\ l_c: \text{owned by}\ x\ \}\\
      L'_c : \text{fresh classical locations}\\
      L' = \loc(x)[L'_c/L_c]\\
      t = s(L_c)
    }{
      \loc[y\mapsto L'],\quad
      \ket\phi,\quad
      s[L'_c \leftarrow t]
    }

    \infer{
    e = \Meas(x)\\
    [l_q] = \loc(x)\\
    l_c : \text{fresh}\\
    \ket{\phi} \in \HH_\text{others}\otimes\HH_{l_q}
    }{
    \loc[x\mapsto[],y\mapsto l_c],\quad
    (\mathrm{id}\otimes \bra{0})\ket{\phi},\quad
    s[l_c \leftarrow 0]
    }

    \infer{
    e = \Meas(x)\\
    [l_q] = \loc(x)\\
    l_c : \text{fresh}\\
    \ket{\phi} \in \HH_\text{others}\otimes\HH_{l_q}
    }{
    \loc[x\mapsto[],y\mapsto l_c],\quad
    (\mathrm{id}\otimes \bra{1})\ket{\phi},\quad
    s[l_c \leftarrow 1]
    }

    \infer{
      e = U(x)\\
      L = \loc(x)\\
      \ket{\phi} \in \HH_\text{others}\otimes\HH_{L}
    }{
      \loc[x\mapsto[],y\mapsto L],\quad
      (\mathrm{id}\otimes U)\ket\phi,\quad
      s
    }

    \infer{
      e = [c](x)\\
      c\colon {\{0,1\}}^n \rightarrow {\{ 0,1\}}^{n+k}\\
      L_q^n = \loc(x)\\
      L_q^k : \text{fresh } m\text{-length list of locations} \\
      \ket{\phi} \in \HH_\text{others}\otimes\HH_{L^n_q}
    }{
      \loc[x\mapsto[],y\mapsto L_q^n + L_q^k],\quad
      (\mathrm{id}\otimes [c])\ket{\phi},\quad
      s
    }

    \infer{
    e = f\langle \alpha_0,\dots,\alpha_{m-1} \rangle
    (x_0,\dots,x_{n-1})\\
    \mathbf{A_f} (x'_0\colon T_0,\dots,x'_{n-1}\colon T_{n-1}) \rightarrow T
    :\text{signature of } f\\
    \{ S; y' \}: \text{body of } f\\
    (\loc[x_i\mapsto[],x'_i\mapsto\loc(x_i)],\ket\phi,s)
    \xrightarrow{S\colon ((x'_i\colon T_i),\,\mathbf{A_f}) \rightarrow ((z\colon T),\,\mathbf{A_f})}
    (\loc',\ket{\psi},s')
    }{
    \loc'[y\mapsto\loc'(y'),y'\mapsto[]],\quad
    \ket\psi,\quad
    s'
    }

    \infer{
    e = \If\ x\ \{\ S; x'\ \}\ \Else\ B_f\\
    \loc(x) = [l]\\
    s(l) = \True\\
    (\loc,\ket{\phi},s)\xrightarrow{S}(\loc',\ket{\psi},s')
    }{
    \loc'[x'\mapsto[], y\mapsto \loc'(x')],\quad
    \ket{\psi},\quad
    s'
    }

    \infer{
    e = \If\ x\ B_t\ \Else\ \{\ S; x'\ \}\\
    \loc(x) = [l]\\
    s(l) = \False\\
    (\loc,\ket{\phi},s)\xrightarrow{S}(\loc',\ket{\psi},s')
    }{
    \loc'[x'\mapsto[], y\mapsto \loc'(x')],\quad
    \ket{\psi},\quad
    s'
    }

    \infer{
    e = \Qif\ x\ \{\ S_1; x_1\ \}\ \Else\ \{\ S_0; x_0\ \}\\
    L_f\colon \text{ frozen locations of quantum memory in branches }\\
    [l_x] = \loc(x)\\
    L_i = \loc_i(x_i)\\
    L : \text{fresh}\\
    |L| = |L_i|\\
    \ket{\phi} = \sum_{i\in \{0,1\}} \ket{i}\ket{\phi_i}\in \HH_{l_x}\otimes\HH_\text{others}\\\\
    (\loc,\ket{0}\ket{\phi_0},s)\xrightarrow{B_\ket{0}}
    (
    \loc_0,
    \ket{0}\ket{\psi_0}\in \HH_{l_x}\otimes(\HH_{L_0}\otimes\HH_{L_f - l_x}),
    s'
    )\\
    (\loc,\ket{1}\ket{\phi_1},s)\xrightarrow{B_\ket{1}}
    (
    \loc_1,
    \ket{1}\ket{\psi_1}\in \HH_{l_x}\otimes(\HH_{L_1}\otimes\HH_{L_f - l_x}),
    s'
    )\\
    }{
    \loc_0[x_0\mapsto[],y\mapsto L],\quad
    \sum_{i\in\{0,1\}} \ket{i}\ket{\psi_i}\in\HH_{l_x}\otimes(\HH_{L}\otimes\HH_{L_f - l_x}),\quad
    s'
    }
  \end{mathpar}
  \Description{}
  \caption{Simulation semantics for \qurtsinline{let y = e}}
  The first row of each rule shows the assumption,
  and the second row shows the environment after the evaluation
  of the statement $\Let y = e$ under the environment $(\loc, \ket{\phi}, s)$.
  \label{fig:rule/sim-expr}
\end{figure}

\begin{figure}
  \footnotesize
  \begin{mathpar}
    \infer{
      S = S_1; S_2\\
      (\loc,\ket\phi,s)\xrightarrow{S_1}(\loc',\ket{\phi'},s')\\
      (\loc',\ket{\phi'},s')\xrightarrow{S_2}(\loc'',\ket{\phi''},s'')
    }{
      \loc'',\quad
      \ket{\phi''},\quad
      s''
    }

    \infer{
      S = \Noop,\Newlft \alpha, \Endlft \alpha, \alpha \leq \beta, x \As T
    }{
      \loc,\quad
      \ket\phi,\quad
      s
    }

    \infer{
      S = \Let (y_0,y_1) = x\\
      \loc(x) = L_0 + L_1\\
      x\colon T_0\times T_1\\
      L_i\ \text{has the form of}\ \sem{T_i}
    }{
      \loc[x\mapsto[], y_0\mapsto L_0, y_1\mapsto L_1],\quad
      \ket\phi,\quad
      s
    }

    \infer{
      S = \Drop x\\
      L_q = \{\ l_q\in \Loc_q\ |\ \text{owned by}\ x\ \}\\
      L_c = \{\ l_c\in \Loc_c\ |\ \text{owned by}\ x\ \}\\
      \ket\phi = \sum_{i\in {\{0,1\}}^{|L_q|}} \ket{\phi_i}\ket{i} \in \HH_{\mathrm{others}}\otimes\HH_{L_q}\\
      s = (s_0,s_1) \in {\{0,1\}}^\text{others} \times {\{0,1\}}^{L_c}
    }{
      \loc[x\mapsto[]],\quad
      \sum_{i\in {\{0,1\}}^{|L_q|}} \ket{\phi_i},\quad
      s_0
    }
  \end{mathpar}
  \Description{}
  \caption{Simulation semantics for statements}
  The first row of each rule shows the assumption,
  and the second row shows the environment after the evaluation
  of the statement $S$ under the environment $(\loc, \ket\phi, s)$.
  \label{fig:rule/sim-stmt}
\end{figure}

\subsection{Proofs for Section \ref{sec:simulation-semantics}}\label{ax-subsec:sim-proofs}

The following lemma includes the proof of \Cref{lem:sim-PQ-loc}.
\begin{lemma}\label{ax-lem:sim-block}
  Let $B = \{ S; x \}$ be a block.
  Assume the following.
  \begin{itemize}
    \item $(\mathbf{\Gamma},\mathbf{A}) \vdash B\colon T$,
    \item $(\loc, v, s) \xrightarrow S (\loc', v', s')$.
  \end{itemize}
  Then, we have the following.
  \begin{enumerate}
    \item The domain of $\loc'$ is $\mathrm{dom}(\loc) - \{x_i\ |\ x_i\in \mathbf{\Gamma}\} + \{x\}$,
          where the domain of $\loc$ is defined by the set of variables $\{x\in Var\mid \loc(x) \neq []\}$.
    \item The frozen locations of the memory in $v'$ is the same as that in $v$.
  \end{enumerate}
  Moreover, if $B,T\colon\PQ$, then $s'$ is the restriction of $s$ to the frozen locations.
\end{lemma}
\begin{proof}
  Since the domain of $\loc$ changes according to the change in typing context, the first statement is clear.

  For the second statement, first of all, observe that
  the type context $\mathbf{\Gamma}$ cannot contain any frozen variables.
  Because if $y\colon^{\dagger\alpha}T'$ is contained in $\mathbf{\Gamma}$,
  then since the $\Endlft\alpha$ cannot appear in the block $B$,
  $y$ has to be kept frozen until the end of $S$.
  However, the succeeding typing context $(x\colon T)$ has no frozen variables.

  Therefore, any location $l$ which was frozen at the beginning of $S$ is still frozen
  at the end of $S$.
  Also, all locations $l$ that become frozen somewhere in $S$ have to be defrosted by the end of $S$
  since $(x\colon T)$ does not contain any frozen variables.

  The last statement is clear from the fact that
  the set of the frozen locations will be the same at the end,
  and all the variables that own classical locations have to be dropped by the end of $S$.
\end{proof}

We now prove \Cref{lem:simulation-progress}.
This lemma can be stated more mathematically as follows.
\begin{lemma}
  Let $S$ be a statement, $e$ be an environment, and 
  $S\colon(\mathbf{\Gamma}, \mathbf{A}) \xrightarrow{\Pi,f} (\mathbf{\Gamma'},\mathbf{A}')$ be a well-typed statement.
  If $e$ is compatible with $\mathbf{\Gamma}$,
  then there exists a compatible environment $e'$ such that $e \xrightarrow{S} e'$.
\end{lemma}
\begin{proof}
  This can be proved by induction on the structure of the statement $S$.
  Most of the cases are trivial, and the only non-trivial case is the $\Qif$ statement,
  which follows from \Cref{ax-lem:sim-block}.
\end{proof}

We now set out to prove preservation of probability in the evaluation step (\Cref{thm:simulation}),
which is the main theorem in this section.
In the simulation semantics, the evaluation step of statement $\Endlft$ can lose some information contained in the frozen locations of the quantum memory.
To prove \Cref{thm:simulation}, we need to show that the qubit we defrost by $\Endlft$ does not contain any information that would be required by a latter $\Drop$ statement.
Therefore, we need to track all the locations of the quantum memory which can contain the information that is needed to drop a qubit.
To this end, we introduce a \emph{dependency graph} to keep track of this dependency relation, and we extend the simulation semantics with the dependency graph.

\begin{definition}
  A \emph{dependency graph} is a triple $(G,D,F)$ where
  \begin{itemize}
    \item $G$ is a directed graph whose vertices are locations,
    \item $D$ is a subset of the set of locations, and
    \item $F$ is a function which maps $l\in D$ to a function $F(l) : {\{0,1\}}^{\mathrm{source}_G(l)} \to {\{0,1\}}$.
  \end{itemize}
\end{definition}

\begin{definition}\label{def:sim-well-formed-graph}
  A tuple $(\mathbf{\Gamma},\mathbf{A},e,\mathcal{G} = (G,D,F))$ of a context, a lifetime preorder, an environment, and a dependency graph is \emph{well-formed} if:
  \begin{itemize}
    \item The environment $e$ is compatible with $\mathbf{\Gamma}$;
    \item Any location $l\in L_q+L_c$ is owned by at most one variable in the type context;
    \item $G$ is acyclic;
    \item $D$ is the set of affinely owned locations;
    \item $l$ is frozen if there is an edge whose source is $l$;
    \item For each $l\in D$, if $L$ is the set of locations $\mathrm{source}_G(l)$, the state of the quantum memory in
          $\HH_{L} \otimes \HH_{l} \otimes \HH_{\mathrm{others}}$
          has the form 
          \[
            \sum_{i\in{\{0,1\}}^{|L|}}  \ket{i}\ket{F(l)(i)}\ket{\phi_i}.
          \]
  \end{itemize}
\end{definition}
Unlike the definition of well-formedness in \Cref{def:sim-well-formed-env},
the subset $D$ can include locations that are frozen.
This is necessary to allow the creation of a immutable reference to an affinely owned location.

\begin{figure}[tbh]
  \footnotesize
  \begin{mathpar}
    \infer{
      S = \Endlft \alpha\\
    }{
      D - \{l\ |\ l \text{ is affine until }\alpha\}
    }

    \infer{
      S = \Drop(x)\\
    }{
      G - \{e\ |\ \mathrm{target}(e)\ \text{is owned by}\ x\}\\
      D - \{l\ |\ l \text{ is owned by }x\}
    }

    \infer{
      S = x \As T\\
      L = \{ l \ |\ l \text{ is affine in preceding, but not after evaluating }S \}
    }{
      D - L
    }

    \infer{
    S = (\Let y = [c](x))\\
    [l_q^1,\dots,l_q^n] = \loc(x)\\
    \loc(x) + L' = \loc'(y)
    }{
    G + \{ l\to l' \mid l\in \mathrm{source}_{G}, l' \in \loc'(y) \}\\
    c (F(l_q^1),\dots,F(l_q^n))
    }

    \infer{
      S = (\Let y = f\langle \alpha_0,\dots,\alpha_{m-1} \rangle(x_0,\dots,x_{n-1}))\\
      \{ S; z \}: \text{body of } f\\
      (G,D,F) \xrightarrow{S} (G',D',F')
    }{
      G'\\ D'\\ F'
    }

    \infer{
    S = (\Let y = \If\ x\ B_t\ \Else\ B_f)\\
    [l_c] = \loc(x)\\ \text{if } s(l_c) = \True \text{ then } S' = B_t \text{ else } S' = B_f\\
    (G,D,F) \xrightarrow{S'} (G',D',F')
    }{
    G'\\ D'\\ F'
    }

    \infer{
    S = (\Let y = \Qif\ x\ \{S_\ket{1}; z_1\}\ \Else\ \{S_\ket{0};z_0\})\\
    [l_x] = \loc(x)\\
    L_i = \loc_i(z_i)\\
    L: \text{fresh locations we chose for }y\\
    (G,D,F) \xrightarrow{S_\ket{i}} (G_i,D_i,F_i)\\
    D' \coloneqq D_0[L/L_0] = D_1[L/L_1]\\
    \forall l\not\in L, F_0(l) = F_1(l)
    }{
    G_0[L/L_0] \cup G_1[L/L_1] + \{ l_x \to l \mid l \in L \}\\
    D'\\
    F'(l)\colon {\{0,1\}}^{l_x} \times {\{0,1\}}^\text{others} \rightarrow {\{0,1\}};\\
    \forall l \in L,\ F'(l)(i,n) = F_i(l)(i,n)\\
    \forall l \not\in L,\ F'(l) = F_0(l) = F_1(l)\\
    }
  \end{mathpar}
  \Description{}
  \caption{Selected rules of simulation semantics with dependency graph}
  The first row of each rule shows the assumption, and the second row shows the dependency graph
  $(G',D',F')$ after the evaluation of the statement $S$ under the dependency graph $(G,D,F)$.
  We only show the non-trivial cases of the rules with the updated part of $(G',D',F')$.
  Rules for $F$ is ommitted in some cases where $F$ is trivially updated by restricting the domain $D$.
  \label{fig:rule/sim-extended}
\end{figure}

In \Cref{fig:rule/sim-extended}, the non-trivial cases of the extended simulation semantics are shown.
Here, again, the most complicated case is the \qurtsinline{qif}.
Before proving the preservation of the well-formedness by the evaluation step,
let us make sure this rule is justified.

\begin{lemma}
  The following two assumption in the rule for the $\Qif$ statement are automatically satisfied.
  \begin{itemize}
    \item $D' \coloneqq D_0[L/L_0] = D_1[L/L_1]$
    \item $\forall l\not\in L\colon F_0(l) = F_1(l)$
  \end{itemize}
\end{lemma}
\begin{proof}
  Since the type of $z_1$ and $z_0$ are the same,
  the set of affinely owned locations in $L_i$ is the same for both branches.
  The other locations which is affinely owned in the succeeding environment for the branches are all frozen,
  which has to be the same for both branches.

  For functions $F$ the proof is similar.
  All the locations which are not in $L$ are frozen during the evaluation.
  Therefore, we cannot change $F(l)$ for all $l\not\in L$.
\end{proof}

Let us now prove the preservation of the well-formedness.

\begin{theorem}\label{thm:simulation-extended}
  The well-formedness of the environment and dependency graph is preserved by the evaluation step
  for any statement $S$ if every $\Endlft \alpha$ has a corresponding $\Newlft \alpha$.
\end{theorem}
\begin{proof}
  The proof proceeds by induction. Compatibility is preserved by the definition of the evaluation step.

  No location can be owned by different variables since the rules are carefully
  designed to use fresh locations.
  Most care was taken for the cases that use blocks.
  Our rules ask to choose fresh locations, including the ones owned by an out-of-scope variable,
  even during the evaluation of a subroutine, which enables us to keep the ownership of the locations unique.

  To check that a graph can not have a cycle, we only need to check the rules for $[c]$ and $\Qif$
  since these are the only cases introducing new edges in the graph.
  In these cases, the target vertex of the new edge has to be owned by some variable which is not frozen.
  Since non-frozen location cannot be a source of any edge, the graph remains acyclic.

  In any step, the rules are carefully designed so that
  $D$ remains the set of affinely owned quantum locations.

  Proving that preservation of all the source of the edges being frozen is not trivial
  because of the rule for $\Endlft \alpha$.
  To prove this, we need to show that any location $l$ which is frozen
  until $\alpha$ cannot be a source of any edge whose target continues
  to be affinely owned after the lifetime.
  Because we assumed that $\Newlft \alpha$ and $\Endlft \alpha$ are in one-to-one correspondence,
  we have the corresponding $\Newlft \alpha$ occuring before,
  and there was no outgoing edge from $l$.
  Therefore, any outgoing edge from $l$ was the one that was created
  between the $\Newlft \alpha$ and $\Endlft \alpha$ by a statement
  $\Let y = \Qif x \dots$ or $\Let y = [c](x)$.
  In the former case, the location of $x$ was $l$, and a new edge was created to a location owned by $y$.
  In the latter case, there was an edge from $l$ to a location owned by $x$, and a new edge was created to a location owned by $y$.
  In both cases, the type of $y$ has to be bound by $\Own^\beta$ with some $\beta \leq \alpha$,
  the location owned by $y$ cannot be affinely owned after the lifetime $\alpha$.

  For the last condition of well-formedness, we first observe that, for any location,
  we can only apply identity gates while the location is frozen,
  and if we want to keep it affinely owned, we can only apply $[c]$,
  which might be controlled by some frozen locations.
  Therefore, we only need to check the two case $\Let y = [c](x)$ and $\Qif x \dots$
  which do essentially affect $F$.
  The former is trivial, and the latter can also be checked straightforwardly by carefully checking the definition of the semantics.
\end{proof}

We can now prove the \Cref{thm:simulation} as a corollary of \Cref{thm:simulation-extended}.
\begin{proof}[Proof of \Cref{thm:simulation}]
  Given any well-formed environment $(\mathbf{\Gamma}, \mathbf{A}, e)$
  with no frozen variables in $\mathbf{\Gamma}$
  in the sense of \Cref{def:sim-well-formed-env},
  we can extend the environment $e$ to a well-formed tuple with dependency graph
  $(\mathbf{\Gamma}, \mathbf{A}, e, \mathcal{G})$ in the sense of \Cref{def:sim-well-formed-graph} as follows.
  \begin{itemize}
    \item $G$ is the graph of all the locations used in the environment,
          with edges $E = \{l_f\to l_a\ |\ l_f:\text{frozen}, l_a:\text{affinely owned}\}$,
    \item $D$ is the set of the affinely owned locations,
    \item $F(l_a)$ is the function $\pi_{l_a}\circ f$ where $\pi_{l_a}\colon {\{0,1\}}^m \to {\{0,1\}}$ is the projection.
  \end{itemize}
  It follows from \Cref{thm:simulation-extended} that the evaluation step preserves well-formedness.

  We prove preservation of probability by induction.
  The non-trivial cases are measure, drop, and quantum if.
  For measurement:
  \[
    \left| \ket{\phi} \right|^2
    =
    \left| (\mathrm{id}\otimes \bra{0})\ket{\phi} \right|^2
    +
    \left| (\mathrm{id}\otimes \bra{1})\ket{\phi} \right|^2
    \text.
  \]

  For the quantum if case $\Let y = \Qif x \{S_1;z_1\} \Else \{S_0;z_0\}$:
  since no measurement can appear in the branches,
  each environment $e_i$ is such that $S_i\colon (\loc,\ket{i}\ket{\phi_i},s) \to e_i$
  satisfies $|\ket{\phi_i}|^2 = \|e_i\|^2$ by the induction hypothesis.
  The fact that the state of quantum memories of $e_0$ and $e_1$ are orthogonal
  proves preservation of the probability in this case.

  For the case of the drop, we use the well-formedness.
  When we try to drop a variable which owns a quantum location $l$,
  it has to be affinely owned.
  Therefore, the quantum state must have the form of
  $
    \sum_{i\in{\{0,1\}}^{|L|}}  \ket{i}\ket{F(l)(i)}\ket{\phi_i} \in \HH_{L} \otimes \HH_{l} \otimes \HH_{\mathrm{others}}
  $.
  Since the set ${\{\ket{i}\}}_{i\in{\{0,1\}}^{|L|}}$ is orthogonal,
  \[
    \left| \sum_{i\in{\{0,1\}}^{|L|}}  \ket{i}\ket{F(l)(i)}\ket{\phi_i} \right|^2
    =
    \left| \sum_{i\in{\{0,1\}}^{|L|}}  \ket{i}\ket{\phi_i} \right|^2.
  \]
\end{proof}
\section{Uncomputation Semantics}\label{ax-sec:unc-semantics}
In this section, we define the physically implementable semantics with uncomputation in detail,
and prove some properties of the semantics.
Our semantics is similar to multi-threaded computation with
garbage collection in classical computation.
The garbage collector in classical computation corresponds to the \emph{uncomputer} in our semantics.

\subsection{Reversible Pebble Game for Qurts in Detail}
In this section, we demonstrate how the pebble game for Qurts is played, and different pebbling strategies simulate different but equivalent quantum circuits.

Consider the following circuit graph which represents a quantum if statement on the left-hand side.
The vertical dotted line separates the vertices with different locations.
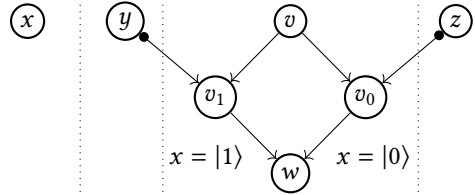
\begin{center}
  \begin{minipage}{0.35\textwidth}
    \begin{lstlisting}[language=qurts]
      let (y,z,t') = qif x {
        let (y,t1) = [cnot](y,t); (y,z,t1)
      } else {
        let (z,t0) = [cnot](z,t); (y,z,t0)
      };
    \end{lstlisting}
  \end{minipage}
  \qquad
  \begin{tikzpicture}[baseline={(0,1)}]
    \node[vertex] (x) at (0,2) {$v$};
    \node[vertex] (y) at (-1,1) {$v_1$};
    \node[vertex] (z) at (1,1) {$v_0$};
    \node[vertex] (m) at (0,0) {$w$};
    \draw[->] (x) -- (y);
    \draw[->] (x) -- (z);
    \draw[->] (y) -- node[auto,swap]{$x=\ket1$} (m);
    \draw[->] (z) -- node[auto]{$x=\ket0$} (m);
    \node[vertex] (a) at (-2.2,2) {$y$};
    \node[vertex] (b) at (2.2,2) {$z$};
    \draw[{Circle}->] (a) -- (y);
    \draw[{Circle}->] (b) -- (z);
    \node[vertex] at (-3.5,2) {$x$};
    \draw[dotted] (-1.7,2.3) -- (-1.7,-.3);
    \draw[dotted] (1.7,2.3) -- (1.7,-.3);
    \draw[dotted] (-2.8,2.3) -- (-2.8,-.3);
  \end{tikzpicture}
\end{center}
Here is an example play of this pebble game, using the split rule, the gate rule, and the merge rule:
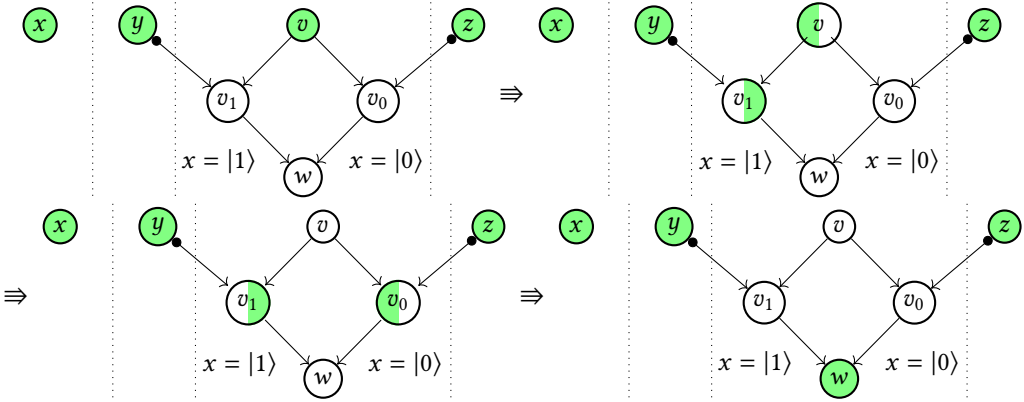
\begin{center}
  \begin{tikzpicture}[baseline={(0,1)}]
    \node[pebbled] (x) at (0,2) {$v$};
    \node[vertex] (y) at (-1,1) {$v_1$};
    \node[vertex] (z) at (1,1) {$v_0$};
    \node[vertex] (m) at (0,0) {$w$};
    \draw[->] (x) -- (y);
    \draw[->] (x) -- (z);
    \draw[->] (y) -- node[auto,swap]{$x=\ket1$} (m);
    \draw[->] (z) -- node[auto]{$x=\ket0$} (m);
    \node[pebbled] (a) at (-2.2,2) {$y$};
    \node[pebbled] (b) at (2.2,2) {$z$};
    \draw[{Circle}->] (a) -- (y);
    \draw[{Circle}->] (b) -- (z);
    \node[pebbled] at (-3.5,2) {$x$};
    \draw[dotted] (-1.7,2.3) -- (-1.7,-.3);
    \draw[dotted] (1.7,2.3) -- (1.7,-.3);
    \draw[dotted] (-2.8,2.3) -- (-2.8,-.3);
  \end{tikzpicture}
  $\ \Rrightarrow\ $
  \begin{tikzpicture}[baseline={(0,1)}]
    \begin{scope}
      \clip (0,2) circle (0.3);
      \fill[light-green] (-.5,1.5) rectangle (0,2.5);
      \draw[ultra thick] (0,2) circle (0.3);
      \node (z) at (0,2) {$v$};
    \end{scope}
    \begin{scope}
      \clip (-1,1) circle (0.3);
      \fill[light-green] (-.5,0.5) rectangle (-1,1.5);
      \draw[ultra thick] (-1,1) circle (0.3);
      \node (y) at (-1,1) {$v_1$};
    \end{scope}
    \node[vertex] (z) at (1,1) {$v_0$};
    \node[vertex] (m) at (0,0) {$w$};
    \draw[->] (x) -- (y);
    \draw[->] (x) -- (z);
    \draw[->] (y) -- node[auto,swap]{$x=\ket1$} (m);
    \draw[->] (z) -- node[auto]{$x=\ket0$} (m);
    \node[pebbled] (a) at (-2.2,2) {$y$};
    \node[pebbled] (b) at (2.2,2) {$z$};
    \draw[{Circle}->] (a) -- (y);
    \draw[{Circle}->] (b) -- (z);
    \node[pebbled] at (-3.5,2) {$x$};
    \draw[dotted] (-1.7,2.3) -- (-1.7,-.3);
    \draw[dotted] (1.7,2.3) -- (1.7,-.3);
    \draw[dotted] (-2.8,2.3) -- (-2.8,-.3);
  \end{tikzpicture}
  \quad
  $\ \Rrightarrow\ $
  \begin{tikzpicture}[baseline={(0,1)}]
    \node[vertex] (x) at (0,2) {$v$};
    \begin{scope}
      \clip (-1,1) circle (0.3);
      \fill[light-green] (-.5,0.5) rectangle (-1,1.5);
      \draw[ultra thick] (-1,1) circle (0.3);
      \node (y) at (-1,1) {$v_1$};
    \end{scope}
    \begin{scope}
      \clip (1,1) circle (0.3);
      \fill[light-green] (.5,0.5) rectangle (1,1.5);
      \draw[ultra thick] (1,1) circle (0.3);
      \node (z) at (1,1) {$v_0$};
    \end{scope}
    \node[vertex] (m) at (0,0) {$w$};
    \draw[->] (x) -- (y);
    \draw[->] (x) -- (z);
    \draw[->] (y) -- node[auto,swap]{$x=\ket1$} (m);
    \draw[->] (z) -- node[auto]{$x=\ket0$} (m);
    \node[pebbled] (a) at (-2.2,2) {$y$};
    \node[pebbled] (b) at (2.2,2) {$z$};
    \draw[{Circle}->] (a) -- (y);
    \draw[{Circle}->] (b) -- (z);
    \node[pebbled] at (-3.5,2) {$x$};
    \draw[dotted] (-1.7,2.3) -- (-1.7,-.3);
    \draw[dotted] (1.7,2.3) -- (1.7,-.3);
    \draw[dotted] (-2.8,2.3) -- (-2.8,-.3);
  \end{tikzpicture}
  $\ \Rrightarrow\ $
  \begin{tikzpicture}[baseline={(0,1)}]
    \node[vertex] (x) at (0,2) {$v$};
    \node[vertex] (y) at (-1,1) {$v_1$};
    \node[vertex] (z) at (1,1) {$v_0$};
    \node[pebbled] (m) at (0,0) {$w$};
    \draw[->] (x) -- (y);
    \draw[->] (x) -- (z);
    \draw[->] (y) -- node[auto,swap]{$x=\ket1$} (m);
    \draw[->] (z) -- node[auto]{$x=\ket0$} (m);
    \node[pebbled] (a) at (-2.2,2) {$y$};
    \node[pebbled] (b) at (2.2,2) {$z$};
    \draw[{Circle}->] (a) -- (y);
    \draw[{Circle}->] (b) -- (z);
    \node[pebbled] at (-3.5,2) {$x$};
    \draw[dotted] (-1.7,2.3) -- (-1.7,-.3);
    \draw[dotted] (1.7,2.3) -- (1.7,-.3);
    \draw[dotted] (-2.8,2.3) -- (-2.8,-.3);
  \end{tikzpicture}
\end{center}
This strategy on the game corresponds to the circuit on the left-hand side.
The circuit on the right is equivalent to the circuit on the left, and the pebbling strategy for the circuit on the right is given by moving the pebble conditioned on $x = \ket0$ to $v_0$ first, and then moving the pebble conditioned on $x = \ket1$ to $v_1$.
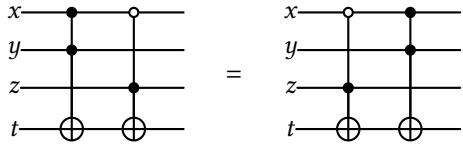
\[
  \begin{quantikz}[row sep=.3cm]
    x&\ctrl{3}&\ctrl[open]{3}& \\
    y&\ctrl{2}&&\\
    z&&\ctrl{1}&\\
    t&\targ{}&\targ{}&
  \end{quantikz}
  \quad = \quad
  \begin{quantikz}[row sep=.3cm]
    x&\ctrl[open]{3}&\ctrl{3}& \\
    y&&\ctrl{2}&\\
    z&\ctrl{1}&&\\
    t&\targ{}&\targ{}&
  \end{quantikz}
\]

This splitting rule also allows handling a situation when a $\Drop$ appears inside a quantum if statement.
Let us now consider the following code snippet.
\begin{center}
  \begin{minipage}{0.6\textwidth}
    \begin{lstlisting}[language=qurts]
      let v0 = $\ket0$;
      let (x,y) = [cnot](x,v0);
      let y' = qif z { drop y; let y = $\ket{0}$; y } else { y }
    \end{lstlisting}
  \end{minipage}
\end{center}
In this code, the variable $y$ is dropped in the then branch of the quantum if statement.
The naive way to execute this code is to actually uncompute $y$ in the then branch; since the uncomputation of $y$ outside the $\Qif$ is to apply \texttt{[cnot]} to $x$ and $y$ again, the uncomputation of $y$ in the then branch is to apply negatively-controlled-\texttt{[cnot]} (=Toffoli gate conjugated by X gate).
Let us see how this code is executed in the pebble game.
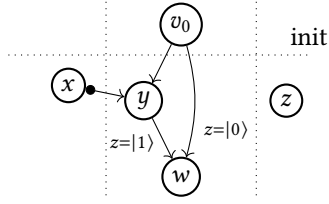
\begin{center}
  \begin{tikzpicture}
    \node[vertex] (init) at (.5,2) {$v_0$};
    \draw[dotted] (-1.8,1.6) -- (2.2,1.6) node[above] {init};
    \node[vertex] (x) at (-1,1.2) {$x$};
    \draw[dotted] (-.5,2.3) -- (-.5,-.3);
    \node[vertex] (y) at (0,1) {$y$};
    \node[vertex] (m) at (.5,0) {$w$};
    \draw[->] (init) -- (y);
    \draw[->] (y) -- node[left,pos=0.6]{$\scriptstyle z=\ket1$} (m);
    \draw[->] (init) to [bend left=15] node[pos=0.6,auto]{$\scriptstyle z=\ket0$} (m);
    \draw[{Circle}->] (x) -- (y);
    \node[vertex] (z) at (1.9,1) {$z$};
    \draw[dotted] (1.5,2.3) -- (1.5,-.3);
  \end{tikzpicture}
\end{center}
The execution explained above amounts to the following pebbling strategy:
Then, move it to $y$ by the \emph{gate} rule.
Now, split the pebble on $y$ into two half pebbles, one conditioned on $z=\ket0$ and the other on $z=\ket1$.
Next, we put the pebble conditioned on $z=\ket0$ back to $v_0$ by the opposite direction of \emph{gate} rule.
Finally, we move the pebbles to $w$ by the \emph{merge} rule.

However, there is a more efficient way to play this game.
First, split the pebble on $v_0$ into two half pebbles, one conditioned on $z=\ket0$ and the other on $z=\ket1$.
Then, move the pebble conditioned on $z=\ket1$ to $y$ by the \emph{gate} rule.
And then, move the pebbles to $w$ by the \emph{merge} rule.
Representing the pebbling strategy in a circuit, this is equivalent to applying just a single Toffoli gate.

\subsection{Uncomputer, Threads, and Schedulers}
The execution of program in a thread is interpreted as generation of circuit graphs
and operation on quantum and classical memory.
On the other hand, the uncomputer is a background process that plays the pebble game on the circuit graph, and updates the state of quantum memory accordingly.
The scheduler controls the whole system including the uncomputer and threads,
decides which thread to execute, and allows the uncomputer to run between the executions of threads.

\paragraph{Quantum Memory}
The number of qubits needed to execute this operational semantics cannot be determined only by the program,  because it depends on the strategy of the pebbling game which is not fixed.
Therefore, although the function $\loc$ still maps variables to locations,
the number of locations does not correspond to the number of qubits.
Instead, in this semantics, as in the original idea of the reversible pebble games
for quantum computation~\cite{MeuliSRBM_2019_RevPebble},
the number of pebbles used, or more precisely, the number of labels attached to pebbles,
corresponds to the space consumption.
Therefore, the definition of the state of quantum memory needs to be modified from the simulation semantics.
Here is the new definition.
\begin{definition}[$\text{label}$, state of memory]
  The set $\mathrm{Label}$ is the disjoint union $\mathrm{MainLabel} \sqcup \mathrm{AuxLabel}$
  of main labels and auxiliary labels.
  and each label $p$ has its location $\loc(p)$.
  For each location $l$, there is a unique main label $p \in \mathrm{MainLabel}$ such that $\loc(p) = l$.

  A \emph{state of quantum memory} $\ket\phi$ is a state of the Hilbert space
  $\HH_{\bar p} \coloneqq \bigotimes_{p_i\in \bar p}\C^{2}$ where $\bar p$ is a subset of $\mathrm{Label}$.
\end{definition}
Roughly speaking, a label corresponds to a qubit used in the execution step.
A main label in $\mathrm{MainLabel}$ whose location is $l$ corresponds to a qubit in the location $l$.
An augmented label in $\mathrm{AugLabel}$ corresponds to a qubit which is a copy of some intermediate state of computation which you must uncompute at some point.
The state of classical memory $s$ is the same as in the simulation semantics: a function $s$ from
a set of locations of classical memory to the set $\{0,1\}$.

\paragraph{Uncomputer}
The uncomputer mainly performs computation or uncomputation
on affinely owned qubits by playing the pebble game.
The rules for the uncomputer are defined by a relation $\Rrightarrow$ on the pair of a state of quantum memory
and pebbling state $(\ket\phi, \mathrm{peb})$
with a context of circuit graph $G$ and lifetime preorder $\mathbf{A}$.
The lifetime preorder is used for the rule to freeze or defrost some linearly owned resource
until the end of some lifetime.
\begin{align}
  \mathbf{A}, G\models(\ket\phi, \mathrm{peb}) \Rrightarrow (\ket{\phi'}, \mathrm{peb}')
  \label{eq:unc-uncomputer}
\end{align}

As with garbage collectors, the uncomputer runs independently of the execution of threads
and releases quantum memory by uncomputation.
Unlike conventional garbage collection, the uncomputer demonstrates greater flexibility by enabling computation or recomputation.
However, it also has some constraints: it may need to free memory by the end of some lifetime $\alpha$.

\paragraph{Thread}
We use multi-threaded semantics since we would like to allow both branches of a $\Qif$ statement to be executed concurrently.
Here is the set of data that a thread consists of.
\begin{definition}
  A \emph{thread} $\Theta$ consists of the following components:
  \begin{itemize}
    \item A thread identifier $t\in\N$.
    \item A statement $S_t\colon(\mathbf{\Gamma}_t,\mathbf{A}_t)\rightarrow(\mathbf{\Gamma}'_t,\mathbf{A}'_t)$.
    \item A memory mapping $\loc_t\colon \mathrm{Var}\rightarrow \mathrm{List}(\Loc_q+ \Loc_c)$.
    \item A set $\mathrm{now}(t)\subseteq V$ of vertices that the thread is currently looking at.
    \item A conjunction $\mathrm{Control}(t)$ of guards $g_1\wedge\cdots\wedge g_n$
          that control the thread.
    \item A status $\mathrm{status}(t) \in \{\mathbf{Await}(x,\bar {t_i}),\ \mathbf{Ready},\ \mathbf{Check}(S),\ \mathbf{Terminated}(\mathbf{\Gamma})\}$;
          $\mathbf{Await}(x,\bar {t_i})$ means that the thread is waiting for the completion of $t_i$
          and $x$ is the variable to store the result of $t_i$;
          $\mathbf{Ready}$ means that the thread is ready to execute the next statement;
          $\mathbf{Check}(S)$ means that the thread is requesting confirmation of
          the state of the pebbles before executing the statement $S$;
          and $\mathbf{Terminated}(\mathbf{\Gamma})$ means that the thread has finished execution
          with the context $\mathbf{\Gamma}$.
  \end{itemize}
\end{definition}
Each thread has a unique identifier $t$, and a statement $S_t$ to execute.
The mapping $\loc_t$ from variables to locations is also owned by each thread,
which means that all the variables can be assumed to be local.
Each thread also has a set of vertices $\mathrm{now}(t)$.
This is the set of vertices that the pebbles should be on if the pebble game is played without delay,
corresponding to the most eager uncomputation strategy in the simulation semantics.
The set $\mathrm{now}(t)$ can only include at most one vertex from each location.
\[|\mathrm{now}(t) \cap G_l| \leq 1\]
The set $\mathrm{Control}$ is a set of vertices that control the thread, which means that
this thread was spawned while executing a $\Qif$ branch which was controlled by variables $\overline{x_i}$
when the thread was assuming pebbles are on the vertices $\overline{v_i}$ with $\loc(x_i) = \loc(v_i)$.
The last component is the status of the thread.
The status $\mathbf{Check}(S)$ can only occur while executing $\Endlft \alpha$ or $x \As T$, $\Let y = U(x), \Meas(x)$,
and $\mathbf{Await}(x, \bar{t_i})$ can only occur while executing $\If$, $\Qif$, or a function call.
In fact, the length of the list $\bar{t_i}$ is 1 or 2;
1 for $\If$ statements and function calls, and 2 for $\Qif$ statements.

As with usual multi-threaded computing, all threads share the same memory space.
In fact, they are even permitted to operate on the same quantum memory address.
Even though the concurrently executed threads share the same memory space,
we show that the semantics is still thread-safe in the sense that no race conditions or data corruption can occur.
This is because any two operations from different threads are orthogonal in the sense that they are under opposite quantum control, in the sense that there must be a qubit $x$ which controls one thread positively and the other negatively.

The rule for execution of a statement in a thread is defined as follows,
where $\Theta$ is a thread, $\mathbf{A}$ is a lifetime preorder,
$(\ket\phi,s)$ is a state of quantum and classical memory, and $G$ is a circuit graph.
\begin{align}
  (\Theta,\mathbf{A},(\ket\phi,s), G)
  \rightarrow
  (\overline{\Theta'},\mathbf{A}',(\ket{\phi'},s'), G')
  \label{eq:unc-thread-exec}
\end{align}
The state of pebbles is invisible from threads.
Instead of knowing their positions, each thread assumes that the pebbles are on the vertices $\mathrm{now}(t)$.
Even if the pebbles are not on the vertices $\mathrm{now}(t)$,
the thread can still execute most of the statement except for $\Endlft \alpha$ and $\Let y = U(x),\Meas(x)$; for these statements, the thread must be executed after the pebbling state has been verified.
That is, the uncomputer only takes care of computation on locations which are affinely owned, and the thread must take care of the rest.
However, the threads do not know the pebbling state, and it is the scheduler's job to check the pebbling state before executing the statement.
Therefore, the thread changes its status to $\mathbf{Check}(S)$ and waits until the scheduler asserts that the pebbling condition is satisfied.

\paragraph{Scheduler}
A \emph{scheduler} decides which thread to execute next based on the status of the whole system.
The whole data of the system, including all the threads and the uncomputer, is called the \emph{system state}. In summary it consists of the following data.
\begin{definition}
  The \emph{system state} $\mathscr{S}$ consists of
  \begin{itemize}
    \item A set of threads $\overline\Theta = \Theta_{t_1},\dots,\Theta_{t_n}$,
    \item A circuit graph $G$,
    \item A lifetime preorder $\mathbf{A}$,
    \item States of the quantum and classical memory $(\ket\phi, s)$.
    \item A state of pebbling ${\{\mathrm{peb}(v)\}}_{v\in V}$.
  \end{itemize}
\end{definition}
The operational semantics is defined by the relation $\mathscr{S}\rightarrow^*\mathscr{S}'$ between system states,
where the relation $\to$ is defined by the following three rules.
(The same rules can be found in~\Cref{fig:rule/scheduler})
  {
    \footnotesize
    \begin{mathpar}
      \inferH{Scheduler Uncomputer}{
        \mathbf{A},G\models(\ket\phi, \mathrm{peb})
        \Rrightarrow
        (\ket{\phi'}, \mathrm{peb}')
      }{
        (\overline\Theta,\mathbf{A}, G, (\ket\phi, s), \mathrm{peb})
        \rightarrow
        (\overline\Theta, \mathbf{A}, G, (\ket{\phi'}, s), \mathrm{peb}')
      }

      \inferH{Scheduler Exec}{
        G, \mathrm{peb}\models \Theta\colon\text{executable} \\
        (\Theta,\mathbf{A},(\ket\phi,s), G)
        \rightarrow
        (\overline{\Theta'},\mathbf{A}',(\ket{\phi'},s'), G')
      }{
        \mathscr{S} = ((\overline\Theta, \Theta),\mathbf{A}, G, (\ket\phi, s), \mathrm{peb})
        \rightarrow
        ((\overline\Theta + \overline{\Theta'}),\mathbf{A}', G', (\ket{\phi'}, s'), \mathrm{peb})
      }

      \inferH{Scheduler Await}{
        \mathrm{state}(t)=\mathbf{Await}(x,t_{1},\dots,t_{m})\\
        \mathrm{state}(t_i)=\mathbf{Terminated}(\mathbf{\Gamma}_i)\\
        (\Theta_t, G)
        \xrightarrow{(x,\mathbf{A},
          \Theta_{t_{1}},
          \dots,
          \Theta_{t_{n}}
          )}
        (\Theta'_t, G')
      }{
        ((\overline\Theta,\Theta_t,\overline{\Theta_{t_i}}),\mathbf{A}, G, (\ket\phi, s), \mathrm{peb})
        \rightarrow
        ((\overline\Theta,\Theta'_t),\mathbf{A}, G', (\ket{\phi'}, s'), \mathrm{peb})
      }
    \end{mathpar}
  }
The second rule \scref{Scheduler Exec} and the third rule \scref{Scheduler Await}
include the following judgements which we have not yet introduced.
\begin{align}
  G,\mathrm{peb}\models \Theta\colon\text{executable}
  \label{eq:unc-executable} \\
  (\Theta_t, G) \xrightarrow{(x,\mathbf{A}, \Theta_{t_1}, \dots \Theta_{t_n})} (\Theta'_t, G')
  \label{eq:unc-thread-await}
\end{align}

The judgement \cref{eq:unc-executable} is the scheduler's judgement of the executability of thread $\Theta$.
It checks the pebbling state on the circuit graph $(G, \mathrm{peb})$
and verifies that the pebbles are placed on the right vertices.
In fact, the only nontrivial cases are the cases when the thread is trying to execute
 an $\Endlft \alpha$ statement or a $\Let y = U(x), \Meas(x)$ statement.
In all other cases the thread status is $\mathbf{Ready}$, so the scheduler needs not check the pebbling state.
The cases when the status of the thread is $\mathbf{Await}(x,t_1,\dots)$ are dealt with in the next rule.

The last judgement \cref{eq:unc-thread-await} is an execution rule for waking up a parent thread $(\Theta_t)$
that has been waiting for the child threads $(\Theta_{t_1},\dots,\Theta_{t_n})$ to terminate.
The parent thread stores the return value of the children threads in the variable $x$.

\subsection{The Rules for the Uncomputation Semantics}

Before defining the semantics, for simplicity, we would like to make some assumptions that do not lose generality.
For circuit graphs, we assume the following.
\begin{itemize}
  \item A circuit graph is acyclic.
  \item For each vertex $v$, a gate vertex $w$ with a unique incoming edge $v\to w$
        which represents $X$ gate is called a \emph{negation} of $v$.
        We denote it by $\neg v$. We also call the vertex $v$ a \emph{negation} of $w$.
  \item Each vertex has a unique negation. This defines an involution $\neg:V\to V$.
  \item Every non-empty induced subgraphs $G_l$ has a unique vertex in $V_\text{init}$.
  \item For each gate vertex, all sources of incoming edges have different locations.
  \item Each linear vertex has a unique outgoing edge to its negation, and the negation does not have any outgoing edges.
\end{itemize}

For statements and memory allocation, we assume the following.
\begin{itemize}
  \item All statements in blocks have $\Noop$ at its end, as in $\{S;\Noop;x\}$.
  \item In statements $\Let x = \Qif x \{ S_{\ket1}; x_1 \} \Else \{ S_{\ket0}; x_0 \}$, the locations of $x_0$ and $x_1$ are the same.
        For example, we prohibit the statement $\Let x = \Qif x \{ (a,b) \} \Else \{ (b,a) \}$
        but only accept a statement with explicit swapping
        $\Let x = \Qif x \{ (a,b) \} \Else \{ \texttt{[swap]}(a,b) \}$.
        This does not lose generality because we can always insert a swap gate to match the locations.
  \item No quantum location can be used after it has been measured.
\end{itemize}

Let us first define the judgement $\mathbf{A}, G, \mathrm{peb},\{p_i\}_i^n\models g_1\wedge\cdots\wedge g_n$
where $\mathbf{A}$ is a lifetime preorder, $G$ is a circuit graph, $\mathrm{peb}$ is a pebbling state, $\bar p$ is a set of labels of pebbles, and $g_i = (v_i, \lft{a}_i)$ are guards.
This judgement is valid when the following is true: $\lft{a}_i\in\mathbf{A}$ is satisfied,
a pebble $p_i\{\phi_i\}$ is on the vertex $v_i$,
and the following formula is provable in classical predicate logic.
\[
  \left(\bigwedge_{i=1}^{n}
  \left(\phi_i \Rightarrow v_i\right)\right)
  \Rightarrow v_1\wedge\cdots\wedge v_n.
\]
Intuitively, this judgement represents a prerequisite of the pebbling state for applying a gate controlled by $v_1,\dots,v_n$.
For example, when we have a whole pebble $p\{\top\}$ on $v$ and a fragmented pebble $p'\{(v,\alpha)\}$ on $w$,
this satisfies the judgement $\mathbf{A}, G, \mathrm{peb},\{p,p'\}\models (v,\alpha)\wedge(w,\alpha)$
if $\alpha \in \mathbf{A}$ because $v\wedge w$ can be deduced from $(\top \Rightarrow v) \wedge (v\Rightarrow w)$.

\paragraph{Rules for the Uncomputer}
\Cref{fig:rule/unc-uncomputer} shows the rules $\Rrightarrow$ for the pebble game.

\begin{figure}[t]
  \hfill
  \fbox{$\mathbf{A},G\models(\ket\phi, \mathrm{peb}) \Rrightarrow (\ket{\phi'}, \mathrm{peb}')$}
  \footnotesize
  \begin{mathpar}
    \inferH{pebble identity}{}{
      \mathbf{A},G\models(\ket\phi, \mathrm{peb}) \Rrightarrow (\ket\phi, \mathrm{peb})
    }

    \inferH{pebble composition}{
      \mathbf{A},G\models(\ket\phi, \mathrm{peb}) \Rrightarrow (\ket{\phi'}, \mathrm{peb}')\\
      \mathbf{A},G\models(\ket{\phi'}, \mathrm{peb}) \Rrightarrow (\ket{\phi''}, \mathrm{peb}'')
    }{
      \mathbf{A},G\models(\ket\phi, \mathrm{peb}) \Rrightarrow (\ket{\phi''}, \mathrm{peb}'')
    }

    \inferH{pebble split}{
      P = P_1 + P_2\\
    }{
      \mathbf{A},  G\models(\ket\phi, \mathrm{peb}[v\mapsto S+\{P\}]) \Lleftarrow\Rrightarrow (\ket\phi, \mathrm{peb}[v\mapsto S+\{P_1,P_2\}])
    }

    \inferH{pebble init}{v\in V_\text{init}\\
      p\in \mathrm{Label} \colon\text{fresh}\\
      \loc(v) = \loc(p)\\
      \ket\psi = \ket\phi\ket0 \in \HH_\text{others}\otimes\HH_{p}\\
    }{
      \mathbf{A}, G\models(\ket\phi, \mathrm{peb}[v\mapsto \emptyset])
      \Lleftarrow\Rrightarrow
      (\ket\psi, \mathrm{peb}[v\mapsto \{p\{\top\}\}])
    }

    \inferH{pebble gate}{
      w \in V_\text{gate}\\
      P = p\{g_1\wedge\cdots\wedge g_m\}\\
      \{v_0,v_1,\dots,v_n\} = \mathrm{source}(w)\\
      U_w = v_0 \oplus (v_1\wedge\cdots\wedge v_n) \\
      p \not\in \{q_i\}\\
      \mathbf{A}, G, \mathrm{peb}, \{q_i\} \models g_1\wedge\cdots\wedge g_m\wedge (v_1,\top)\wedge \cdots \wedge (v_n,\top)\\
      \ket\psi = ( \id \otimes C^{n+m}X )\ket\phi \text{ where the controlls are }q_i
      \text{, and the target is } p\\
    }{
      \mathbf{A}, G\models(\ket\phi, \mathrm{peb}[v_0\mapsto S+\{P\}, w\mapsto T])
      \Lleftarrow\Rrightarrow
      (\ket\psi, \mathrm{peb}[v_0\mapsto S, w\mapsto T+\{P\}])
    }

    \inferH{pebble copy/delete}{
      l\in \Loc_q\\
      v_0, v \in G_l\\
      v_0 \in V_\text{init}\\
      v,\neg v \not\in V_\text{linear}\\
      P = p\{\bar{g_i}\}\\
      P' = p'\{\bar{g_i}\}\\
      p,p' \not\in \{q_i\}\\
      \mathbf{A}, G, \mathrm{peb}, \{q_i\}\models g_1\wedge\cdots\wedge g_n\\
      \ket\psi = (\id\otimes C^{n+1}X)\ket\phi \text{ where the target is } p'
      \text{ and the controls are }{q_i} \text{ and } p.
    }{
      \mathbf{A}, G \models (\ket\phi, \mathrm{peb}[v_0 \mapsto S + \{P'\}, v\mapsto T + \{ P \}])
      \Lleftarrow\Rrightarrow
      (\ket\psi, \mathrm{peb}[v_0 \mapsto S, v\mapsto T + \{P,P'\}])
    }

    \inferH{pebble merge guard}{
      w\in V_\text{merge}\\
      i\in \{0,1\}\\
      (e_i\colon v_i\to w) \in G \text{ is the }i\text{-th incoming edge of } w\\
      g \colon\text{guard of } w\\
      P = p\{\neg^{1-i}g\wedge g_1\wedge\cdots\wedge g_n\}\\
    }{
      \mathbf{A}, G\models(\ket\phi, \mathrm{peb}[v_i\mapsto S+\{P\}, w\mapsto T])
      \Lleftarrow\Rrightarrow
      (\ket\phi, \mathrm{peb}[v_i \mapsto S, w\mapsto T + \{P\}])
    }

    \inferH{pebble linear guard}{
      w \in V_\text{liner}\\
      p \in \mathrm{MainLabel}\\
      (v\to w) \in G\\
      \bigwedge_i g_i\colon\text{guards of this edge}\\
      P = p\{\bigwedge_i g_i \wedge \bigwedge_j g'_j\}\\
    }{
      \mathbf{A}, G\models(\ket\phi, \mathrm{peb}[v\mapsto S+\{P\}, w\mapsto T])
      \Rrightarrow
      (\ket\phi, \mathrm{peb}[v\mapsto S, w\mapsto T+\{P\}])
    }
  \end{mathpar}
  \Description{}
  \caption{Rules for uncomputer $\Rrightarrow$ which playes a pebble game}
  \label{fig:rule/unc-uncomputer}
\end{figure}

The first two rules, \scref{pebble identity} and \scref{pebble composition},
are the reflexivity and transitivity of the relation $\Rrightarrow$.

The rule \scref{pebble split} splits a pebble into two,
with addition defined conjunctively as
$p\{g\wedge g'_1\wedge\cdots\wedge g'_n\} + p\{\neg g\wedge g'_1\wedge\cdots\wedge g'_n\} = p\{g'_1\wedge\cdots\wedge g'_n\}$.

The rule \scref{pebble init} places a pebble on an initial vertex.
This also initializes the corresponding qubit in the quantum memory.
The assumption $p\in \mathrm{Label} \colon\text{fresh}$ means that
there is no pebble with label $p$ on the graph.

The rule \scref{pebble gate} applies a single target gate to a qubit.
The assumption says that when the gate is $U_w = v_0 \oplus (v_1\wedge\cdots\wedge v_n)$,
and the pebble $p$ is fragmented by the guard $g_1\wedge\cdots\wedge g_m$,
we need to check if the pebbling state is ready to apply $C^{n+m}X$ gate ($X$ gate controlled by $n+m$ qubits).
Note that the negation is a special case of this without any controls.

The rule \scref{pebble copy/delete} moves a pebble
from the initial vertex to another vertex that is already pebbled.
The intuition is to change the state $\ket i\ket0$ to $\ket i\ket i$ by applying $CX$ gate.
When the pebble is fragmented by the guard $g_1\wedge\cdots\wedge g_n$,
we have to apply $C^{n+1}X$ gate, and the assumption of the rule guarantees the pebbling state is ready for this.

The rule \scref{pebble merge guard} moves a pebble onto a merge vertex.
It requires that the pebble is fragmented by the guard of the merge vertex.
Note that we do not need to check whether the lifetime of the guard has not ended.
The aliveness of lifetime is checked only when we apply a gate, which is enough to ensure correctness.

The rule \scref{pebble linear guard} moves a pebble onto a linear vertex.
This also requires the pebble to be fragmented by the guards.
For simplicity, we do not allow any auxiliary qubits to have a state corresponding to the linear vertex.
Therefore, a pebble can only move with the main label to the linear vertex.
This rule is the only one which is not bidirectional, making the pebble game irreversible.

\paragraph{Rules for Executable Threads}
\Cref{fig:rule/unc-executable} defines the executability rules for threads.

\begin{figure}[t]
  \hfill
  \fbox{$G,\mathrm{peb}\models \Theta\colon\text{executable}$}
  \footnotesize
  \begin{mathpar}
    \infer{\mathrm{state}(t)=\mathbf{Ready}}{
      \mathscr{S}\models \Theta_t\colon\text{executable}
    }

    \infer{\mathrm{state}(t)=\mathbf{Check}(S)\\
      S = \Endlft \alpha, x \As T\colon (\mathbf{\Gamma},\mathbf{A})\to(\mathbf{\Gamma}',\mathbf{A}')\\
      P = \{ p \in \mathrm{MainLabel} \mid \loc_t(p) \text{ is affinely owned in }
      (\mathbf{\Gamma}, \mathbf{A})
      \text{ but not in }
      (\mathbf{\Gamma}', \mathbf{A}')\}\\
      \forall p \in P,\text{ the pebble } p\{\mathrm{Control}(t)\} \text{ is on a vertex in } V_\text{linear}\\
    }{
      G,\mathrm{peb}\models \Theta_t\colon\text{executable}
    }

    \infer{\mathrm{state}(t)=\mathbf{Check}(\Let y = \Meas(x))\\
    \text{In }G_{\loc_t(x)},\ \text{a whole pebble}\ p\{\top\} \text{ is on }
    v \in G_{\loc_t(x)} \cap V_\mathrm{linear}\\
    }{
    G,\mathrm{peb}\models \Theta_t\colon\text{executable}
    }

    \infer{\mathrm{state}(t)=\mathbf{Check}(\Let y = U(x))\\
    (v_1,-)\wedge\cdots\wedge(v_n,-) = \mathrm{Control}(t)\\
    \{w_1,\dots,w_m\} = \{ w \in V\mid l\in \loc_t(x),\ w \in G_l\cap V_\text{linear}\}\\
    \text{No pebble is fragmented by }w_i\text{ or }\neg w_i\\
    \{p_i\}_{i=1}^n, \{q_j\}_{j=1}^m \subseteq \mathrm{MainLabel}\\
    \loc(p_i) = \loc_t(v_i)\\
    \loc(q_j) = \loc_t(w_j)\\\\
    \mathbf{A}, G,\mathrm{peb}, \{p_i\}\models \mathrm{Control}(t)\\
    q_j\{\mathrm{Control}(t)\} \text{ is on } w_j\\
    }{
    G,\mathrm{peb}\models \Theta_t\colon\text{executable}
    }
  \end{mathpar}
  \Description{}
  \caption{Rules for executable threads}
  \label{fig:rule/unc-executable}
\end{figure}

The first rule says that any thread with the status $\mathbf{Ready}$ is executable,
and the rest of the rules are for the status $\mathbf{Check}(S)$.

When $S = \Endlft \alpha$ or $S = x \As T$, before executing the statement,
we need to check if all the linear vertices whose location is linearly owned are pebbled.

When $S = \Let y = \Meas(x)$, we need to know if the whole main pebble is on the linear vertex,
and not on its negation.

The case $S = \Let y = U(x)$ is a bit more complicated since it can occur under control of some qubits, and does not commute with controlled gates.
The rule checks three things:
\begin{itemize}
  \item The target qubits are ready to apply the gate.
  \item The control qubits are pebbled.
  \item All computations that need to be done before the gate have been done.
\end{itemize}
The first condition checks that the linear vertices $w_j$,
whose location $\loc_t(w_j)$ is where we want to apply $U$,
is pebbled by $q_j\{\mathrm{Control}(t)\}$.
The second condition checks that the control qubits $v_i$ are pebbled by $p_i$s.
For convenience, we assume that $p_i$s are main labels.
The last condition checks that no pebbles on the graph are fragmented by $w_j$s.

\paragraph{Rules for Checked Statements}
\Cref{fig:rule/unc-checked-stmt} displays the rules for executing statements, which are already verified to be executable by the previous rules.

\begin{figure}[t]
  \hfill
  \fbox{$(\Theta_t,\mathbf{A},(\ket\phi,s), G) \rightarrow (\Theta'_t,\mathbf{A},(\ket{\phi'},s'), G)$}
  \footnotesize
  \footnotesize
  \begin{mathpar}
    \infer{\mathrm{status}(t) = \mathbf{Check}(S)\\
      S = \Endlft \alpha, x \As T\\
    }{
      \Theta_t \rightarrow\Theta_t[\mathrm{status}(t) = \mathbf{Ready}]\\
    }

    \infer{\mathrm{status}(t) = \mathbf{Check}(\Let y = \Meas(x))\\
      \ket\psi = (\id\otimes \bra0)\ket\phi \text{ where the target is the main label whose location is } \loc_t(x)\\
      l_c\colon \text{fresh classical location}\\
    }{
      \Theta_t, \ket\phi, s
      \rightarrow
      \Theta_t[\mathrm{status}(t) = \mathbf{Ready}, \loc_t[x\mapsto[], y\mapsto[l_c]]],\ket\psi, s[l_c \leftarrow 0]
    }

    \infer{\mathrm{status}(t) = \mathbf{Check}(\Let y = \Meas(x))\\
      \ket\psi = (\id\otimes \bra1)\ket\phi \text{ where the target is the main label whose location is } \loc_t(x)\\
      l_c\colon \text{fresh classical location}\\
    }{
      \Theta_t, \ket\phi, s
      \rightarrow
      \Theta_t[\mathrm{status}(t) = \mathbf{Ready}, \loc_t[x\mapsto[], y\mapsto[l_c]]],\ket\psi, s[l_c \leftarrow 1]
    }

    \infer{\mathrm{status}(t) = \mathbf{Check(\Let y = U(x))}\\
      (v_1,-)\wedge\cdots\wedge (v_n,-) =\mathrm{Control}(t)\\\\
      P = \{ p\in \mathrm{MainLabel} \mid \loc(v_i) \in \loc_t(q) \}\\
      Q = \{ q \in \mathrm{MainLabel} \mid \loc(p) = \loc(x)\}\\\\
      \ket\psi = (\id\otimes C^nU)\ket\phi \text{ where the controlls are } P \text{ and the targets are } Q\\
    }{
      \Theta, \ket\phi
      \rightarrow
      \Theta_t[\mathrm{status}(t) = \mathbf{Ready}],\ket\psi
    }
  \end{mathpar}
  \Description{}
  \caption{Execution rules for checked statement}
  \label{fig:rule/unc-checked-stmt}
\end{figure}

When the statement is $\Endlft \alpha$ or $x \As T$, we do not need to do anything.

When the statement is $\Let y = \Meas(x)$, we measure the qubit $x$ and store the result in the classical memory.
This is a probabilistic operation which has two possible outcomes, 0 or 1.

When the statement is $\Let y = U(x)$, we apply the unitary $U$ to the qubit $x$ controlled by the guards.

\paragraph{Rules for waking up threads}
\Cref{fig:rule/unc-merge} defines the rules for threads that are waiting for a child thread to terminate.
The judgement has the form
$(\Theta_t, G)\xrightarrow{(y,\mathbf{A},\Theta_{t_1},\dots, \Theta_{t_n})} (\Theta'_t, G')$
which means that the thread $\Theta_t$ is waking up with the terminated the child threads $\Theta_{t_1},\dots,\Theta_{t_n}$.

\begin{figure}[t]
  \hfill
  \fbox{$(\Theta_t, G)\xrightarrow{(y,\mathbf{A},\Theta_{t_1},\dots, \Theta_{t_n})} (\Theta'_t, G')$}
  \footnotesize
  \begin{mathpar}
    \infer{n = 1\\
      \mathbf{Terminated}(x\colon T) = \mathrm{state}(t_1)\\
      \Theta'_t = \Theta_{t}[
        \loc_{t_1}[y\mapsto \loc_{t_1}(x)],\
        \mathrm{now}(t) = \mathrm{now}(t_1),\
        \mathrm{status}(t) = \mathbf{Ready}
      ]
    }{
      (\Theta_t, G) \xrightarrow{(y, \mathbf{A}, \Theta_{t_1})} (\Theta'_t, G)
    }

    \infer{
      n = 2\\
      \mathbf{Terminated}(x_i \colon T_i) = \mathrm{state}(t_i)\\
      g \colon \text{guard of conditional qubit of the } \Qif\\
      L = \loc_{t_1}(x_1) = \loc_{t_2}(x_2)\\
      L_\text{affine} = \{l\in L \mid l :\text{affinely owned by } \mathbf{\Gamma}_i \text{ at } \mathbf{A}\}\\
      \text{For each }l\in L_\text{affine},\ v_l^i \in \mathrm{now}(t_i) \cap G_l,
      \text{ and } w_l \text{ is a new vertex in } V_\text{merge} \cap G_l \text{ with the guard }g.\\
      G ' = G + \{ e \colon v_l^i \to w_l \mid l\in L_\text{affine}\}\\
      \mathrm{now}'(t) \Coloneqq \mathrm{now}(t_1) - \{v^1_l\} + \{w_l\} = \mathrm{now}(t_2) - \{v^2_l\} + \{w_l\}\\
      \Theta_t' = \Theta_t[
        \mathbf{Ready},\
        \loc_{t_i}[y\mapsto L],\
        \mathrm{now}'(t)
      ]\\
    }{
      (\Theta_t, G)\xrightarrow{(y, \mathbf{A},\Theta_{t_1}, \Theta_{t_2})} (\Theta'_t, G')
    }
  \end{mathpar}
  \Description{}
  \caption{Rules for waking up threads}
  \label{fig:rule/unc-merge}
\end{figure}

The first rule is relatively simple.
When the parent spawned only one child thread, and the child thread has terminated,
the parent thread wakes up and stores the returned value in the variable $y$.
It also reflects the changes in the set \textrm{now}  of vertices that the thread is currently working on, and sets the status of the thread to $\mathbf{Ready}$.

The second rule governs the case when the parent spawned two child threads.
In this case, we might have to create new merge vertices to merge the results of the child threads.
This rule appears to be more complicated, but all it is doing is adding merge vertices.

\paragraph{Rules for the execution of expressions (without new threads)}
\Cref{fig:rule/unc-expr-no-new-thread} are the rules for the execution of expressions without spawning new threads.

\begin{figure}[t]
  \hfill
  \fbox{$(\Theta_t,\mathbf{A},(\ket\phi,s), G) \rightarrow (\Theta_t',\mathbf{A},(\ket\phi',s'), G')$}
  \footnotesize
  \begin{mathpar}
    \infer{e=x\\\mathrm{status}(t) = \mathbf{Ready}}{
      \Theta_t[S_t = (\Let y = e; S')]
      \rightarrow
      \Theta_t[S_t = S', \loc_t[x\mapsto[], y\mapsto \loc_t(x)]]
    }

    \infer{e = b\in\{\False,\True\}\\\mathrm{status}(t) = \mathbf{Ready}\\l_c \colon \text{fresh classical location}}{
      \Theta_t[S_t = (\Let y = e; S')], s
      \rightarrow
      \Theta_t[S_t = S', \loc_t[y\mapsto l_c]], s[l_c \leftarrow b]
    }

    \infer{e = ()\\\mathrm{status}(t) = \mathbf{Ready}}{
      \Theta_t[S_t = (\Let y = e; S')]
      \rightarrow
      \Theta_t[S_t = S']
    }

    \infer{e=(x_0,x_1)\\\mathrm{status}(t) = \mathbf{Ready}}{
      \Theta_t[S_t = (\Let y = e; S')]
      \rightarrow
      \Theta_t[S_t = S', \loc_t[x_0,x_1\mapsto[], y\mapsto \loc_t(x_0) + \loc_t(x_1)]]
    }

    \infer{e=\Copy x\\\mathrm{status}(t) = \mathbf{Ready}\\
      L_c = \{\ l_c\in \Loc_c\ |\ l_c: \text{owned by}\ x\ \}\\
      L'_c : \text{fresh classical locations}\\
      L' = \loc(x)[L'_c/L_c]\\
      t = s(L_c)
    }{
      \Theta_t[S_t = (\Let y = e; S')], s
      \rightarrow
      \Theta_t[S_t = S', \loc_t[y\mapsto L']], s[t\leftarrow L'_c]
    }

    \infer{e=\Meas(x), U(x)\\\mathrm{status}(t) = \mathbf{Ready}}{
      \Theta_t[S_t = (\Let y = e; S')]
      \rightarrow
      \Theta_t[S_t = S', \mathrm{status}(t) = \mathbf{Check(\Let y = e)}]
    }

    \infer{e=\ket0\\
      \mathrm{status}(t) = \mathbf{Ready}\\
      l_q \colon \text{fresh location for }\Theta_t\\
      v\colon \text{possibly fresh vertex in }V_\text{init} \text{ whose location is }l_q\\
      G' = G \cup \{v \to \neg v\}
    }{
      \Theta_t[S_t = (\Let y = e; S')], G
      \rightarrow
      \Theta_t[S_t = S', \loc_t[y\mapsto l_q], \mathrm{now}(t) + \{v\}], G'
    }

    \infer{e=x_0\oplus(\neg^{i_1}x_1\wedge\cdots\wedge \neg^{i_n}x_n)\\
      \mathrm{status}(t) = \mathbf{Ready}
      \\[l_i] = \loc_t(x_i)\\
      v_i \in\mathrm{now}(t) \cap G_i\\\\
      v_0 \in V_\text{gate}\\
      U_{v_0} = w \oplus (\neg^{i_1}v_1\wedge\cdots\wedge \neg^{i_n}v_n)
    }{
      \Theta[S_t = (\Let y = e; S'), \mathrm{now}(t) = S + \{v_0\}] \rightarrow
      \Theta[S_t = S', \mathrm{now}(t) = S + \{w\}]
    }

    \infer{e=x_0\oplus(\neg^{i_1}x_1\wedge\cdots\wedge \neg^{i_n}x_n)\\
      \mathrm{status}(t) = \mathbf{Ready}
      \\[l_i] = \loc_t(x_i)\\
      v_i \in\mathrm{now}(t) \cap G_i\\\\
      w\colon \text{fresh vertex in }V_\text{gate} \text{ whose location is }l_0
      \text{ with single target gate } U_w =
      v_0 \oplus (\neg^{i_1}v_1\wedge\cdots\wedge \neg^{i_n}v_n)\\
      G' = G + \{v_i \to w \to \neg w\}
    }{
      \Theta[S_t = (\Let y = e; S'), \mathrm{now}(t) = S + \{v_0\}],G \rightarrow
      \Theta[S_t = S', \mathrm{now}(t) = S + \{w\}], G'
    }
  \end{mathpar}
  \Description{}
  \caption{Rules of execution for $S_t = (\Let y = e; S')$ that do not create new thread}
  \label{fig:rule/unc-expr-no-new-thread}
\end{figure}

The rules for the expressions $e = x, b \in \{\False,\True\}, (), (x_0,x_1), \Copy x$ are relatively simple
-- they just replicate the rules in simulation semantics.

The rule for the expression $\Meas(x), U(x)$ is also simple.
Since they need checking before the execution, we simply change the status of the thread to $\mathbf{Check}$ to indicate that.

For the expression $\ket0$, we may need to create a new vertex $v$ in the initial vertices $V_\text{init}$.
The meaning of $l$ being a ``fresh location'' in this rule is that it is not used in the thread $\Theta_t$,
that is, $l$ is not in the image of $\loc_t$ for any variable.
However, the location $l$ does not need to be a new location in the whole system.
The following example code has two threads introducing a fresh location $l$,
but the latter thread to be executed must use the same location as the former thread.
\begin{center}
  \begin{minipage}{0.7\textwidth}
    \begin{lstlisting}[language=Qurts]
      let y = qif x { let y1 = $\ket0$; y1 } else { let y0 = $\ket0$; y0 }
    \end{lstlisting}
  \end{minipage}
\end{center}

The rule for the expression $x_0\oplus(\neg^{i_1}x_1\wedge\cdots\wedge \neg^{i_n}x_n)$ has two cases.
The latter case is the default, and creates a new gate vertex $w$ that represents $U_w$.
The former case is a special case using an existing vertex $w$ instead of creating a new one.
One of the examples we can simplify with this special rule is when we apply the same gate in sequence
such as $C^nX; C^nX$, and we can reduce it to the empty circuit.
Not only that, we can also know that $X(x); CX(x,y); X(x); CX(x,z)$ is equivalent to $CX(x,z); X(x); CX(x,z); X(x)$
since after applying $X$ gate twice to the same qubit, it comes back to the original state.
With this rule, we can detect repeatedly used states as in Reqomp~\cite{Reqomp_2024}.

\paragraph{Rules for the execution of expressions (with new threads)}
\Cref{fig:rule/unc-expr-new-thread} gives the rules for the execution of a thread that creates new threads.

\begin{figure}[t]
  \hfill
  \fbox{$(\Theta_t,\mathbf{A},(\ket\phi,s), G) \rightarrow (\overline{\Theta'},\mathbf{A},(\ket\phi,s'), G)$}
  \footnotesize
  \begin{mathpar}
    \infer{e=f\langle\alpha_0,\dots,\alpha_{m-1}\rangle(x_0,\dots,x_{n-1}{.})\\
    \mathrm{status}(t) = \mathbf{Ready}\\
    t'\colon \text{fresh thread id}\\
    \Fn f \mathbf{A}' \mathbf{\Gamma}' \rightarrow T \{ S_f; y_f \}\colon \text{ function definition}\\
    \Theta'_t = \Theta_t[S_t = S', \mathrm{status}(t) = \mathbf{Await}(y,t')]\\
    \Theta_{t'} = \Theta_t[
    t'\colon \text{thread id},\
    S_{t'} = S_f[\alpha'_i/\alpha_i] ,\
    \loc_{t'} = \loc_t[x_i\mapsto[], x'_i \mapsto \loc_t(x_i)],\
    \mathrm{status}(t) = \mathbf{Ready}
    ]\\
    }{
    \Theta_t[S_t = (\Let y = e; S')] \rightarrow
    (\Theta'_t, \Theta_{t'})
    }

    \infer{e=\If\ x\ \{S_t;y_t\}\ \Else\ \{S_f;y_f\}\\
      \mathrm{status}(t) = \mathbf{Ready}\\
      t' \colon \text{fresh thread id}\\
      s(\loc_t(x)) = b \in \{\False, \True\}\\\\
      \Theta'_t = \Theta_t[S_t = S', \mathrm{status}(t) = \mathbf{Await}(y, t_b)]\\
      \Theta_{t'} = \Theta_t[t'\colon\text{thread id}, S_{t'} = S_b, \mathrm{status}(t) = \mathbf{Ready}]\\
    }{
      \Theta_t[S_t = (\Let y = e; S')] \rightarrow
      (\Theta'_t, \Theta_{t'})
    }

    \infer{e=\Qif\ x\ \{S_\ket{1};y_{\ket 1}\}\ \Else\ \{S_\ket{0};y_\ket0\}\\
    \mathrm{status}(t) = \mathbf{Ready}\\
    v_x = \mathrm{now}(t) \cap G_{\loc_t(x)}\\
    v_x, \neg v_x\not\in \mathrm{Control}(t)\\
    t_0, t_1\colon \text{fresh thread id}\\
    \Theta'_t = \Theta_t[S_t = S', \mathrm{status}(t) = \mathbf{Await}(y, t_0,t_1)]\\
    \Theta_{t_i} = \Theta_t[
      t_i \colon \text{thread id},\
      S_{t_i} = S_{\ket{i}},\
      \loc_{t_i},\
      \mathrm{status}(t) = \mathbf{Ready},\
      \mathrm{Control}(t) \wedge \neg^{1-i} v_x
    ]\\
    s = (s_0,s_1) \in \{0,1\}^\text{others} \times \{0,1\}^\text{owned at the beginning of the branches}
    }{
    \Theta_t[S_t = (\Let y = e; S')], s \rightarrow
    (\Theta'_t, \Theta_{t_0}, \Theta_{t_1}), (s, s_1, s_1)
    }

    \infer{e=\Qif\ x\ \{S_\ket{1};y_{\ket 1}\}\ \Else\ \{S_\ket{0};y_\ket0\}\\
    \mathrm{status}(t) = \mathbf{Ready}\\
    v_x \in \mathrm{now}(t) \cap G_{\loc_t(x)}\\
    i \in \{0,1\}\\
    \neg^{1-i} v_x \in \mathrm{Control}(t)\\
    t'\colon \text{fresh thread id}\\
    \Theta'_t = \Theta_t[S_t = S', \mathrm{status}(t) = \mathbf{Await}(y, t')]\\
    \Theta_{t'} = \Theta_t[
      t' \colon \text{thread id},\
      S_{t'} = S_{\ket{i}},\
      \mathrm{status}(t) = \mathbf{Ready}
    ]\\
    }{
    \Theta_t[S_t = (\Let y = e; S')] \rightarrow
    (\Theta'_t, \Theta_{t'})
    }
  \end{mathpar}
  \Description{}
  \caption{Rules of execution for $S_t = (\Let y = e; S')$ that create new thread}
  \label{fig:rule/unc-expr-new-thread}
\end{figure}

The first two rules are for the expression of function application and conditional branching.
They create a new thread that executes the function body or the branch body,
and set the status of the current thread to $\mathbf{Await}$.

The third and fourth rules govern quantum conditional branching.
We have two cases for the quantum conditional branching:
the first case is a normal case, or the latter case is a special case where only one branch is executed.

First consider the normal case.
In this case, we create two new threads $\Theta_{t_0}$ and $\Theta_{t_1}$ that execute the two branches.
The status of the current thread is set to $\mathbf{Await}$, and the status of the new threads are set to $\mathbf{Ready}$.
In the new threads, we record that the thread is controlled by the qubit $x$.
For mere convenience we copy the classical data for the new threads, as one thread might drop the data while the other thread still needs it.

The special case is when this $\Qif$ expression is nested, and the outer $\Qif$ expression is controlled by the same qubit.
In this case, it only makes sense to execute one branch.
For example, $\Qif\ x\ \{\Qif\ x\ B_{\ket1}\ \Else\ B_{\ket0}\}$,
is equivalent to $\Qif\ x\ B_{\ket1}$.
Therefore, we create only one new thread $\Theta_{t'}$ that executes the branch $B_{\ket i}$,

\paragraph{Rules for the execution of statements}
See \Cref{fig:rule/unc-statement}.

\begin{figure}[t]
  \hfill
  \fbox{$(\Theta_t,\mathbf{A},(\ket\phi,s), G) \rightarrow (\Theta'_t,\mathbf{A}',(\ket\phi,s), G')$}
  \footnotesize
  \begin{mathpar}
    \infer{S_t = \Noop\colon (\mathbf{\Gamma},\mathbf{A})\rightarrow (\mathbf{\Gamma},\mathbf{A})\\
      \mathrm{status}(t) = \mathbf{Ready}\\
    }{
      \Theta_t
      \rightarrow
      \Theta_t[\mathrm{status}(t)=\mathbf{Terminated}(\mathbf{\Gamma})]
    }

    \infer{S = \Newlft \alpha\\
      \mathrm{status}(t) = \mathbf{Ready}\\
      v' \text{ is a new for each }v\in V_\text{linear}\\
    }{
      \Theta_t[S_t = S; S'], \mathbf{A}
      \rightarrow
      \Theta_t[S_t = S'], \mathbf{A} + \{\alpha\}
    }

    \infer{
      S = \Endlft \alpha\colon (\mathbf{\Gamma}_t,\mathbf{A}_t)\to(\mathbf{\Gamma}_t,\mathbf{A}'_t)\\
      \mathrm{status}(t) = \mathbf{Ready}\\
      W = \{ w \in V \mid w\in \mathrm{now}(t),\ \loc(w) \text{ is affinely owned in }
      (\mathbf{\Gamma}_t, \mathbf{A}_t),
      \text{ but not in }
      (\mathbf{\Gamma}_t, \mathbf{A}'_t)\}\\
      G' = G + \{ w \to w_\text{lin} \mid
      w \in W, w_\text{lin}\in G_{\loc(w)}\cap V_\text{linear},
      \}\\
      \text{The guards of the added edges are }\mathrm{Control}(t)
    }{
      \Theta_t[S_t = S; S'], \mathbf{A},G \rightarrow
      \Theta_t[S_t = S', \mathrm{status}(t) = \mathbf{Check}(S), \mathrm{now}(t) - W + \{w_\text{lin}\mid w
        \in W\}], \mathbf{A} - \{\alpha\}, G'
    }

    \infer{S = \alpha \leq \beta\\
      \mathrm{status}(t) = \mathbf{Ready}\\
    }{
      \Theta_t[S_t = S; S'], \mathbf{A}
      \rightarrow
      \Theta_t[S_t = S'], \mathbf{A} + \{\alpha \leq \beta\}
    }

    \infer{S = \Let y = \&^\alpha x\\
      \mathrm{status}(t) = \mathbf{Ready}\\
    }{
      \Theta_t[S_t = S; S']
      \rightarrow
      \Theta_t[S_t = S', \loc_t[y\mapsto \loc_t(x)]]
    }

    \infer{S = \Let (y_0,y_1) = x\\
      \mathrm{status}(t) = \mathbf{Ready}\\
      L_i \colon \text{ compatible locations to the type of } y_i
    }{
      \Theta_t[S_t = S; S', \loc_t[x\mapsto L_0 + L_1]]
      \rightarrow
      \Theta_t[S_t = S', \loc_t[x\mapsto[], y_0\mapsto L_0, y_1 \mapsto L_1]]
    }

    \infer{S = \Drop x\\
      \mathrm{status}(t) = \mathbf{Ready}\\
      W = \{w \in V \mid w \in \mathrm{now}(t),\ \loc(w) \text{ is a quantum location owned by } x \}\\
      L_c = \{\ l_c\in \Loc_c\ |\ \text{owned by}\ x\ \}\\
      s = (s_0,s_1) \in {\{0,1\}}^\text{others} \times {\{0,1\}}^{L_c}
    }{
      \Theta_t[S_t = S; S'], s
      \rightarrow
      \Theta_t[S_t = S', \loc_t[x\mapsto[]], \mathrm{now}(t) - W], s_0\\
    }

    \infer{S = x \As T\colon (\mathbf{\Gamma}_t,\mathbf{A}_t)\rightarrow(\mathbf{\Gamma}'_t,\mathbf{A}_t)\\
      \mathrm{status}(t) = \mathbf{Ready}\\
      W = \{ w \in V \mid w\in \mathrm{now}(t),\ \loc(w) \text{ is affinely owned in }
      (\mathbf{\Gamma}_t, \mathbf{A}_t)
      \text{ but not in }
      (\mathbf{\Gamma}'_t, \mathbf{A}_t)\}\\
      G' = G + \{ w \to w_\text{lin} \mid
      w \in W, w_\text{lin}\in G_{\loc(w)}\cap V_\text{linear},
      \}\\
      \text{The guards of the added edges are }\mathrm{Control}(t)
    }{
      \Theta_t[S_t = S; S'], G
      \rightarrow
      \Theta_t[S_t = S', \mathrm{status}(t) = \mathbf{Check}(S), \mathrm{now}(t) - W + \{w_\text{lin}\mid w
        \in W\}], G'
    }
  \end{mathpar}
  \Description{}
  \caption{Rules of execution for statements}
  \label{fig:rule/unc-statement}
\end{figure}

The first rule is the rule terminates the thread.
When the thread is terminated, the status of the thread is set to $\mathbf{Terminated}$.

The rules for the statements $\Endlft \alpha$ and $x \As T$ change the status of the thread to $\mathbf{Check}$.
They first check the locations which were affinely owned before the execution of the statement but not after it.
They collect those locations, and update the set \textrm{now}$(t)$.
After that, they change the status of the thread to $\mathbf{Check}$ to indicate that the thread is waiting
for the pebble state to be updated.

The rule for the statement $\Drop x$ forgets the variable $x$ from the function $\loc_t$.
It also frees the corresponding classical memory, but it does not free the quantum memory, because the uncomputation can be left to the uncomputer
in this uncomputation semantics.

The rest of the rules are straightforward.

\subsection{Proofs for \Cref{subsec:unc-semantics}}
This section proves the two theorems in \Cref{subsubsec:unc-semantics-prop}.
We only consider the system states which can be obtained by executing a closed program
in a single thread from the initial state.
Therefore, we only consider circuit graphs which can be constructed by the execution rules
from the empty graph, and also, we only consider pebbling states which can be obtained
by applying the pebbling rules from the empty pebbling state but ignoring lifetime constraints.

\begin{definition}
  We inductively define a family of boolean-valued functions $F_v\colon \{0,1\}^{S_v}\to\{0,1\}$
  for each vertex
  $v\in V$ satisfying $v,\neg v\not\in V_\text{linear}$ as follows.
  \begin{itemize}
    \item If $v\in V_\text{init}$, then $F_v$ is the constant function from the singleton set
          $0\colon \{*\}\to \{0,1\}$.
    \item If $v\in V_\text{gate}$, and the sources of incoming edges are $v_0,\dots,v_n$
          where $v_0$ is the target, then $F_v$ is  defined by
          $F_{v_0} \oplus \left(F_{v_1}\wedge\cdots\wedge F_{v_n}\right)
            \colon {\{0,1\}}^{\bigcup_i S_{v_i}} \to \{0,1\}$.
    \item If $v\in V_\text{merge}$, the sources of incoming edges are $v_0,v_1$,
          and the guard is $(w,-)$,
          then $F_v$ is defined by $S_v = S_{v_0} \cup S_{v_1} \cup \{w\}$ and
          $F_v(f) = \begin{cases}
              F_{v_0}(f\restriction_{S_{v_0}}) & \text{if } f(w) = 0 \\
              F_{v_1}(f\restriction_{S_{v_1}}) & \text{if } f(w) = 1
            \end{cases}$.
  \end{itemize}
\end{definition}

\begin{definition}
  We call a circuit graph $G$ \emph{valid} if for each
  for each function ${f} \colon V_\text{linear} \rightarrow \{0,1\}$,
  there is a unique mapping $\ev{v}_{f} \in \{0,1\}$ for each vertex $v\in V$ satisfing
  \[
    \ev{v}_{f} = \begin{cases}
      {f}(v)                     & \text{if }v\in V_\text{linear} \\
      F_v(\lambda w. \ev{w}_{f}) & \text{otherwise}
    \end{cases}.
  \]
\end{definition}

\begin{lemma}
  All the rules preserve validity of the circuit graph.
\end{lemma}
\begin{proof}
  The only non-trivial rule is adding merge vertices.
  Validity of the circuit graph is preserved in this case because
  only vertices that already exist can be used as a guard of a merge vertex.
  Thus we can inductively determine the value of the vertices.
\end{proof}

For example, if we have two merge vertices $v,w$ that guard each other,
then the circuit graph will not be valid, which does not happen in our semantics.

From now on, we only consider valid circuit graphs.

\begin{lemma}
  Let $f\colon V_\text{linear}\to\{0,1\}$ be a function, and $\mathrm{peb}$ be a pebbling.
  Then, for each label $p$, there is a unique fragmented pebble on the graph
  which includes $p\{(v_1,-)\wedge\cdots\wedge (v_n,-)\}$
  and satisfies $\ev{v_i}_{f} = 1$ for all $i = 1,\dots,n$.
  The vertex where the pebble is placed will be denoted by $v_{f,p}$.
\end{lemma}
\begin{proof}
  By splitting the pebbles with a label $p$ on the graph often enough, we may assume there is a set of guards $\{g_1,\dots,g_n\}$, and
  we have $2^n$ pebbles $p\{\neg^{i_1}g_1\wedge\cdots\wedge\neg^{i_n}g_n\}$
  for every combination of $i_1,\dots,i_n\in\{0,1\}$.
  There is a pebble that satisfies $\ev{\neg^{i_k}v_k}_{f} = 1$ for all $k = 1,\dots,n$.
\end{proof}

We call such a pebble $p\{(v_1,-)\wedge\cdots\wedge (v_n,-)\}$
and the vertex is \emph{focused} under $f$.

\begin{definition}\label{def:unc/valid-pair}
  A \emph{valuation} is a partial function $val\colon V_\text{linear}\pmap\{0,1\}$
  whose domain $W$ is maximal among those satisfying the following conditions.
  \begin{itemize}
    \item For each extended total function $f\colon V_\text{linear}\to\{0,1\}$,
          $\ev{v}_f$ gives the same value for each $v$ with $v,\neg v \not\in V_\text{linear}$.
          If this is the case, then $\ev{v}_f$ is denoted by $\ev{v}_{val}$.
    \item Let $w \in V$.
          Then $w\in W$ if and only if the main label $p_w$ whose location is $\loc(w)$
          satisfies $v_{p_w,f} \in V_\text{linear}$ or $\neg v_{p_w,f} \in V_\text{linear}$
          for all extended total functions $f$.
  \end{itemize}

  The pebbling is \emph{consistent} if for each valuation $val$,
  for all extension $f,f':V_\text{linear} \to \{0,1\}$ and for all labels $p$,
  the vertices $v_{f,p}$ and $v_{f',p}$ are connected by enabled guarding edges, that is,
  edges with targets in $V_\text{merge}$ or $V_\text{linear}$ and guards $(w_i,-)$
  satisfying $\ev{w_i}_{val} = 1$.

  When the pebbling is consistent, even though the vertex $v_{f,p} \not\in V_\text{linear}$
  might not be uniquely determined, the value $\ev{v_{f,p}}_{val}$ is unique.
  With abuse of notation we denote it as $\ev{v_{val,p}}_{val}$.

  The pair $(\mathrm{peb},\ket\phi)$ is \emph{valid} if $\mathrm{peb}$ is consistent,
  and if there is a family of complex numbers ${\{\gamma_{val}\}}_{val}$
  satisfying
  \[
    \ket\phi = \sum_{val}
    \gamma_{val} \bigotimes_{p\in P}  \ket{\ev{v_{val,p}}_{val}},
  \]
  for each valuation $val$ and label $p$.
\end{definition}

Intuitively, the domain of valuation is the set of the vertices where we might have applied arbitrary unitary gates.
For example, when executing the following code in line order, after line~3 is executed,
we have the following circuit graph.
Before executing line~4, the fragmented pebble $p\{(x,\alpha)\}$ has to be on the linear vertex $z$.
On the other hand, the pebble $p\{(\neg x, \alpha)\}$ cannot be on the vertex $z$ since there has not been any edge that the pebble can move along yet.
Even in such a situation that $z$ is not fully pebbled, we can apply an Hadamard gate to qubit $y$ in line~4,
which makes the state $\alpha\ket{00} + \beta\ket{10} + \gamma\ket{11} \in \HH_{x,y}$.
At this point, we have three valuations, $[x\mapsto 0]$, $[x\mapsto 0, y\mapsto 0]$, and $[x\mapsto 1, y\mapsto 1]$.
Therefore, the above definition of valuations lets us represent this state, where $\ket{01}$ cannot be included, but $\ket{10}$ and $\ket{11}$ can be.
\begin{center}
  \begin{minipage}{0.3\textwidth}
    \begin{lstlisting}[language=qurts]
    // x : &'a qbit
    let y = $\ket 0$;
    qif x {
      y as #'0 qbit;
      let y' = H(y); y'
    } else {
      y as #'0 qbit;
      y
    }
  \end{lstlisting}
  \end{minipage}
  \qquad
  \begin{minipage}{.4\textwidth}
    \[
      \begin{tikzpicture}[]
        \node[vertex] (x) at (-2.2,-.7) {$x$};
        \node[vertex] (y) at (0,0) {$y$};
        \node[vertex] (z) at (0,-1.5) {$z$};
        \draw[dashed] (-1.5,-.7) to node[pos=0.9,above]{{init}} node[pos=0.9,below]{{linear}} (2,-.7);
        \draw[dashed] (-1.5,-1.7) to (-1.5,.2);
        \draw[->,bend right] (y) to node[pos=.3,left]{$x=\ket1$} (z);
      \end{tikzpicture}
    \]
  \end{minipage}
\end{center}

\begin{lemma}\label{lem:unc/validity}
  Let $(\mathrm{peb}, \ket\phi)$ be a valid pair of a circuit graph $G$.
  If $\mathbf{A}, G \models (\mathrm{peb}, \ket\phi) \Rrightarrow (\mathrm{peb}',\ket\psi)$ is satisfied,
  then $(\mathrm{peb}',\ket\psi)$ is also valid,
  and $\ket\psi$ is uniquely determined by the others.
\end{lemma}
\begin{proof}
  First, we check that the consistency of the pebbling is preserved.
  The rules \scref{pebble split} and \scref{pebble init} obviously preserve the consistency.
  The rule \scref{pebble merge guard} is not also at issue here since
  it does not change any pebble on a vertex in $V_\text{linear}$,
  and can only move pebbles on the vertices connected by enabled guarding edges.
  The rest of the rules are a bit more complicated.

  The rule \scref{pebble linear guard} can change the set of the valuation.
  Since the number of pebbles on linear vertices can increase,
  a previously single valuation can be refined to multiple valuations extending the original partial function.
  However, refining can not make pebbling inconsistent,
  and because this rule can only move pebbles along enabled guarding edges,
  consistency is preserved.

  To check the rule \scref{pebble gate} and \scref{pebble copy/delete},
  first fix a valuation $val$ and two extended total functions $f,f'$.
  We assume that the pebbles $p\{\bar{g_i}\}$ and $p\{\bar{g'_i}\}$
  are the different pebbles focused on under $f$ and $f'$.
  Now, these pebbles cannot be moved by the rule
  \scref{pebble gate} or \scref{pebble copy/delete}  because all the vertices $v$ in the circuit graph satisfy $\ev{v}_{f} = \ev{v}_{f'}$
  except for the vertices outside the domain of $val$,
  so there must be some set of vertices $X\subseteq V_\text{linear} - \mathrm{dom}(val)$ that are used as a guard of $p\{\bar{g_i}\}$ either positively or negatively.
  Therefore, if we want to move the pebble $p\{\bar{g_i}\}$,
  we have to pebble those vertices first.
  That is to say, $f$ can be restricted to a map $X\to \{0,1\}$ that makes the guards $\ev{\bar{g_i}}_f = 1$,
  and $x$ or $\neg x$ is focused under $f$ for each $x\in X$.
  However, this contradicts the assumption that $val$ is the maximum valuation
  since the domain of $val$ can be extended to $\mathrm{dom}(val) \cup X$.

  Next, consider preservation of validity and uniqueness.
  First assume that the pebbling was consistent with initial coefficients $\{\gamma_{val}\}_{val\in Val}$,
  and that the valuation $val \in Val$ is refined to valuations $\{val_i\}_i$.
  Then we show that for each $val$, the coefficients $\gamma_{val_i}$ can be determined using $\gamma_{val}$.
  This $\gamma_{val_i}$ is defined in the following way:
  Let $val_i$ be a refinement of $val$,
  and $v$ be a linear vertex in $\mathrm{dom}(val_i)$ but not in $\mathrm{dom}(val)$.
  For each extension $f$ of $val_i$,
  there is a unique edge $w_f \to v$ whose guard is enabled.
  Those $\{w_f\}_f$ are connected by some enabled guarding edges.
  Hence the value of $\ev{w_f}_{val_i}$ is uniquely determined.
  Use this $w_f$ to define $\gamma_{val_i}$ by
  \[
    \gamma_{val_i} = \begin{cases}
      \gamma_{val} & \text{if } \ev{w_f}_{val_i} = val_i(v)
      \text{ for all } v \in \mathrm{dom}(val_i) - \mathrm{dom}(val) \\
      0            & \text{otherwise}
    \end{cases}.
  \]
  We will show later that the refinement $val_k$  satisfying the above condition is unique.

  Let us first check the rules that cannot change the set of valuations, making the coefficients irrelevant.
  Checking the rules \scref{pebble split}, \scref{pebble init}, \scref{pebble merge guard} is easy.
  For the rule \scref{pebble gate}, assume that we can move a pebble $p\{g_1\wedge\cdots\wedge g_n\}$
  to a gate vertex $v$ whose gate is $w\oplus(w_1\wedge\cdots\wedge w_m)$.
  Then, if an extension $f$ of $val$ enables all the guards,
  the guard is also enabled by any other extensions of $val$, and
  $\ev{v} = \neg \ev{w}$ if and only if $\ev{w_i}_{val} = 1$ is satisfied for all $i$.
  This is equivalent to applying a $C^{n+m}X$ gate to the corresponding qubits.
  The rule \scref{pebble copy/delete} can be checked similarly.

  To check the rule \scref{pebble linear guard},
  we have to consider the cases where valuations are refined.
  That is, situations when we move a pebble $p\{g_1\wedge\cdots\wedge g_n\}$ along an edge $w\to v$
  where $v$ is in $V_\text{linear}$ and whose guards are included in $g_1,\dots,g_n$.
  Since this rule does not change any quantum state, we would like to show the following equation
  where the left-hand side uses the previous pebbling,
  and the right-hand side uses the pebbling after the move.
  \[
    \sum_{val} \gamma_{val\in Val} \bigotimes_{p\in P}  \ket{\ev{v_{val,p}}_{val}}
    =
    \sum_{val\in Val}\ \sum_{val_i: \text{ refinement}}
    \gamma_{val\in Val} \bigotimes_{p\in P}  \ket{\ev{v_{val_i,p}}_{val_i}}
  \]
  It suffices to show that there is a unique refinement $val_k$ for each $val$ which satisfies
  two conditions:
  $\ev{w_f}_{val_k} = val_k(v)$ for all $v \in \mathrm{dom}(val_i) - \mathrm{dom}(val)$;
  and $\ev{v_{val,p}}_{val} = \ev{v_{val_i,p}}_{val_i}$.
  So, let us take an arbitrary $x \in \mathrm{dom}(val_i) - \mathrm{dom}(val)$.
  If $x = v$, then $val_k(v)$ has to be equal to $\ev{w}_{val}$.
  For other $x$, all the pebbles focused under extension $f$ of $val_k$ which are on $x$
  must be fragmented by $v$.
  Otherwise, for any extension of $val$, the focused pebble would be on $x$, whence $x$ would have to be in the domain of $val$.
  Therefore, the fragments of the pebble have passed through some guarded edges $\{w_j\to x\}$.
  But these $w_j$s have to be connected by some enabled guarding edges,
  since the pebble was fragmented by $v$, which is not in the domain of $val$.
  This shows that $\ev{w_j}_{val}$ is uniquely determined,
  and we were forced to choose $val_k(x)$ to be $\ev{w_j}_{val}$.
\end{proof}

\begin{lemma}\label{lem:unc/strategy-independent}
  Assume the two judgements
  \begin{align*}
    \mathscr{S} \xrightarrow\Rrightarrow \mathscr{S}^1_1 \xrightarrow{R} \mathscr{S}^2_1 \xrightarrow\Rrightarrow \mathscr{S}^3_1 \\
    \mathscr{S} \xrightarrow\Rrightarrow \mathscr{S}^1_2 \xrightarrow{R} \mathscr{S}^2_2 \xrightarrow\Rrightarrow \mathscr{S}^3_2
  \end{align*}
  where $\mathscr{S}\to \mathscr{S}^1_i$ and $\mathscr{S}^2_i\to\mathscr{S}^3_i$ are
  judgements derived from the uncomputer rules $\Rrightarrow$ via \scref{Scheduler Uncomputer},
  and $\mathscr{S}^1_i\to\mathscr{S}^2_i$ are derived from the same execution rule in
  \Cref{fig:rule/unc-checked-stmt,fig:rule/unc-expr-no-new-thread,fig:rule/unc-expr-new-thread,fig:rule/unc-statement}
  to the same thread via \scref{Scheduler Exec}.
  If the pebbling state $\mathrm{peb}$ is the same in $\mathscr{S}^3_1$ and $\mathscr{S}^3_2$,
  then the quantum state is the same in $\mathscr{S}^3_1$ and $\mathscr{S}^3_2$.
\end{lemma}
\begin{proof}
  All rules other than applying unitary gates and measurements in \Cref{fig:rule/unc-checked-stmt}
  are trivial since they do not change the quantum state but only enlarge the circuit graph.
  We only examine the case of unitary gates by the rule
  in \Cref{fig:rule/unc-checked-stmt},
  since the case of applying measurements is much easier.

  We first prove that when we can apply some unitary gate to locations $L$ by the rule
  under the control of $v_1,\dots,v_n$,
  if $val$ is a valuation, then $val[w_l \mapsto b_l]$ is also a valuation for all $b_l\in\{0,1\}$
  where $l$ ranges over $l \in L$, and where $w_l$ is the linear vertex in $G_l$.
  This is because the derivation rule of executable threads in \Cref{fig:rule/unc-executable}
  asserts that there are no pebbles fragmented by $v_l$ or $\neg v_l$.
  It follows that application of a controlled unitary gate
  can be represented as a unitary operation on the coefficients as follows.
  \[
    \gamma'_{val[w_l\mapsto b_l]}=\begin{cases}
      \sum_{b'_l\in\{0,1\}} u_{\{b_l,b'_l\}} \gamma_{val[w_l\mapsto b'_l]}
       & \text{if }\ev{v_i}_{val[w_l\mapsto b_l]} = 1 \text{ for all }i \\
      \gamma_{val}[w_l\mapsto b_l]
       & \text{otherwise}
    \end{cases}
  \]

  Let $val_1$ and $val_2$ be valuations in $\mathscr{S}^3_1$ and $\mathscr{S}^3_2$.
  Assume that $val_1$ and $val_2$ agree on $\mathrm{dom}(val_1)\cap \mathrm{dom}(val_2) - \{w_l\}$,
  and $x$ is in the domain of $val_1$ but not in the domain of $val_2$.
  This is the case when there is a pebble $P$ which is fragmented by some $w_l$ or $\neg w_l$ and on
  $x$ or $\neg x$.
  Since the rule $G,\mathrm{peb}\models \Theta\colon\text{executable}$ asserts that
  no pebble is fragmented by $x$ or $\neg x$ when transition $\mathscr{S}^1_i\to\mathscr{S}^2_i$
  occurred, the pebble $P$ was split after the application of the gate.
  Therefore, $P$ passed through some guarding edges to $x$
  during the uncomputer execution $\mathscr{S}^2_i\to\mathscr{S}^3_i$.
  Also, since the locations $L$ are linearly owned and not frozen when the unitary gate can be applied,
  any guarding edges guarded by $w_l$ or $\neg w_l$ are not enabled because of the lifetime constraints.
  This shows that the source $y_f$ of enabled guarding edges to $x$
  for each extension $f$ of $val_1$ is unique up to connected enabled guarding edges.
  Define another pebbling state $peb'$ as follows.
  \begin{itemize}
    \item Each small fragment of $P$ is moved back
          along the corresponding guarding edges $y_f \to x$.
    \item Change the pebbles fragmented by $x$ to be divided by corresponding $y_f$.
  \end{itemize}
  Then, the following is also derivable by choosing not to apply the rule \scref{pebble linear guard}
  for the fragments, and continue the execution by using $y_f$ instead of $x$.
  Note that the state of the quantum memory is the same as the original one.
  \[
    \mathscr{S} \xrightarrow\Rrightarrow \mathscr{S}^1_i \xrightarrow{R}
    \mathscr{S}^2_i \xrightarrow\Rrightarrow \mathscr{S}^3_i[peb']
  \]
  Note that, in this new pebbling, $val_1$ is no more a valuation, but we have to exclude $x$ from the domain.

  We apply the same modification to every such $x$ to obtain the pebbling state $\overline{peb}$.
  Observe that the set of valuations for $\mathscr{S}^3_i[\overline{peb}]$
  satisfies that, for all valuation $val$, the partial functions $val[w_l \mapsto b_l]$
  are also valuations for $b_l\in\{0,1\}$.
  To complete the proof, we check that the state of the quantum memory can be written with
  the coefficients $\gamma'$ as observed above.
  Again, the non-trivial case is occurrence of  refinement via \scref{pebble linear guard}.
  However, by \Cref{lem:unc/validity}, we can assume that
  any \scref{pebble linear guard} rule applied during
  $\mathscr{S}^2_i \xrightarrow\Rrightarrow \mathscr{S}^3_i[peb']$
  is done for pebbles which are not fragmented by $w_l$ or $\neg w_l$
  since these are not used as guards.
  Thus, for any refinement occurring during the execution,
  the selection of refinement which inherits the former coefficients
  does not depend on the values for $\{w_l\}$. It follows that application of the unitary to the coefficients and
  refinement by \scref{pebble linear guard} commute.
  Therefore, coefficients can still be represented by the $\gamma'$ above
  in $\mathscr{S}^3_i[peb']$.
\end{proof}

The scheduler allows the rules to be applied to two different threads in any order.
We consider the case where it is done in a different order.

\begin{lemma}\label{lem:unc/thread-safety}
  Assume that we have two rules which can be applied to different threads, and we have the following two decisions.
  The first judgement is made by applying the first rule at $\mathscr{S}^1_1 \to \mathscr{S}^2_1$
  and the second rule at $\mathscr{S}^3_1\to \mathscr{S}^4_1$.
  The second judgement is made by applying the first rule at $\mathscr{S}^3_2 \to \mathscr{S}^4_2$
  and the second rule at $\mathscr{S}^1_2 \to \mathscr{S}^2_2$.
  \begin{align*}
    \mathscr{S} \xrightarrow\Rrightarrow \mathscr{S}^1_1 \xrightarrow{R_1} \mathscr{S}^2_1 \xrightarrow\Rrightarrow \mathscr{S}^3_1\xrightarrow{R_2} \mathscr{S}^4_1 \xrightarrow\Rrightarrow \mathscr{S}^5_1 \\
    \mathscr{S} \xrightarrow\Rrightarrow \mathscr{S}^1_2 \xrightarrow{R_2} \mathscr{S}^2_2 \xrightarrow\Rrightarrow \mathscr{S}^3_2 \xrightarrow{R_1} \mathscr{S}^4_2 \xrightarrow\Rrightarrow \mathscr{S}^5_2
  \end{align*}
  All the rest of the rules applied are derived from the uncomputer rules $\Rrightarrow$ via \scref{Scheduler Uncomputer}.
  If the pebbling state $\mathrm{peb}$ is the same in $\mathscr{S}^5_1$ and $\mathscr{S}^5_2$,
  then the quantum state is also the same in those two system states.
\end{lemma}
\begin{proof}
  The rest of the parameters, the circuit graph $G$, the lifetime preorder $\mathbf{A}$,
  the state of classical memory $s$, threads $\bar\Theta$ are the same
  in $\mathscr{S}^5_1$ and $\mathscr{S}^5_2$.
  As in the previous lemma, we only need to check the application of unitary gates and measurements.

  In the rules that spawn a new thread, defined in \Cref{fig:rule/unc-expr-new-thread}, the parent thread is changed to the state $\mathbf{Await}(y,\bar{t_i})$ when spawning child threads.
  The rule to wake up the parent thread defined in \Cref{fig:rule/unc-merge}
  requires the child threads to terminate.
  Since the rule of executing $\Qif$ defined in \Cref{fig:rule/unc-expr-new-thread}
  is the only rule that can spawn two threads that can be executed in different order,
  and the rule \textbf{Purely Quantum} in \Cref{fig:rule/PQ} prohibits measurements
  to happen under any quantum control,
  we can execute the measurement rule in \Cref{fig:rule/unc-checked-stmt}
  when the thread is the only one which is not in the state $\mathbf{Await}(y,t)$.
  That means that the rule applying measurement cannot be the rule used in $R_1$ or $R_2$.

  Therefore, we can assume that both rules used in $R_1$ and $R_2$ apply unitary gates, as defined in \Cref{fig:rule/unc-checked-stmt}.
  As in \Cref{lem:unc/strategy-independent},
  we need only consider the application of unitary operators to the coefficients.
  We show that the two applications of controlled unitaries commute
  by proving that their controls are orthogonal.
  First, let us note that the set of controls $\mathrm{Control}(t)$ can only be updated
  by the rule executing $\Qif$ in \Cref{fig:rule/unc-expr-new-thread}.
  In this rule, one of the child threads spawned has control $v_x$
  and the other has control $\neg v_x$.
  This shows that for each valuation $val$, at most one child thread $\Theta_i$
  satisfies $\ev{v_j}_{val} = 1$ for all $v_j\in \mathrm{Control}(t_i)$.
  Therefore, for each valuation,
  at most one application of the unitary changes the corresponding coefficients,
  which shows that the order of the application of unitaries does not matter.
\end{proof}

\begin{proof}[Proof for \Cref{thm:uncomp-independent}]
  \Cref{lem:unc/validity,lem:unc/strategy-independent,lem:unc/thread-safety}
  show that if the pebbling state is the same after the execution of a closed program
  from the initial state, then the quantum state is also the same,
  which establishes thread safety and independence of the pebbling strategy.
\end{proof}

Finally, we prove the equivalence of the semantics (\Cref{thm:semantics-equiv}).
For the proof, we fix a pebbling strategy that does both computation and uncomputation eagerly right after
the execution of statements.

\begin{definition}
  A system state $\mathscr{S}$ is called \emph{eagerly computed}
  if the circuit graph is pebbled under the following conditions, and the quantum state is valid.
  \begin{itemize}
    \item For each thread $\Theta_t$ and each $v\in \mathrm{now}(t)$,
          the pebble $p\{\mathrm{Control}(t)\}$ is on $v$
          where $p$ is the main label whose location is $\loc(v)$.
    \item For main labels $p$ whose fragments do not sum to a whole pebble
          when the pebbles satisfy the above condition,
          the remaining fragments are assumed to be on the initial vertices.
    \item No pebble with auxiliary label is on any vertex.
  \end{itemize}
\end{definition}

The idea of the definition of eagerly computed system state is simple. Since each thread assumes that the vertices in set $\text{now}(t)$ are pebbled, we simply do what the thread says.

\begin{lemma}\label{lem:unc/eager-strategy}
  Let $\mathscr{S}$ be an eagerly computed system state
  which obtained from the initial state by executing a closed program
  with at least one non-terminated thread.
  Then, if there is a thread whose statement $S_t$ is not of the form $\Drop x;S'$,
  there is another eagerly computed system state $\mathscr{S}'$ satisfying
  \[
    \mathscr{S} \xrightarrow\Rrightarrow \mathscr{S}_1 \xrightarrow{R}
    \mathscr{S}_2 \xrightarrow\Rrightarrow \mathscr{S}'
  \]
  where the second rule $R$ applied is \scref{Scheduler Exec} or \scref{Scheduler Await},
  and the rest are uncomputer rules.
\end{lemma}
\begin{proof}
  By induction, we can easily check that if a thread $\Theta_t$ owns a location $l$ linearly,
  then the thread includes the vertex $v\in G_l \cap V_\text{linear}$ in $\mathrm{now}(t)$
  because of the rules for $\Endlft \alpha$ and $x \As T$ in \Cref{fig:rule/unc-expr-no-new-thread}.
  Also, if the set $\mathrm{now}(t)$ includes a vertex which the thread does not own,
  then the set $\mathrm{now}(t)$ will include that vertex until it terminates.

  We show that all threads in a state $\mathbf{Check}(S)$ are executable
  in an eagerly computed system state by slightly modifying the pebbling state in
  $\mathscr{S} \xrightarrow\Rrightarrow \mathscr{S}_1$.
  Then, if there are no threads satisfying $G, \mathrm{peb}\models\Theta\colon\text{executable}$,
  there must be some threads waiting for termination of its child threads,
  which can be woken up by the rule in \Cref{fig:rule/unc-merge}.
  Let $\Theta_t$ be a thread in the state $\mathbf{Check}(S)$.
  The rules are defined in \Cref{fig:rule/unc-checked-stmt}.
  \begin{itemize}
    \item If $S$ is $\Endlft\alpha$ or $S = x\As T$, then
          the last execution rule applied to $\Theta_t$ is the rule in \Cref{fig:rule/unc-statement}
          which adds the set $\{w_\text{lin}\}$ to $\mathrm{now}(t)$.
          Therefore, for each location $l$ affinely owned before the execution of $S$ but not after,
          the vertex $v \in G_l\cap V_\text{linear}$ is pebbled in $\mathscr{S}$ since it is eagerly executed.
    \item If $S = \Let y = \Meas(x)$, then the thread does not have any control
          by the rule \textbf{Purely Quantum} in \Cref{fig:rule/PQ}.
          Also, for each location owned by $x$, which is linearly owned,
          the linear vertex is included in $\mathrm{now}(t)$.
          Therefore, no fragmented pebbles are on  linear vertices.
    \item The case of $S = \Let y = U(x)$ is the only one where we might have to modify the pebbling state
          to satisfy the condition $\mathbf{A}, G, \mathbf{peb}, \{p_i\} \models\mathrm{Control}(t)$.
          Since the controls are the frozen locations used as condition qubits of $\Qif$
          and owned by parent threads,
          the controls can be arranged in the order $v_1,\dots,v_n$ in which they were used as controls
          in the nested $\Qif$ statements in the program.
          Here, we observe that the vertex $v_i$ or $\neg v_i$ has to be pebbled at least by
          $p\{v_1\wedge\cdots\wedge v_{i-1}\}$ in the eagerly computed system state $\mathscr{S}$.
          This is because every time $\Qif$ is executed, it spawns two threads with control
          $v_i$ and $\neg v_i$ and the same set $now(t_i)$ that includes $v_i$ or $\neg v_i$.
          But since both child threads do not own $v_i$ or $\neg v_i$,
          eager computation asserts that the condition vertex is still pebbled.
          Therefore, in order to make the pebbling state satisfy
          $\mathbf{A}, G, \mathbf{peb}, \{p_i\} \models\mathrm{Control}(t)$,
          we can move the pebble $p\{v_1\wedge\cdots\wedge v_{i-1}\}$ to the vertex $v_i$
          by applying the rule \scref{pebble gate} in \Cref{fig:rule/unc-uncomputer}
          in the reverse order $v_n,\dots,v_1$.
          Note that this modification of controls is reversible,
          and we can make the pebbling state eagerly computed again.

          The rest of the conditions are easy to check.
          If $\{w_j\}$ are the targets of the unitary gate to apply,
          then since $w_j$s are linearly owned by the thread, and the threades do not have any children,
          eager computation asserts that the pebbles are fragmented by $w_j$ or $\neg w_j$.
          It also asserts that $\{w_j\}$ is pebbled by $q_j\{\mathrm{Control}(t)\}$.
  \end{itemize}
  Therefore, a thread in state $\mathbf{Check}(S)$ is executable.

  Finally, we prove that if the rule applied in the second step is not $\Drop$,
  then the uncomputer can make the system state eagerly executed again in
  $\mathscr{S}_2\xrightarrow\Rrightarrow \mathscr{S}'$.
  To prove this, it suffices to check the rules which can change the set $\mathrm{now}(t)$.
  Those are the execution rules:
  \begin{itemize}
    \item Waking up the parent thread in \Cref{fig:rule/unc-merge},
    \item $\Let y = \ket0,x_0\oplus(\neg^{i_1}x_1\wedge\cdots\wedge \neg^{i_n}x_n)$
          in \Cref{fig:rule/unc-expr-no-new-thread}.
    \item $\Endlft\alpha$ and $x\As T$ in \Cref{fig:rule/unc-statement},
  \end{itemize}
  The first and third case are rules to add guard edges to the circuit graph.
  By applying the rule \scref{pebble merge guard} and \scref{pebble linear guard} respectively,
  we can make the system state eagerly computed again.
  For the second case, we only need to apply the rule \scref{pebble init} or
  \scref{pebble gate} from left to right or the other way around.
\end{proof}

\begin{lemma}\label{lem:unc/reachable-affine-vertex}
  Let $\mathscr{S}$ be an eagerly computed system state
  and $v$ be a vertex in the circuit graph $G$ such that $v,\neg v\not\in V_\text{linear}$.
  In this lemma, we consider a strategy of putting a pebble $p\{\bigwedge_i g_i\}$
  with auxiliary label on the vertex $v$,
  where $g_i$ satisfies $\mathbf{A},G,\mathrm{peb},\{p_i\}\models \bigwedge_i g_i$ with
  some main labels $\{p_i\}$.

  Assume that for each merge vertex $w$ from which $v$ is reachable,
  the guard $g'$ of $w$ satisfies
  \[
    \textstyle
    \mathbf{A},G,\mathrm{peb},\{p_i\}+\{q\}\models \left(\bigwedge_i g_i\right)\wedge g'
    \qquad\text{or}\qquad
    \mathbf{A},G,\mathrm{peb},\{p_i\}+\{q\}\models \left(\bigwedge_i g_i\right)\wedge \neg g'.
  \]
  with the corresponding main label $q$.
  Then, there is a system state $\mathscr{S}'$ such that
  the vertex $v$ is pebbled by a pebble fragmented by $\bigwedge_{i} g_{i}$,
  and the pebbling of $\mathscr{S}'$ and $\mathscr{S}$ differs only on the pebbles with auxiliary labels,
  which can be obtained by the uncomputer rules.
  \[
    \mathscr{S} \xrightarrow\Rrightarrow \mathscr{S}'
  \]
\end{lemma}
\begin{proof}
  By induction on the length of the longest path from an initial vertex to $v$.

  If the vertex $v$ is in $V_\text{init}$, then we can apply the rule \scref{pebble init}.

  If the vertex $v$ is in $V_\text{gate}$, then by induction hypothesis,
  for each predecessor $w_i$ of $v$,
  there is a system state $\mathscr{S}_i$ such that $w_i$ is pebbled.
  By performing those uncomputer executions in sequence, we get a system state $\mathscr{S}'_i$
  where $w_0\dots w_i$ are pebbled.
  \[
    \mathscr{S} \xrightarrow\Rrightarrow \mathscr{S}'_0 = \mathscr{S}_0
    \xrightarrow\Rrightarrow \mathscr{S}'_1\xrightarrow\Rrightarrow
    \cdots \xrightarrow\Rrightarrow \mathscr{S}'_n
  \]
  Since in $\mathscr{S}'_n$, all the predecessors of $v$ are pebbled,
  we can then apply the rule \scref{pebble gate} to get the desired system state $\mathscr{S}'$.

  If the vertex $v$ is in $V_\text{merge}$, then by the induction hypothesis,
  there is a system state $\mathscr{S}_j$ such that the predecessor $v_j$ of $v$ is pebbled for each $j=0,1$.
  As in the case of $V_\text{gate}$, we can obtain a system state $\mathscr{S}^*$
  where $v_0$ and $v_1$ are pebbled by $p_0\{\bigwedge_i g_i\}$ and $p_1\{\bigwedge_i g_i\}$, respectively.
  Now, if the guard $g'$ of $v$ or its negation $\neg g'$ is included in $\{g_i\}$,
  then one of the pebbles labeled by $p_0$ or $p_1$ can be moved to $v$ by the rule
  \scref{pebble merge guard}.
  If not, then we introduce another auxiliary label $q$ and put a whole pebble on the initial vertex
  $w$ with the same location.
  
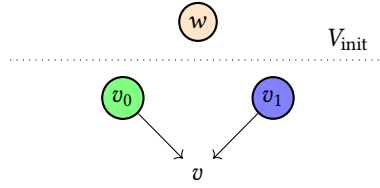
\[
    \begin{tikzpicture}
      \node[vertex,fill=light-orange] (w) at (0,2) {$w$};
      \node[vertex,fill=light-green] (v0) at (-1,1) {$v_0$};
      \node[vertex,fill=light-blue] (v1) at (1,1) {$v_1$};
      \node (v) at (0,0) {$v$};
      \draw[dotted] (-2.5,1.5) to node[pos=0.9,auto]{$V_\text{init}$} (2.5,1.5);
      \draw[->] (v0) -- (v);
      \draw[->] (v1) -- (v);
    \end{tikzpicture}
  \]
  By splitting the pebble by $g'$ and applying the rule \scref{pebble merge guard},
  we can move two fragments of the pebble onto $v$
  labeled by $p_0$ and $p_1$.
  \[
    \begin{tikzpicture}
      \node[vertex,fill=light-orange] (w) at (0,2) {$w$};
      \begin{scope}
        \clip (-1,1) circle (0.25);
        \fill[light-green] (-2,0) rectangle (-1,2);
        \fill[white] (0,0) rectangle (-1,2);
        \draw[ultra thick] (-1,1) circle (0.25) node {$v_0$};
      \end{scope}
      \begin{scope}
        \clip (1,1) circle (0.25);
        \fill[light-blue] (2,0) rectangle (1,2);
        \fill[white] (0,0) rectangle (1,2);
        \draw[ultra thick] (1,1) circle (0.25) node {$v_1$};
      \end{scope}
      \begin{scope}
        \clip (0,0) circle (0.25);
        \fill[light-blue] (-1,-1) rectangle (0,2);
        \fill[light-green] (1,-1) rectangle (0,2);
        \draw[ultra thick] (0,0) circle (0.25) node {$v$};
      \end{scope}
      \draw[dotted,line width=.5pt] (-2.5,1.5) to node[pos=0.9,auto]{$V_\text{init}$} (2.5,1.5);
      \draw[->] (v0) -- (v);
      \draw[->] (v1) -- (v);
    \end{tikzpicture}
  \]
  It seems like the vertex $v$ has been pebbled correctly,
  but there is still a problem that the labels of pebbles on $v$ do not match.
  To fix this, observe that we can swap the labels of pebbles on any vertex and the initial vertex
  by applying the rule \scref{pebble copy/delete} twice.
  By using this property, we can first swap the pebble on $v_0$ labeled by $p_0$ to the initial vertex $w$,
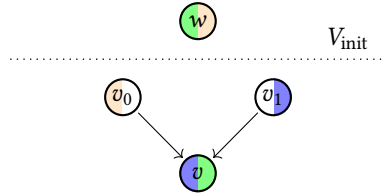
  \[
    \begin{tikzpicture}
      \begin{scope}
        \clip (0,2) circle (0.25);
        \fill[light-orange] (1,1) rectangle (0,3);
        \fill[light-green] (-1,1) rectangle (0,3);
        \draw[ultra thick] (0,2) circle (0.25) node {$w$};
      \end{scope}
      \begin{scope}
        \clip (-1,1) circle (0.25);
        \fill[light-orange] (-2,0) rectangle (-1,2);
        \fill[white] (0,0) rectangle (-1,2);
        \draw[ultra thick] (-1,1) circle (0.25) node {$v_0$};
      \end{scope}
      \begin{scope}
        \clip (1,1) circle (0.25);
        \fill[light-blue] (2,0) rectangle (1,2);
        \fill[white] (0,0) rectangle (1,2);
        \draw[ultra thick] (1,1) circle (0.25) node {$v_1$};
      \end{scope}
      \begin{scope}
        \clip (0,0) circle (0.25);
        \fill[light-blue] (-1,-1) rectangle (0,2);
        \fill[light-green] (1,-1) rectangle (0,2);
        \draw[ultra thick] (0,0) circle (0.25) node  {$v$};
      \end{scope}
      \draw[dotted,line width=.5pt] (-2.5,1.5) to node[pos=0.9,auto]{$V_\text{init}$} (2.5,1.5);
      \draw[->] (v0) -- (v);
      \draw[->] (v1) -- (v);
    \end{tikzpicture}
  \]
  and then swap the pebble on $v$ labeled by $p_1$ to the initial vertex $w$.
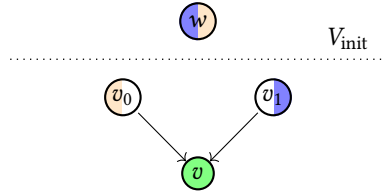
  \[
    \begin{tikzpicture}
      \begin{scope}
        \clip (0,2) circle (0.25);
        \fill[light-orange] (1,1) rectangle (0,3);
        \fill[light-blue] (-1,1) rectangle (0,3);
        \draw[ultra thick] (0,2) circle (0.25) node {$w$};
      \end{scope}
      \begin{scope}
        \clip (-1,1) circle (0.25);
        \fill[light-orange] (-2,0) rectangle (-1,2);
        \fill[white] (0,0) rectangle (-1,2);
        \draw[ultra thick] (-1,1) circle (0.25) node {$v_0$};
      \end{scope}
      \begin{scope}
        \clip (1,1) circle (0.25);
        \fill[light-blue] (2,0) rectangle (1,2);
        \fill[white] (0,0) rectangle (1,2);
        \draw[ultra thick] (1,1) circle (0.25) node {$v_1$};
      \end{scope}
      \node[pebbled] (v) at (0,0) {$v$};
      \draw[dotted,line width=.5pt] (-2.5,1.5) to node[pos=0.9,auto]{$V_\text{init}$} (2.5,1.5);
      \draw[->] (v0) -- (v);
      \draw[->] (v1) -- (v);
    \end{tikzpicture}
  \]
  Now we have the pebble $p_0\{\bigwedge g_i\}$ on $v$, which completes the proof.
\end{proof}

\begin{corollary}\label{cor:unc/eager-drop}
  \Cref{lem:unc/eager-strategy} can also be extended to the case of $\Drop$.
\end{corollary}
\begin{proof}
  Let $\mathscr{S}$ be an eagerly computed system state obtained from the initial state
  by executing a closed program.
  Assume that there is a thread $\Theta_t$ whose state is $\mathbf{Ready}$
  and whose statement $S_t$ is $\Drop x;S'$.
  We show that, for each $v \in G_l \cap \mathrm{now}(t)$ where the location $l$ is owned by $x$,
  we can remove the pebble $p\{\mathrm{Control}(t)\}$ on $v$.
  The strategy of removing the pebble is to first put a pebble $q\{\mathrm{Control}(t)\}$
  with an auxiliary label on $v$ using the previous \Cref{lem:unc/reachable-affine-vertex},
  then, using the rule \scref{pebble copy/delete} to remove the pebble $p\{\mathrm{Control}(t)\}$,
  and finally, by reversing the uncomputer execution of the first step, to remove all
  the unnecessary pebbles with auxiliary labels.
  It thus suffices to show that the vertex $v$ satisfies the condition of
  \Cref{lem:unc/reachable-affine-vertex} with $\bigwedge_i g_i \coloneqq \mathrm{Control}(t)$.

  As in the proof of \Cref{lem:unc/eager-strategy} for the case of $\Let y = U(x)$,
  we may assume $\mathbf{A}, G, \mathrm{peb}, \{p_i\} \models \mathrm{Control}(t)$ is satisfied
  since, if necessary, we can apply the rule \scref{pebble gate} along edges $v \to \neg v$ repeatedly to make it so.

  If there is a path from a merge vertex to $v$, and the guard is $g' = (w,\alpha)$,
  then the location $\loc(w)$ was frozen by some parent thread of $\Theta_t$ before.
  We claim that $\alpha$ must still be in $\mathbf{A}$.
  This is because any type of variable which used the location
  must be bound by $\Own^\beta$ with some $\beta\leq\alpha$,
  but since $v$ is still affinely owned at $\mathbf{A}$,
  both $\alpha$ and $\beta$ are in $\mathbf{A}$.
  Therefore, the location $\loc(w)$ must still be frozen,
  which means that some parent thread $\Theta_{t'}$ must have $w$ in its set $\mathrm{now}(t')$.
  This proves that $w$ is pebbled by some $p'\{\mathrm{Control}(t')\}$,
  and since $\mathrm{Control}(t')$ is a subset of $\mathrm{Control}(t)$,
  the condition $\mathbf{A},G,\mathrm{peb},\{p_i\}\cup\{p'\}\models \mathrm{Control}(t)\wedge g'$ is satisfied.
\end{proof}

The strategy of uncomputation seems to be different from Bennett's construction~\cite{Bennett}.
In fact, the strategy in the previous proof takes too many steps to remove the pebble because it first puts a pebble with auxiliary label and then removes it.
Bennett's construction foregoes that first step altogether.
Briefly, if we apply the Bennett's construction to our pebble game,
we do not start from the eagerly computed state, but from the state where some pebbles which will be used during the uncomputation step are copied appropriately during the computation step in advance.
For example, in the following program, the eagerly computed state at line~4 is presented on the right below.
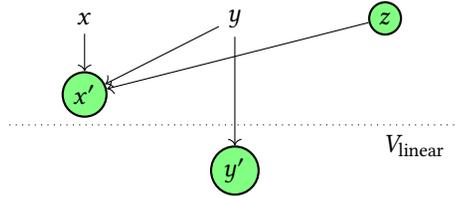
\begin{center}
  \begin{minipage}{0.3\textwidth}
    \begin{lstlisting}[language=qurts]
        // x, y, z: #'a qbit
        let x' = x$\oplus$(y$\wedge$z);
        y as #'0 qbit;
        let y' = H(y);
        drop x';
      \end{lstlisting}
  \end{minipage}
  \begin{minipage}{0.5\textwidth}
    \[
      \begin{tikzpicture}
        \node (x) at (0,1) {$x$};
        \node (y) at (2,1) {$y$};
        \node[pebbled] (z) at (4,1) {$z$};
        \node[pebbled] (x') at (0,0) {$x'$};
        \node[pebbled] (y') at (2,-1) {$y'$};
        \draw[dotted] (-1,-0.4) to node[pos=0.9,auto,swap]{$V_\text{linear}$} (5,-0.4);
        \draw[->] (x) -- (x');
        \draw[->] (y) -- (x');
        \draw[->] (z) -- (x');
        \draw[->] (y) -- (y');
      \end{tikzpicture}
    \]
  \end{minipage}
\end{center}
To uncompute the pebble on $x'$, we first need to put a pebble on $y$ in some way, which is the strategy in the proof of \Cref{lem:unc/reachable-affine-vertex}.
On the other hand, Bennett's construction puts a pebble on $y$ in advance,
which enables moving the pebble on $x'$ back to $x$ immediately.
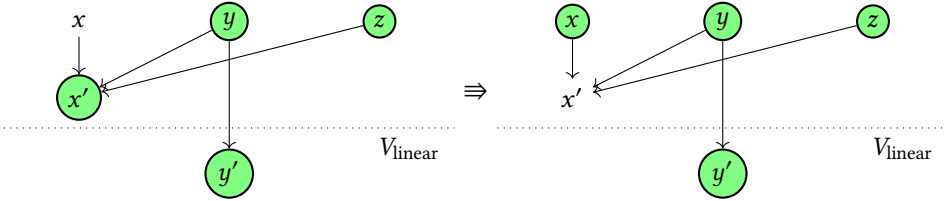
\[
  \begin{tikzpicture}[baseline=(x')]
    \node (x) at (0,1) {$x$};
    \node[pebbled] (y) at (2,1) {$y$};
    \node[pebbled] (z) at (4,1) {$z$};
    \node[pebbled] (x') at (0,0) {$x'$};
    \node[pebbled] (y') at (2,-1) {$y'$};
    \draw[dotted] (-1,-0.4) to node[pos=0.9,auto,swap]{$V_\text{linear}$} (5,-0.4);
    \draw[->] (x) -- (x');
    \draw[->] (y) -- (x');
    \draw[->] (z) -- (x');
    \draw[->] (y) -- (y');
  \end{tikzpicture}
  \Rrightarrow
  \begin{tikzpicture}[baseline=(x')]
    \node[pebbled] (x) at (0,1) {$x$};
    \node[pebbled] (y) at (2,1) {$y$};
    \node[pebbled] (z) at (4,1) {$z$};
    \node (x') at (0,0) {$x'$};
    \node[pebbled] (y') at (2,-1) {$y'$};
    \draw[dotted] (-1,-0.4) to node[pos=0.9,auto,swap]{$V_\text{linear}$} (5,-0.4);
    \draw[->] (x) -- (x');
    \draw[->] (y) -- (x');
    \draw[->] (z) -- (x');
    \draw[->] (y) -- (y');
  \end{tikzpicture}
\]
The previous lemma only aimed to show that we could make eagerly computed states in some way, so followed a simpler and less efficient strategy.
But we believe that it should also be possible to formalise Bennett's construction in our pebble game.

We now prove the equivalence of the two semantics.

\begin{proof}[Proof for \Cref{thm:semantics-equiv}]
  We prove that a quantum state of an eagerly computed system state
  is the same as the quantum state in the simulation semantics by induction on the construction of the statements.
  This proves the theorem since we have already shown
  in the \Cref{lem:unc/eager-strategy} and \Cref{cor:unc/eager-drop}
  that there is a pebbling strategy which keeps the pebbling state eagerly computed.

  The non-trivial cases are $\Qif$ and $\Drop$.

  In the simulation semantics, the application of a $\Qif$ statement is defined by the superposition
  of the two states defined by the execution of the two branches.
  Therefore, if the execution of the block $B_{\ket i}$
  is equivalent to the application of a sequence of unitary gates $\{U^{\ket i}_j\}_{j=0}^{n_i}$,
  then the execution of the statement $\Qif\ x\ B_{\ket 1}\ \Else\ B_{\ket 0}$ is equivalent
  to the application of the sequences of controlled operators
  $\{C_{\ket 0}U^{\ket 0}_j\}_{j=0}^{n_0}$
  and
  $\{C_{\ket 1}U^{\ket 1}_j\}_{j=0}^{n_1}$
  where $C_{\ket 1}U^{\ket 1}_j$ is the normal controlled operator of $U^{\ket 1}_j$,
  and $C_{\ket 0}U^{\ket 0}_j$ is the \emph{negatively controlled} operator of $U^{\ket 0}_j$,
  that is, $(X\otimes \id) C_{\ket 1}U^{\ket 0}_j (X\otimes \id)$.
  Therefore, it is equivalent to add condition qubits to the set $\mathrm{Control(t)}$ of child threads,
  and continue executing the child threads, but substituting the operators
  with the ones controlled by the conditions.

  The semantics of $\Drop x$ in the simulation semantics is defined by
  \[
    \sum_{i\in {\{0,1\}}^{|L_1|}} \ket{\phi_i}\ket{i}
    \rightarrow
    \sum_{i\in \{0,1\}^{|L_q|}} \ket{\phi_i},
  \]
  where $\ket{i}$ ranges over the computational basis $\{\ket{i}\}$ of the Hilbert space $\HH_{L_q}$
  of qubits owned by $x$.
  This is equivalent to removing pebbles from the vertices whose location is in $L_q$
  by the definition of a valid pair $(\mathrm{peb},\ket\phi)$ of a pebbling and a quantum state
  in \Cref{def:unc/valid-pair}.
\end{proof}

\section{Comparison with Silq}\label{ax-sec:silq}

In Silq~\cite{Silq_2020}, a variable with type annotation \silqinline{const} can be treated like
an immutable pointer and can be copied implicitly.
However, the copy operation is not interpreted as a copy of the pointer, but as a copy of the value.
For example, the following function \silqinline{my_identity}, which simply returns its input,
is interpreted as $\ket{i} \mapsto \ket{ii}$.
This operation implicitly involves a CX gate.
\begin{center}
  \begin{minipage}{0.35\textwidth}
    \begin{lstlisting}[language=silq]
def my_identity(const x: B) : B {
  return x;
}
def main() {
  x := 0:B;
  return (my_identity(x),x);
}
    \end{lstlisting}
  \end{minipage}
\end{center}
There are several reasons why this may not be the ideal way to interpret this function:
\begin{itemize}
  \item it is not consistent with a linear type system;
  \item the implicit insertion of CX, a relatively costly gate, is undesirable; 
  \item it is not clear to the programmer that the function is interpreted as CX gate.
\end{itemize}
Qurts can copy an immutable reference to a qubit, but not the qubit itself.
In other words: Qurts uses a move strategy like Rust, not a copy strategy.

Silq can also throw runtime errors caused by illegal uncomputation, as in the following code.
\begin{center}
  \begin{minipage}{0.8\textwidth}
    \begin{lstlisting}[language=silq]
def f(x:B) {
    tmp := 0:B;
    (x, tmp) := (tmp, x);
    return x;
    // runtime error: bad forget
}
def main() {
    x := H(0:B); // if this x is 0:B, then this is not a runtime error
    return f(x);
}
    \end{lstlisting}
  \end{minipage}
\end{center}
This could be a bug in the compiler, caused by assignment of a qubit to another qubit in line 3.
Since assignment is not discussed in the article~\cite{Silq_2020}, it is unclear whether this should be detected by the type checker or not.

The Qurts type system detects this error as follows.
\begin{center}
  \begin{minipage}{0.8\textwidth}
    \begin{lstlisting}[language=qurts]
      fn f<'a!='0>(x: #'a qbit) -> #'a qbit {
        let tmp: #'a qbit = $\ket0$;
        let (x, tmp) = (tmp, x); // both x and tmp has the same type
        return x;
      }
      fn main() {
        let x: #'0 qbit = $\ket0$.H();
        return f(x); // type error: x does not have type #'a qbit where 'a!='0.
      }
    \end{lstlisting}
  \end{minipage}
\end{center}
Note that the program does pass the type check if the input qubit is $\ket0$:
\begin{center}
  \begin{minipage}{0.8\textwidth}
    \begin{lstlisting}[language=qurts]
      fn f<'a!='0>(x: #'a qbit) -> #'a qbit {
        let tmp: #'a qbit = $\ket0$;
        let (x, tmp) = (tmp, x);
        return x;
      }
      fn main() {
        let x: #'static qbit = $\ket0$;
        return f(x);
      }
    \end{lstlisting}
  \end{minipage}
\end{center}



\begin{thebibliography}{39}


\ifx \showCODEN    \undefined \def \showCODEN     #1{\unskip}     \fi
\ifx \showDOI      \undefined \def \showDOI       #1{#1}\fi
\ifx \showISBNx    \undefined \def \showISBNx     #1{\unskip}     \fi
\ifx \showISBNxiii \undefined \def \showISBNxiii  #1{\unskip}     \fi
\ifx \showISSN     \undefined \def \showISSN      #1{\unskip}     \fi
\ifx \showLCCN     \undefined \def \showLCCN      #1{\unskip}     \fi
\ifx \shownote     \undefined \def \shownote      #1{#1}          \fi
\ifx \showarticletitle \undefined \def \showarticletitle #1{#1}   \fi
\ifx \showURL      \undefined \def \showURL       {\relax}        \fi
\providecommand\bibfield[2]{#2}
\providecommand\bibinfo[2]{#2}
\providecommand\natexlab[1]{#1}
\providecommand\showeprint[2][]{arXiv:#2}

\bibitem[Altenkirch and Grattage(2005)]%
        {Altenkirch_G05_QML}
\bibfield{author}{\bibinfo{person}{Thorsten Altenkirch} {and} \bibinfo{person}{Jonathan Grattage}.} \bibinfo{year}{2005}\natexlab{}.
\newblock \showarticletitle{A Functional Quantum Programming Language}. In \bibinfo{booktitle}{\emph{20th {IEEE} Symposium on Logic in Computer Science ({LICS} 2005), 26-29 June 2005, Chicago, IL, USA, Proceedings}}. \bibinfo{publisher}{{IEEE} Computer Society}, \bibinfo{pages}{249--258}.
\newblock
\urldef\tempurl%
\url{https://doi.org/10.1109/LICS.2005.1}
\showDOI{\tempurl}


\bibitem[Amy et~al\mbox{.}(2017)]%
        {AmyRoettelerSvore_CAV2017_ReVerC}
\bibfield{author}{\bibinfo{person}{Matthew Amy}, \bibinfo{person}{Martin Roetteler}, {and} \bibinfo{person}{Krysta~M. Svore}.} \bibinfo{year}{2017}\natexlab{}.
\newblock \showarticletitle{Verified Compilation of Space-Efficient Reversible Circuits}. In \bibinfo{booktitle}{\emph{Computer Aided Verification - 29th International Conference, {CAV} 2017, Heidelberg, Germany, July 24-28, 2017, Proceedings, Part {II}}} \emph{(\bibinfo{series}{Lecture Notes in Computer Science}, Vol.~\bibinfo{volume}{10427})}, \bibfield{editor}{\bibinfo{person}{Rupak Majumdar} {and} \bibinfo{person}{Viktor Kuncak}} (Eds.). \bibinfo{publisher}{Springer}, \bibinfo{pages}{3--21}.
\newblock
\urldef\tempurl%
\url{https://doi.org/10.1007/978-3-319-63390-9\_1}
\showDOI{\tempurl}


\bibitem[Bennett(1973)]%
        {Bennett}
\bibfield{author}{\bibinfo{person}{Charles~H. Bennett}.} \bibinfo{year}{1973}\natexlab{}.
\newblock \showarticletitle{Logical Reversibility of Computation}.
\newblock \bibinfo{journal}{\emph{IBM Journal of Research and Development}} \bibinfo{volume}{17}, \bibinfo{number}{6} (\bibinfo{year}{1973}), \bibinfo{pages}{525--532}.
\newblock
\urldef\tempurl%
\url{https://doi.org/10.1147/rd.176.0525}
\showDOI{\tempurl}


\bibitem[Bennett(1989)]%
        {Bennett-89-rev-pebble-tradeoff}
\bibfield{author}{\bibinfo{person}{Charles~H. Bennett}.} \bibinfo{year}{1989}\natexlab{}.
\newblock \showarticletitle{Time/Space Trade-Offs for Reversible Computation}.
\newblock \bibinfo{journal}{\emph{SIAM J. Comput.}} \bibinfo{volume}{18}, \bibinfo{number}{4} (\bibinfo{year}{1989}), \bibinfo{pages}{766--776}.
\newblock
\urldef\tempurl%
\url{https://doi.org/10.1137/0218053}
\showDOI{\tempurl}


\bibitem[Bhattacharjee et~al\mbox{.}(2019)]%
        {BhattacharjeeSD19-ISMVL-rev-pebble}
\bibfield{author}{\bibinfo{person}{Debjyoti Bhattacharjee}, \bibinfo{person}{Mathias Soeken}, \bibinfo{person}{Srijit Dutta}, \bibinfo{person}{Anupam Chattopadhyay}, {and} \bibinfo{person}{Giovanni~De Micheli}.} \bibinfo{year}{2019}\natexlab{}.
\newblock \showarticletitle{Reversible Pebble Games for Reducing Qubits in Hierarchical Quantum Circuit Synthesis}. In \bibinfo{booktitle}{\emph{2019 {IEEE} 49th International Symposium on Multiple-Valued Logic (ISMVL), Fredericton, NB, Canada, May 21-23, 2019}}. \bibinfo{publisher}{{IEEE}}, \bibinfo{pages}{102--107}.
\newblock
\urldef\tempurl%
\url{https://doi.org/10.1109/ISMVL.2019.00026}
\showDOI{\tempurl}


\bibitem[Bichsel et~al\mbox{.}(2020)]%
        {Silq_2020}
\bibfield{author}{\bibinfo{person}{Benjamin Bichsel}, \bibinfo{person}{Maximilian Baader}, \bibinfo{person}{Timon Gehr}, {and} \bibinfo{person}{Martin~T. Vechev}.} \bibinfo{year}{2020}\natexlab{}.
\newblock \showarticletitle{Silq: a high-level quantum language with safe uncomputation and intuitive semantics}. In \bibinfo{booktitle}{\emph{Proceedings of the 41st {ACM} {SIGPLAN} International Conference on Programming Language Design and Implementation, {PLDI} 2020, London, UK, June 15-20, 2020}}, \bibfield{editor}{\bibinfo{person}{Alastair~F. Donaldson} {and} \bibinfo{person}{Emina Torlak}} (Eds.). \bibinfo{publisher}{{ACM}}, \bibinfo{pages}{286--300}.
\newblock
\urldef\tempurl%
\url{https://doi.org/10.1145/3385412.3386007}
\showDOI{\tempurl}


\bibitem[Chan et~al\mbox{.}(2015)]%
        {ChanLNV_15_RevPebble_PSPACE}
\bibfield{author}{\bibinfo{person}{Siu~Man Chan}, \bibinfo{person}{Massimo Lauria}, \bibinfo{person}{Jakob Nordstr{\"{o}}m}, {and} \bibinfo{person}{Marc Vinyals}.} \bibinfo{year}{2015}\natexlab{}.
\newblock \showarticletitle{Hardness of Approximation in {PSPACE} and Separation Results for Pebble Games}. In \bibinfo{booktitle}{\emph{{IEEE} 56th Annual Symposium on Foundations of Computer Science, {FOCS} 2015, Berkeley, CA, USA, 17-20 October, 2015}}, \bibfield{editor}{\bibinfo{person}{Venkatesan Guruswami}} (Ed.). \bibinfo{publisher}{{IEEE} Computer Society}, \bibinfo{pages}{466--485}.
\newblock
\urldef\tempurl%
\url{https://doi.org/10.1109/FOCS.2015.36}
\showDOI{\tempurl}


\bibitem[Ding et~al\mbox{.}(2020)]%
        {SQUARE_2020}
\bibfield{author}{\bibinfo{person}{Yongshan Ding}, \bibinfo{person}{Xin{-}Chuan Wu}, \bibinfo{person}{Adam Holmes}, \bibinfo{person}{Ash Wiseth}, \bibinfo{person}{Diana Franklin}, \bibinfo{person}{Margaret Martonosi}, {and} \bibinfo{person}{Frederic~T. Chong}.} \bibinfo{year}{2020}\natexlab{}.
\newblock \showarticletitle{{SQUARE:} Strategic Quantum Ancilla Reuse for Modular Quantum Programs via Cost-Effective Uncomputation}. In \bibinfo{booktitle}{\emph{47th {ACM/IEEE} Annual International Symposium on Computer Architecture, {ISCA} 2020, Virtual Event / Valencia, Spain, May 30 - June 3, 2020}}. \bibinfo{publisher}{{IEEE}}, \bibinfo{pages}{570--583}.
\newblock
\urldef\tempurl%
\url{https://doi.org/10.1109/ISCA45697.2020.00054}
\showDOI{\tempurl}


\bibitem[Fu et~al\mbox{.}(2020)]%
        {FKRS_RC2020_Dependent_ProtoQuipper}
\bibfield{author}{\bibinfo{person}{Peng Fu}, \bibinfo{person}{Kohei Kishida}, \bibinfo{person}{Neil~J. Ross}, {and} \bibinfo{person}{Peter Selinger}.} \bibinfo{year}{2020}\natexlab{}.
\newblock \showarticletitle{A Tutorial Introduction to Quantum Circuit Programming in Dependently Typed Proto-Quipper}. In \bibinfo{booktitle}{\emph{Reversible Computation - 12th International Conference, {RC} 2020, Oslo, Norway, July 9-10, 2020, Proceedings}} \emph{(\bibinfo{series}{Lecture Notes in Computer Science}, Vol.~\bibinfo{volume}{12227})}, \bibfield{editor}{\bibinfo{person}{Ivan Lanese} {and} \bibinfo{person}{Mariusz Rawski}} (Eds.). \bibinfo{publisher}{Springer}, \bibinfo{pages}{153--168}.
\newblock
\urldef\tempurl%
\url{https://doi.org/10.1007/978-3-030-52482-1\_9}
\showDOI{\tempurl}


\bibitem[Gidney(2019)]%
        {Gidney_2019_Spooky_blog}
\bibfield{author}{\bibinfo{person}{Craig Gidney}.} \bibinfo{year}{2019}\natexlab{}.
\newblock \bibinfo{title}{Spooky pebble games and irreversible uncomputation}.
\newblock
\newblock
\urldef\tempurl%
\url{https://algassert.com/post/1905}
\showURL{%
\tempurl}


\bibitem[Grover(1996)]%
        {grover:search}
\bibfield{author}{\bibinfo{person}{Lov~K. Grover}.} \bibinfo{year}{1996}\natexlab{}.
\newblock \showarticletitle{A fast quantum mechanical algorithm for database search}. In \bibinfo{booktitle}{\emph{Proceedings of the Twenty-Eighth Annual ACM Symposium on Theory of Computing}} (Philadelphia, Pennsylvania, USA) \emph{(\bibinfo{series}{STOC '96})}. \bibinfo{publisher}{Association for Computing Machinery}, \bibinfo{address}{New York, NY, USA}, \bibinfo{pages}{212--219}.
\newblock
\showISBNx{0897917855}
\urldef\tempurl%
\url{https://doi.org/10.1145/237814.237866}
\showDOI{\tempurl}


\bibitem[JavadiAbhari et~al\mbox{.}(2014)]%
        {ScaffCC_framework}
\bibfield{author}{\bibinfo{person}{Ali JavadiAbhari}, \bibinfo{person}{Shruti Patil}, \bibinfo{person}{Daniel Kudrow}, \bibinfo{person}{Jeff Heckey}, \bibinfo{person}{Alexey Lvov}, \bibinfo{person}{Frederic~T. Chong}, {and} \bibinfo{person}{Margaret Martonosi}.} \bibinfo{year}{2014}\natexlab{}.
\newblock \showarticletitle{ScaffCC: a framework for compilation and analysis of quantum computing programs}. In \bibinfo{booktitle}{\emph{Computing Frontiers Conference, CF'14, Cagliari, Italy - May 20 - 22, 2014}}, \bibfield{editor}{\bibinfo{person}{Pedro Trancoso}, \bibinfo{person}{Diana Franklin}, {and} \bibinfo{person}{Sally~A. McKee}} (Eds.). \bibinfo{publisher}{{ACM}}, \bibinfo{pages}{1:1--1:10}.
\newblock
\urldef\tempurl%
\url{https://doi.org/10.1145/2597917.2597939}
\showDOI{\tempurl}


\bibitem[JavadiAbhari et~al\mbox{.}(2015)]%
        {ScaffCC_scalable}
\bibfield{author}{\bibinfo{person}{Ali JavadiAbhari}, \bibinfo{person}{Shruti Patil}, \bibinfo{person}{Daniel Kudrow}, \bibinfo{person}{Jeff Heckey}, \bibinfo{person}{Alexey Lvov}, \bibinfo{person}{Frederic~T. Chong}, {and} \bibinfo{person}{Margaret Martonosi}.} \bibinfo{year}{2015}\natexlab{}.
\newblock \showarticletitle{ScaffCC: Scalable compilation and analysis of quantum programs}.
\newblock \bibinfo{journal}{\emph{Parallel Comput.}}  \bibinfo{volume}{45} (\bibinfo{year}{2015}), \bibinfo{pages}{2--17}.
\newblock
\urldef\tempurl%
\url{https://doi.org/10.1016/j.parco.2014.12.001}
\showDOI{\tempurl}


\bibitem[Kornerup et~al\mbox{.}(2024)]%
        {Kornerup_2024_Spooky_tight_bound}
\bibfield{author}{\bibinfo{person}{Niels Kornerup}, \bibinfo{person}{Jonathan Sadun}, {and} \bibinfo{person}{David Soloveichik}.} \bibinfo{year}{2024}\natexlab{}.
\newblock \bibinfo{title}{Tight Bounds on the Spooky Pebble Game: Recycling Qubits with Measurements}.
\newblock
\newblock
\showeprint[arxiv]{2110.08973}~[quant-ph]


\bibitem[Kumar~Pati and Braunstein(2000)]%
        {patibraunstein:nodeleting}
\bibfield{author}{\bibinfo{person}{Arun Kumar~Pati} {and} \bibinfo{person}{Samuel~L. Braunstein}.} \bibinfo{year}{2000}\natexlab{}.
\newblock \showarticletitle{Impossibility of deleting an unknown quantum state}.
\newblock \bibinfo{journal}{\emph{Nature}} \bibinfo{volume}{404}, \bibinfo{number}{6774} (\bibinfo{year}{2000}), \bibinfo{pages}{164--165}.
\newblock
\showISBNx{1476-4687}
\urldef\tempurl%
\url{https://doi.org/10.1038/404130b0}
\showDOI{\tempurl}


\bibitem[Lewis and Gehr(2021)]%
        {silq-github-issue}
\bibfield{author}{\bibinfo{person}{Marco Lewis} {and} \bibinfo{person}{Timon Gehr}.} \bibinfo{year}{2021}\natexlab{}.
\newblock \bibinfo{title}{Silq Issue 28: Quantum redefinition through quantum control behaviour inconsistent depending on scope}.
\newblock \bibinfo{howpublished}{\url{https://github.com/eth-sri/silq/issues/28}}.
\newblock
\urldef\tempurl%
\url{https://github.com/eth-sri/silq/issues/28}
\showURL{%
\tempurl}
\newblock
\shownote{GitHub Issue}.


\bibitem[Matsakis and Klock(2014)]%
        {rust-paper}
\bibfield{author}{\bibinfo{person}{Nicholas~D. Matsakis} {and} \bibinfo{person}{Felix~S. Klock}.} \bibinfo{year}{2014}\natexlab{}.
\newblock \showarticletitle{The rust language}. In \bibinfo{booktitle}{\emph{Proceedings of the 2014 ACM SIGAda Annual Conference on High Integrity Language Technology}} (Portland, Oregon, USA) \emph{(\bibinfo{series}{HILT '14})}. \bibinfo{publisher}{Association for Computing Machinery}, \bibinfo{address}{New York, NY, USA}, \bibinfo{numpages}{2}~pages.
\newblock
\showISBNx{9781450332170}
\urldef\tempurl%
\url{https://doi.org/10.1145/2663171.2663188}
\showDOI{\tempurl}


\bibitem[Matsushita et~al\mbox{.}(2020)]%
        {rusthorn}
\bibfield{author}{\bibinfo{person}{Yusuke Matsushita}, \bibinfo{person}{Takeshi Tsukada}, {and} \bibinfo{person}{Naoki Kobayashi}.} \bibinfo{year}{2020}\natexlab{}.
\newblock \showarticletitle{RustHorn: CHC-Based Verification for Rust Programs}. In \bibinfo{booktitle}{\emph{Programming Languages and Systems - 29th European Symposium on Programming, {ESOP} 2020, Held as Part of the European Joint Conferences on Theory and Practice of Software, {ETAPS} 2020, Dublin, Ireland, April 25-30, 2020, Proceedings}} \emph{(\bibinfo{series}{Lecture Notes in Computer Science}, Vol.~\bibinfo{volume}{12075})}, \bibfield{editor}{\bibinfo{person}{Peter M{\"{u}}ller}} (Ed.). \bibinfo{publisher}{Springer}, \bibinfo{pages}{484--514}.
\newblock
\urldef\tempurl%
\url{https://doi.org/10.1007/978-3-030-44914-8\_18}
\showDOI{\tempurl}


\bibitem[Meuli et~al\mbox{.}(2019)]%
        {MeuliSRBM_2019_RevPebble}
\bibfield{author}{\bibinfo{person}{Giulia Meuli}, \bibinfo{person}{Mathias Soeken}, \bibinfo{person}{Martin Roetteler}, \bibinfo{person}{Nikolaj~S. Bj{\o}rner}, {and} \bibinfo{person}{Giovanni~De Micheli}.} \bibinfo{year}{2019}\natexlab{}.
\newblock \showarticletitle{Reversible Pebbling Game for Quantum Memory Management}. In \bibinfo{booktitle}{\emph{Design, Automation {\&} Test in Europe Conference {\&} Exhibition, {DATE} 2019, Florence, Italy, March 25-29, 2019}}, \bibfield{editor}{\bibinfo{person}{J{\"{u}}rgen Teich} {and} \bibinfo{person}{Franco Fummi}} (Eds.). \bibinfo{publisher}{{IEEE}}, \bibinfo{pages}{288--291}.
\newblock
\urldef\tempurl%
\url{https://doi.org/10.23919/DATE.2019.8715092}
\showDOI{\tempurl}


\bibitem[Nielsen and Chuang(2010)]%
        {nielsenchuang:quantum}
\bibfield{author}{\bibinfo{person}{Michael~A. Nielsen} {and} \bibinfo{person}{Isaac~L. Chuang}.} \bibinfo{year}{2010}\natexlab{}.
\newblock \bibinfo{booktitle}{\emph{Quantum Computation and Quantum Information: 10th Anniversary Edition}}.
\newblock \bibinfo{publisher}{Cambridge University Press}.
\newblock
\showISBNx{978-1-10-700217-3}


\bibitem[Paradis et~al\mbox{.}(2021)]%
        {Unqomp_2021}
\bibfield{author}{\bibinfo{person}{Anouk Paradis}, \bibinfo{person}{Benjamin Bichsel}, \bibinfo{person}{Samuel Steffen}, {and} \bibinfo{person}{Martin~T. Vechev}.} \bibinfo{year}{2021}\natexlab{}.
\newblock \showarticletitle{Unqomp: synthesizing uncomputation in Quantum circuits}. In \bibinfo{booktitle}{\emph{{PLDI} '21: 42nd {ACM} {SIGPLAN} International Conference on Programming Language Design and Implementation, Virtual Event, Canada, June 20-25, 2021}}, \bibfield{editor}{\bibinfo{person}{Stephen~N. Freund} {and} \bibinfo{person}{Eran Yahav}} (Eds.). \bibinfo{publisher}{{ACM}}, \bibinfo{pages}{222--236}.
\newblock
\urldef\tempurl%
\url{https://doi.org/10.1145/3453483.3454040}
\showDOI{\tempurl}


\bibitem[Paradis et~al\mbox{.}(2024)]%
        {Reqomp_2024}
\bibfield{author}{\bibinfo{person}{Anouk Paradis}, \bibinfo{person}{Benjamin Bichsel}, {and} \bibinfo{person}{Martin~T. Vechev}.} \bibinfo{year}{2024}\natexlab{}.
\newblock \showarticletitle{Reqomp: Space-constrained Uncomputation for Quantum Circuits}.
\newblock \bibinfo{journal}{\emph{Quantum}}  \bibinfo{volume}{8} (\bibinfo{year}{2024}), \bibinfo{pages}{1258}.
\newblock
\urldef\tempurl%
\url{https://doi.org/10.22331/Q-2024-02-19-1258}
\showDOI{\tempurl}


\bibitem[Pratt(1992)]%
        {pratt:linear}
\bibfield{author}{\bibinfo{person}{Vaughan Pratt}.} \bibinfo{year}{1992}\natexlab{}.
\newblock \showarticletitle{Linear Logic for Generalized Quantum Mechanics}. In \bibinfo{booktitle}{\emph{Workshop on Physics and Computation}}. \bibinfo{pages}{166--180}.
\newblock
\urldef\tempurl%
\url{https://doi.org/10.1109/PHYCMP.1992.615518}
\showDOI{\tempurl}


\bibitem[Quist and Laarman(2023)]%
        {QuistLaarman_RC2023_Spooky_optimizing_quantum}
\bibfield{author}{\bibinfo{person}{Arend{-}Jan Quist} {and} \bibinfo{person}{Alfons Laarman}.} \bibinfo{year}{2023}\natexlab{}.
\newblock \showarticletitle{Optimizing Quantum Space Using Spooky Pebble Games}. In \bibinfo{booktitle}{\emph{Reversible Computation - 15th International Conference, {RC} 2023, Giessen, Germany, July 18-19, 2023, Proceedings}} \emph{(\bibinfo{series}{Lecture Notes in Computer Science}, Vol.~\bibinfo{volume}{13960})}, \bibfield{editor}{\bibinfo{person}{Martin Kutrib} {and} \bibinfo{person}{Uwe Meyer}} (Eds.). \bibinfo{publisher}{Springer}, \bibinfo{pages}{134--149}.
\newblock
\urldef\tempurl%
\url{https://doi.org/10.1007/978-3-031-38100-3\_10}
\showDOI{\tempurl}


\bibitem[Quist and Laarman(2024)]%
        {Quist_2024_Spooky_tradeoffs}
\bibfield{author}{\bibinfo{person}{Arend-Jan Quist} {and} \bibinfo{person}{Alfons Laarman}.} \bibinfo{year}{2024}\natexlab{}.
\newblock \bibinfo{title}{Trade-offs between classical and quantum space using spooky pebbling}.
\newblock
\newblock
\showeprint[arxiv]{2401.10579}~[quant-ph]


\bibitem[Rand et~al\mbox{.}(2018)]%
        {ReQWIRE_2018}
\bibfield{author}{\bibinfo{person}{Robert Rand}, \bibinfo{person}{Jennifer Paykin}, \bibinfo{person}{Dong{-}Ho Lee}, {and} \bibinfo{person}{Steve Zdancewic}.} \bibinfo{year}{2018}\natexlab{}.
\newblock \showarticletitle{ReQWIRE: Reasoning about Reversible Quantum Circuits}. In \bibinfo{booktitle}{\emph{Proceedings 15th International Conference on Quantum Physics and Logic, {QPL} 2018, Halifax, Canada, 3-7th June 2018}} \emph{(\bibinfo{series}{{EPTCS}}, Vol.~\bibinfo{volume}{287})}, \bibfield{editor}{\bibinfo{person}{Peter Selinger} {and} \bibinfo{person}{Giulio Chiribella}} (Eds.). \bibinfo{pages}{299--312}.
\newblock
\urldef\tempurl%
\url{https://doi.org/10.4204/EPTCS.287.17}
\showDOI{\tempurl}


\bibitem[{Rust Community}(2024a)]%
        {rust-reference}
\bibfield{author}{\bibinfo{person}{{Rust Community}}.} \bibinfo{year}{2024}\natexlab{a}.
\newblock \bibinfo{booktitle}{\emph{References and Borrowing -- The Rust Programming Language}}.
\newblock
\urldef\tempurl%
\url{https://doc.rust-lang.org/book/ch04-02-references-and-borrowing.html}
\showURL{%
\tempurl}


\bibitem[{Rust Community}(2024b)]%
        {rust-web}
\bibfield{author}{\bibinfo{person}{{Rust Community}}.} \bibinfo{year}{2024}\natexlab{b}.
\newblock \bibinfo{booktitle}{\emph{The Rust Programming Language}}.
\newblock
\urldef\tempurl%
\url{http://rust-lang.org/}
\showURL{%
\tempurl}


\bibitem[{Rust Community}(2024c)]%
        {rust-lifetime}
\bibfield{author}{\bibinfo{person}{{Rust Community}}.} \bibinfo{year}{2024}\natexlab{c}.
\newblock \bibinfo{booktitle}{\emph{Validating References with Lifetimes -- The Rust Programming Language}}.
\newblock
\urldef\tempurl%
\url{https://doc.rust-lang.org/book/ch10-03-lifetime-syntax.html}
\showURL{%
\tempurl}


\bibitem[{Rust Community}(2024d)]%
        {rust-ownership}
\bibfield{author}{\bibinfo{person}{{Rust Community}}.} \bibinfo{year}{2024}\natexlab{d}.
\newblock \bibinfo{booktitle}{\emph{What Is Ownership? -- The Rust Programming Language}}.
\newblock
\urldef\tempurl%
\url{https://doc.rust-lang.org/book/ch04-01-what-is-ownership.html}
\showURL{%
\tempurl}


\bibitem[Sabry et~al\mbox{.}(2018)]%
        {sabryvalironvizzotto:control}
\bibfield{author}{\bibinfo{person}{Amr Sabry}, \bibinfo{person}{Beno{\^{\i}}t Valiron}, {and} \bibinfo{person}{Juliana~Kaizer Vizzotto}.} \bibinfo{year}{2018}\natexlab{}.
\newblock \showarticletitle{From Symmetric Pattern-Matching to Quantum Control}. In \bibinfo{booktitle}{\emph{Foundations of Software Science and Computation Structures - 21st International Conference, {FOSSACS} 2018, Held as Part of the European Joint Conferences on Theory and Practice of Software, {ETAPS} 2018, Thessaloniki, Greece, April 14-20, 2018, Proceedings}} \emph{(\bibinfo{series}{Lecture Notes in Computer Science}, Vol.~\bibinfo{volume}{10803})}, \bibfield{editor}{\bibinfo{person}{Christel Baier} {and} \bibinfo{person}{Ugo~Dal Lago}} (Eds.). \bibinfo{publisher}{Springer}, \bibinfo{pages}{348--364}.
\newblock
\urldef\tempurl%
\url{https://doi.org/10.1007/978-3-319-89366-2\_19}
\showDOI{\tempurl}


\bibitem[Seidel et~al\mbox{.}(2023)]%
        {Qrisp_2023}
\bibfield{author}{\bibinfo{person}{Raphael Seidel}, \bibinfo{person}{Nikolay Tcholtchev}, \bibinfo{person}{Sebastian Bock}, {and} \bibinfo{person}{Manfred Hauswirth}.} \bibinfo{year}{2023}\natexlab{}.
\newblock \showarticletitle{Uncomputation in the Qrisp High-Level Quantum Programming Framework}. In \bibinfo{booktitle}{\emph{Reversible Computation - 15th International Conference, {RC} 2023, Giessen, Germany, July 18-19, 2023, Proceedings}} \emph{(\bibinfo{series}{Lecture Notes in Computer Science}, Vol.~\bibinfo{volume}{13960})}, \bibfield{editor}{\bibinfo{person}{Martin Kutrib} {and} \bibinfo{person}{Uwe Meyer}} (Eds.). \bibinfo{publisher}{Springer}, \bibinfo{pages}{150--165}.
\newblock
\urldef\tempurl%
\url{https://doi.org/10.1007/978-3-031-38100-3\_11}
\showDOI{\tempurl}


\bibitem[Selinger and Valiron(2006)]%
        {selingervaliron:lambda}
\bibfield{author}{\bibinfo{person}{Peter Selinger} {and} \bibinfo{person}{Beno{\^{\i}}t Valiron}.} \bibinfo{year}{2006}\natexlab{}.
\newblock \showarticletitle{A lambda calculus for quantum computation with classical control}.
\newblock \bibinfo{journal}{\emph{Math. Struct. Comput. Sci.}} \bibinfo{volume}{16}, \bibinfo{number}{3} (\bibinfo{year}{2006}), \bibinfo{pages}{527--552}.
\newblock
\urldef\tempurl%
\url{https://doi.org/10.1017/S0960129506005238}
\showDOI{\tempurl}


\bibitem[Venev et~al\mbox{.}(2024)]%
        {venevetal:modularsynthesis}
\bibfield{author}{\bibinfo{person}{Hristo Venev}, \bibinfo{person}{Timon Gehr}, \bibinfo{person}{Dimitar Dimitrov}, {and} \bibinfo{person}{Martin Vechev}.} \bibinfo{year}{2024}\natexlab{}.
\newblock \showarticletitle{Modular Synthesis of Efficient Quantum Uncomputation}.
\newblock \bibinfo{journal}{\emph{Proceedings of the ACM on Programming Languages}} \bibinfo{volume}{8}, \bibinfo{number}{OOPSLA2}, Article \bibinfo{articleno}{345} (\bibinfo{year}{2024}).
\newblock
\urldef\tempurl%
\url{https://doi.org/10.1145/3689785}
\showDOI{\tempurl}


\bibitem[Wadler(1990)]%
        {wadler:linear}
\bibfield{author}{\bibinfo{person}{Philip Wadler}.} \bibinfo{year}{1990}\natexlab{}.
\newblock \showarticletitle{Linear Types can Change the World!}. In \bibinfo{booktitle}{\emph{Programming concepts and methods: Proceedings of the {IFIP} Working Group 2.2, 2.3 Working Conference on Programming Concepts and Methods, Sea of Galilee, Israel, 2-5 April, 1990}}, \bibfield{editor}{\bibinfo{person}{Manfred Broy} {and} \bibinfo{person}{Cliff~B. Jones}} (Eds.). \bibinfo{publisher}{North-Holland}, \bibinfo{pages}{561}.
\newblock


\bibitem[Wootters and Zurek(1982)]%
        {wootterszurek:nocloning}
\bibfield{author}{\bibinfo{person}{William~K. Wootters} {and} \bibinfo{person}{Wojciech~H. Zurek}.} \bibinfo{year}{1982}\natexlab{}.
\newblock \showarticletitle{A single quantum cannot be cloned}.
\newblock \bibinfo{journal}{\emph{Nature}} \bibinfo{volume}{299}, \bibinfo{number}{5886} (\bibinfo{year}{1982}), \bibinfo{pages}{802--803}.
\newblock
\showISBNx{1476-4687}
\urldef\tempurl%
\url{https://doi.org/10.1038/299802a0}
\showDOI{\tempurl}


\bibitem[Yanofsky and Mannucci(2008)]%
        {yanofskymannucci:quantumcomputing}
\bibfield{author}{\bibinfo{person}{Noson~S. Yanofsky} {and} \bibinfo{person}{Mirco Mannucci}.} \bibinfo{year}{2008}\natexlab{}.
\newblock \bibinfo{booktitle}{\emph{Quantum Computing for Computer Scientists}}.
\newblock \bibinfo{publisher}{Cambridge University Press}.
\newblock
\showISBNx{978-0-521-87996-5}


\bibitem[Ying(2016)]%
        {ying}
\bibfield{author}{\bibinfo{person}{Mingsheng Ying}.} \bibinfo{year}{2016}\natexlab{}.
\newblock \bibinfo{booktitle}{\emph{Foundations of Quantum Programming} (\bibinfo{edition}{1st} ed.)}.
\newblock \bibinfo{publisher}{Morgan Kaufmann Publishers Inc.}, \bibinfo{address}{San Francisco, CA, USA}.
\newblock
\showISBNx{0128023066}
\urldef\tempurl%
\url{https://doi.org/10.1016/C2014-0-02660-3}
\showDOI{\tempurl}


\bibitem[Yuan et~al\mbox{.}(2022)]%
        {twist_2022}
\bibfield{author}{\bibinfo{person}{Charles Yuan}, \bibinfo{person}{Christopher McNally}, {and} \bibinfo{person}{Michael Carbin}.} \bibinfo{year}{2022}\natexlab{}.
\newblock \showarticletitle{Twist: sound reasoning for purity and entanglement in Quantum programs}.
\newblock \bibinfo{journal}{\emph{Proc. {ACM} Program. Lang.}} \bibinfo{volume}{6}, \bibinfo{number}{{POPL}} (\bibinfo{year}{2022}), \bibinfo{pages}{1--32}.
\newblock
\urldef\tempurl%
\url{https://doi.org/10.1145/3498691}
\showDOI{\tempurl}


\end{thebibliography}
\end{document}